\documentclass[11pt]{article}

\usepackage{subcaption}
\usepackage{tikz}
\usetikzlibrary{decorations.pathmorphing}
\usetikzlibrary{decorations.markings}
\usetikzlibrary{shapes.gates.logic.US,trees,positioning,arrows}

\tikzset{
	arn/.style = {circle, white, draw=black, fill=gray!30, inner sep = 10.5},
	arn_t/.style = {circle, white, draw=black, very thick, fill=gray!30, inner sep = 11.0},
	arn_l/.style = {circle, white, draw=black, very thick, fill=black, inner sep = 2},
	photon/.style={draw=black, very thick, dashed},
	electron/.style={draw=black, very thick},
	tr/.style={buffer gate US,thick,draw,fill=gray!60,rotate=90,	anchor=east,minimum width=2.25cm},
	br/.style={buffer gate US,thick,draw,fill=gray!60,rotate=90,	anchor=east,minimum width=4.5cm},
	brr/.style={buffer gate US,draw,fill=gray!60,rotate=90,	anchor=east,minimum width=4.5cm, opacity = 0.6},
	trr/.style={buffer gate US,thick,draw,fill=gray!60,rotate=90,	anchor=east,minimum width=2.25cm, opacity = 0.6},
	trrr/.style={buffer gate US,draw,fill=white!60,rotate=90,	anchor=east,minimum width=2.25cm, opacity = 0.5}
}

\usepackage{bbm}
\usepackage[margin=1in,letterpaper]{geometry}
\usepackage[T1]{fontenc}
\usepackage{lmodern}
\usepackage{amsmath}
\usepackage{amssymb}
\usepackage{amsthm}
\usepackage{thm-restate}
\usepackage{graphicx}
\usepackage{xcolor}
\usepackage{algorithm}
\usepackage{algpseudocode}
\usepackage{graphicx}
\usepackage{subcaption}
\usepackage{hyperref}
\usepackage{cleveref}
\usepackage{todonotes}
\newcommand{\algoname}[1]{\textnormal{\textsc{#1}}}

\newcommand{\Varcond}[2]{\text{Var}\!\left[#1 \; \middle| \; #2 \right]}
\newcommand{\Rbb}{\mathbb{R}}

\DeclareMathOperator*{\Nnz}{nnz}

\newcommand{\Var}{\mathrm{Var}}
\newcommand{\tr}{\mathbf{tr}}
\newcommand{\Prp}[1]{\Pr\!\left[#1\right]}
\newcommand{\Prpcond}[2]{\Pr\!\left[#1 \; \middle| \; #2 \right]}
\newcommand{\Ep}[1]{\mathbb{E}\!\left[#1\right]}
\newcommand{\Epcond}[2]{\mathbb{E}\!\left[#1 \; \middle| \; #2 \right]}
\newcommand{\E}{{\mathcal E}}
\newcommand{\sbase}{S_\mathrm{base}}
\newcommand{\tbase}{T_\mathrm{base}}

\newcommand{\abs}[1]{\left\lvert #1\right\rvert}

\newcommand{\numberthis}{\addtocounter{equation}{1}\tag{\theequation}}

\newcommand{\eps}{\varepsilon}

\newcommand{\PNorm}[2]{\left\| #1 \right\|_{L^{#2}}}

\theoremstyle{plain}
\newtheorem{thm}{\protect\theoremname}
\theoremstyle{plain}
\newtheorem{claim}[thm]{\protect\claimname}
\theoremstyle{plain}

\theoremstyle{plain}
\newtheorem{lem}[thm]{\protect\lemmaname}
\theoremstyle{plain}
\newtheorem{cor}[thm]{\protect\corollaryname}
\theoremstyle{definition}
\newtheorem{defn}{Definition}
\theoremstyle{definition}

\theoremstyle{definition}
\newtheorem{rem}{\protect\remarkname}
\theoremstyle{plain}

\makeatother

\providecommand{\claimname}{Claim}
\providecommand{\lemmaname}{Lemma}
\providecommand{\propositionname}{Proposition}
\providecommand{\theoremname}{Theorem}
\providecommand{\corollaryname}{Corollary}

\providecommand{\assumptionname}{Assumption}
\providecommand{\remarkname}{Remark}

\global\long\def\RR{\mathbb{R}}
\global\long\def\CC{\mathbb{C}}
\global\long\def\ZZ{\mathbb{Z}}

\global\long\def\R{{\cal R}}

\global\long\def\nnz#1{\mathrm{nnz}\left(#1\right)}

\global\long\def\sign#1{\mathrm{sign}\left(#1\right)}
\global\long\def\poly#1{\mathrm{poly}\left(#1\right)}

\global\def\eqdef{\equiv}

\newcommand{\wh}{\widehat}
\newcommand{\ex}{\mathbb{E}}

\newcommand{\norm}[1]{\|#1\|}

\title{Oblivious Sketching of High-Degree Polynomial Kernels\thanks{This paper is a merged version of the work of Ahle and Knudsen~\cite{ahle2019almost} and Kapralov, Pagh, Velingker, Woodruff and Zandieh~\cite{kapralov2019oblivious}.}}

\author{
Thomas D.~Ahle\\ \small ITU and BARC\\ \small \texttt{thdy@itu.dk}
\and
Michael Kapralov\\ \small EPFL\\ \small \texttt{michael.kapralov@epfl.ch}
\and
Jakob B.~T.~Knudsen\\ \small U.~Copenhagen and BARC\\ \small \texttt{jakn@di.ku.dk}
\and
Rasmus Pagh\\ \small ITU and BARC\\ \small \texttt{pagh@itu.dk}
\and
Ameya Velingker\\ \small Google Research\\ \small \texttt{ameyav@google.com}
\and
David P.~Woodruff\\ \small CMU\\ \small \texttt{dwoodruf@cs.cmu.edu}
\and
Amir Zandieh\\ \small EPFL\\ \small \texttt{amir.zandieh@epfl.ch}
}

{\makeatletter
	\gdef\xxxmark{%
		\expandafter\ifx\csname @mpargs\endcsname\relax 
		\expandafter\ifx\csname @captype\endcsname\relax 
		\marginpar{xxx}
		\else
		xxx 
		\fi
		\else
		xxx 
		\fi}
	\gdef\xxx{\@ifnextchar[\xxx@lab\xxx@nolab}
	\long\gdef\xxx@lab[#1]#2{{\bf [\xxxmark #2 ---{\sc #1}]}}
	\long\gdef\xxx@nolab#1{{\bf [\xxxmark #1]}}
}

\begin{document}
\maketitle
\begin{abstract}
	Kernel methods are fundamental tools in machine learning that allow detection of non-linear dependencies between data without explicitly constructing feature vectors in high dimensional spaces. A major disadvantage of kernel methods is their poor scalability: primitives such as kernel PCA or kernel ridge regression generally take prohibitively large quadratic space and (at least) quadratic time, as kernel matrices are usually dense.  Some methods for speeding up kernel linear algebra are known, but they all invariably take time exponential in either the dimension of the input point set (e.g., fast multipole methods suffer from the \emph{curse of dimensionality}) or in the degree of the kernel function.
	
	\emph{Oblivious sketching} has emerged as a powerful approach to speeding up numerical linear algebra over the past decade, but our understanding of oblivious sketching solutions for kernel matrices has remained quite limited, suffering from the aforementioned exponential dependence on input parameters. Our main contribution is a general method for applying sketching solutions developed in numerical linear algebra over the past decade to a tensoring of data points without forming the tensoring explicitly. This leads to the first oblivious sketch for the polynomial kernel with a target dimension that is only polynomially dependent on the degree of the kernel function, as well as the first oblivious sketch for the Gaussian kernel on bounded datasets that does not suffer from an exponential dependence on the dimensionality of input data points.
\end{abstract}

\newpage
\tableofcontents
\newpage

\section{Introduction}
Data dimensionality reduction, or \emph{sketching}, is a common technique for quickly reducing the size of a large-scale optimization problem while approximately preserving the solution space, thus allowing one to instead solve a much smaller optimization problem, typically in a smaller amount of time. This technique has led to near-optimal algorithms for a number of fundamental problems in numerical linear algebra and machine learning, such as least squares regression, low rank approximation, canonical correlation analysis, and robust variants of these problems. In a typical instance of such a problem, one is given a large matrix $X \in \mathbb{R}^{d \times n}$ as input, and one wishes to choose a random map $\Pi$ from a certain family of random maps and replace $X$ with $\Pi X$. As $\Pi$ typically has many fewer rows than columns, $\Pi X$ compresses the original matrix $X$, which allows one to perform the original optimization problem on the much smaller matrix $\Pi X$. For a survey of such techniques, we refer the reader to the survey by Woodruff \cite{w14}. 

A key challenge in this area is to extend sketching techniques to kernel-variants of the above linear algebra problems. Suppose each column of $X$ corresponds to an example while each of the $d$ rows corresponds to a feature. Then these algorithms require an explicit representation of $X$ to be made available to the algorithm. This is unsatisfactory in many machine learning applications, since typically the actual learning is performed in a much higher (possibly infinite) dimensional feature space, by first mapping each column of $X$ to a much higher dimensional space. Fortunately, due to the kernel trick, one need not ever perform this mapping explicitly; indeed, if the optimization problem at hand only depends on inner product information between the input points, then the kernel trick allows one to quickly compute the inner products of the high dimensional transformations of the input points, without ever explicitly computing the transformation itself.
However, evaluating the kernel function easily becomes a bottleneck in algorithms that rely on the kernel trick because it typically takes $O(d)$ time to evaluate the kernel function for $d$ dimensional datasets. There are a number of recent works which try to improve the running times of kernel methods; we refer
the reader to the recent work of \cite{mm17} and the references therein.
A natural question is whether it is possible to instead apply sketching techniques on the high-dimensional feature space without ever computing the high-dimensional mapping.

For the important case of \emph{polynomial kernel}, such sketching techniques are known to be possible\footnote{The lifting function corresponding to the polynomial kernel maps $x \in \mathbb{R}^d$ to $\phi(x) \in \mathbb{R}^{d^p}$, where $\phi(x)_{i_1, i_2, \ldots, i_p} = x_{i_1} x_{i_2} \cdots x_{i_p}$, for $i_1, i_2, \ldots, i_p \in \{1, 2, \ldots, d\}$}. This was originally shown by Pham and Pagh in the context of kernel support vector machines \cite{pp13}, using the {\sf TensorSketch} technique for compressed matrix multiplication due to Pagh \cite{p13}. This was later extended
in \cite{avron2014subspace} to a wide array of kernel problems in linear algebra, including principal component analysis, principal component regression, and canonical correlation analysis.

The running times of the algorithms above, while nearly linear in the number of non-zero entries of the input matrix $X$, depend {\it exponentially} on the degree $q$ of the polynomial kernel. For example, suppose one wishes to do low rank approximation on $A$, the matrix obtained by replacing each column of $X$ with its kernel-transformed version. One would like to express $A \approx U V$, where $U \in \mathbb{R}^{d^p \times k}$ and $V \in \mathbb{R}^{k \times n}$. Writing down $U$ explicitly is problematic, since the columns belong to the much higher $d^p$-dimensional space. Instead, one can express $UV$ implicitly via column subset selection, by expressing it as a $A Z Z^\top$ and then outputting $Z$. Here $Z$ is an $n \times k$ matrix. In \cite{avron2014subspace}, an algorithm running in $\nnz{X} + (n+d)\poly{3^p,k,1/\epsilon}$ time was given for outputting such $Z$ with the guarantee that $\|A - A Z Z^\top\|_F^2 \leq (1+\epsilon)\|A - A_k\|_F^2$ with constant probability, where $A_k$ is the best rank-$k$ approximation to $A$. Algorithms with similar running times were proposed for principal component regression and canonical correlation analysis. The main message here is that all analyses of all existing sketches require the sketch $\Pi$ to have at least $3^p$ rows in order to guarantee their correctness. Moreover, the existing sketches work with constant probability only and no high probability result was known for the polynomial kernel.

The main drawback with previous work on applying dimensionality reduction for the polynomial kernel is the exponential dependence on $p$ in the sketching dimension and consequently in the running time. Ideally, one would like a polynomial dependence. This is especially useful for the application of approximating the Gaussian kernel by a sum of polynomial kernels of various degrees, for which large values of $p$, e.g., $p = \poly{\log n}$ are used \cite{c11}. This raises the main question of our work:

\begin{center}
{\it Is it possible to desing a data oblivious sketch with a sketching dimension (and, hence, running time) that is not exponential in $p$ for the above applications in the context of the polynomial kernel?}
\end{center}

While we answer the above question, we also study it in a more
general context, namely, that of regularization. In many machine learning
problems, it is crucial to regularize so as to prevent overfitting or
ill-posed problems. Sketching and related sampling-based techniques
have also been extensively applied in this
setting. For a small sample of such work see \cite{rr07,am15,pw15,mm17,acw17,acw17b,AKMMVZ17,a19}.
As an example application, in ordinary least squares regression
one is given a $d \times n$ matrix $A$, and a $d \times 1$ vector $b$,
and one seeks to find a $y \in \mathbb{R}^n$ so as to minimize
$\|Ay-b\|_2^2$. In ridge regression, we instead seek a $y$ so as to minimize
$\|Ay-b\|_2^2 + \lambda \|y\|_2^2$, for a parameter $\lambda > 0$. Intuitively,
if $\lambda$ is much larger than the operator norm $\|A\|_2$ of $A$, then
a good solution is obtained simply by setting $y = 0^d$. On the other hand,
if $\lambda = 0$, the problem just becomes an ordinary least squares regression. In
general, the \emph{statistical dimension} (or \emph{effective degrees of freedom}), $s_\lambda$, captures
this tradeoff, and is defined as
$\sum_{i=1}^d \frac{\lambda_i(A^\top A)}{\lambda_i(A^\top A) + \lambda}$,
where $\lambda_i(A^\top A)$ is the $i$-th eigenvalue of $A^\top A$. Note that the statistical
dimension is always at most $\min(n,d)$, but in fact can be much smaller. A key
example of its power is that for ridge regression, it is known \cite{acw17}
that if one chooses a random Gaussian matrix $\Pi$ with $O(s_{\lambda}/\epsilon)$
rows, and if $y$ is the minimizer to $\|\Pi A y-\Pi b\|_2^2 + \lambda \|y\|_2^2$,
then $\|Ay-b\|_2^2 + \lambda \|y\|_2^2 \leq (1+\epsilon) \min_{y'} (\|Ay'-b\|_2^2 + \lambda \|y'\|_2^2)$. Note that for ordinary regression $(\lambda = 0)$ one would need that $\Pi$ has $\Omega(\text{rank}(A)/\epsilon)$ rows \cite{cw09}. Another drawback of existing sketches for the polynomial kernel is that their running time and target dimension depend at least quadratically on $s_\lambda$ and no result is known with linear dependence on $s_\lambda$, which would be optimal.
We also ask if the exponential dependence on $p$ is avoidable in the \emph{regularized} setting:

\begin{center}
{\it Is it possible to obtain sketching dimension bounds and running
times that are not exponential in $p$ in the context of regularization? Moreover, is it possible to obtain a running time that depends only linearly on $s_{\lambda}$?}
\end{center}

\subsection{Our Contributions}
In this paper, we answer the above questions in the affirmative. In other words, for each of the aforementioned applications, our algorithm depends only \emph{polynomially} on $p$. We state these applications as corollaries of our main results, which concern approximate matrix product and subspace embeddings. In particular, we devise a new distribution on oblivious linear maps $\Pi \in \RR^{m\times d^p}$ (i.e., a randomized family of maps that does not depend on the dataset $X$), so that for any fixed $X \in \RR^{d\times n}$, it satisfies the approximate matrix product and subspace embedding properties. These are the key properties needed for kernel low rank approximation. We remark that our \emph{data oblivious sketching} is greatly advantageous to data dependent methods because it results in a one-round distributed protocol for kernel low rank approximation \cite{kannan2014principal}.

We show that our oblivious linear map $\Pi \in \RR^{m\times d^p}$ has the following key properties:

\paragraph{Oblivious Subspace Embeddings (OSEs).}
Given $\eps > 0$ and an $n$-dimensional subspace $E \subseteq \RR^{d}$, we say that $\Pi \in \RR^{m \times d}$ is an $\eps$-subspace embedding for $E$ if $(1-\eps)\|x\|_2 \le \|\Pi x\|_2 \le (1+\eps)\|x\|_2$ for all $x \in E$. In this paper we focus on Oblivious Subspace Embeddings in the regularized setting. In order to define a (regularized) Oblivious Subspace Embedding, we need to introduce the notion of \emph{statistical dimension}, which is defined as follows:

\begin{defn}[Statistical Dimension]\label{stat-dim-definition}
Given $\lambda\ge0$, for every positive semidefinite matrix $K \in \RR^{n\times n}$, we define the $\lambda$-statistical dimension of $K$ to be
$$s_\lambda(K) := \tr(K (K + \lambda I_n)^{-1}).$$
\end{defn}

Now, we can define the notion of an oblivious subspace embedding (OSE):
\begin{defn}[Oblivious Subspace Embedding (OSE)] \label{OSE-definition}
Given $\eps, \delta, \mu > 0$ and integers $d,n \ge 1$, an \emph{$(\eps, \delta, \mu, d, n)$-Oblivious Subspace Embedding (OSE)} is a distribution $\mathcal{D}$ over $m \times d$ matrices (for arbitrary $m$) such that for every $\lambda\ge0$, every $A \in \RR^{d \times n}$ with $\lambda$-statistical dimension $s_\lambda(A^\top A) \le \mu$, the following holds,\footnote{For symmetric matrices $K$ and $K'$, the spectral inequality relation $K \preceq K'$ holds if and only if $x^\top K x \leq x^\top K' x$ for all vectors $x$}
\begin{equation}\label{eq:aspectral}
\Pr_{\Pi \sim \mathcal{D}}\left[(1-\epsilon)(A^\top A + \lambda I_n) \preceq (\Pi A)^\top \Pi A + \lambda I_n \preceq (1+\epsilon) (A^\top A + \lambda I_n)\right] \ge 1-\delta.
\end{equation}
\end{defn}

The goal is to have the target dimension $m$ small so that $\Pi$ provides dimensionality reduction. If we consider the non-oblivious setting where we allow the sketch matrix $\Pi$ to depend on $A$, then by leverage score sampling we can achieve a target dimension of $m \approx s_\lambda(A^\top A)$, which is essentially optimal \cite{avron2018universal}. But as we discussed the importance of oblivious embeddings, the ultimate goal is to get an oblivious subspace embedding with target dimension of $m \approx s_\lambda(A^\top A)$.

\paragraph{Approximate Matrix Product.} We formally define this property in the following definition.

\begin{defn}[Approximate Matrix Product] \label{aprox-matrix-prod-definition}
Given $\eps, \delta >0$, we say that a distribution $\mathcal{D}$ over $m \times d$ matrices has the \emph{$(\eps, \delta)$-approximate matrix product} property if for every $C, D \in \RR^{d \times n}$,
\[
\Pr_{\Pi\sim \mathcal{D}}\left[\|C^\top \Pi^\top \Pi D - C^\top D\|_F \leq \eps \|C\|_F \|D\|_F\right] \ge 1-\delta.
\]
\end{defn}

Our main theorems, which provide the aforementioned guarantees, are as follows,\footnote{Throughout this paper, the notations $ \widetilde{O}, \widetilde{\Omega}, \widetilde{\Theta}$ suppress $\poly{\log(nd/\eps)}$ factors.}

\begin{restatable}{thm}{constprobthrm}\label{const-prob-thrm}
For every positive integers $n, p, d$, every $\eps, s_\lambda>0$, there exists a distribution on linear sketches $\Pi^p \in \RR^{m\times d^p}$ such that: {\bf (1)} If $m = \Omega\left(p s_\lambda^2 \epsilon^{-2} \right)$, then $\Pi^p$ is an $(\eps, 1/10, s_\lambda, d^p, n)$-oblivious subspace embedding as in Definition \ref{OSE-definition}.
{\bf (2)} If $m = \Omega\left( p\eps^{-2} \right)$, then $\Pi^p$ has the $(\eps, 1/10)$-approximate matrix product property as in Definition \ref{aprox-matrix-prod-definition}.

Moreover, for any  $X \in \RR^{d \times n}$, if $A \in \RR^{d^p \times n}$ is the matrix whose columns are obtained by the $p$-fold self-tensoring of each column of $X$ then the matrix $\Pi^p A$ can be computed using Algorithm \ref{alg:main} in time $\widetilde{O}\left( p n m + p \Nnz(X) \right)$.
\end{restatable}

\begin{restatable}{thm}{highprobmoment}\label{thm:high-prob-moment}
For every positive integers $n, p, d$, every $\eps, s_\lambda >0$, there exists a distribution on linear sketches $\Pi^p \in \RR^{m\times d^p}$ such that:
{\bf (1)} If $m = \widetilde{\Omega}\left( p  s_\lambda^2\epsilon^{-2} \right)$, then $\Pi^p$ is an $(\eps, 1/\poly{n}, s_\lambda, d^p, n)$-oblivious subspace embedding (Definition \ref{OSE-definition}).
{\bf (2)} If $m = \widetilde{\Omega}\left( p\eps^{-2} \right)$, then $\Pi^p$
has the $(\eps, 1/\poly{n})$-approximate matrix product property (Definition \ref{aprox-matrix-prod-definition}).

Moreover, in the setting of {\bf (1)}, for any  $X \in \RR^{d \times n}$, if $A \in \RR^{d^p \times n}$ is the matrix whose columns are obtained by a $p$-fold self-tensoring of each column of $X$, then the matrix $\Pi^p A$ can be computed using \Cref{alg:main} in time $\widetilde{O}\left( p n m + p^{3/2} s_\lambda\eps^{-1}  \Nnz(X) \right)$.
\end{restatable}

\begin{restatable}{thm}{highprobsketch}\label{high-prob-sketch}
For every positive integers $p, d, n$, every $\eps, s_\lambda >0$, there exists a distribution on linear sketches $\Pi^p \in \RR^{m\times d^p}$ which is an $(\eps, 1/\poly{n}, s_\lambda, d^p, n)$-oblivious subspace embedding as in Definition \ref{OSE-definition},
provided that the integer $m$ satisfies $m = \widetilde{\Omega}\left( p^4  s_\lambda  / \epsilon^2 \right)$.

Moreover, for any  $X \in \RR^{d \times n}$, if $A \in \RR^{d^p \times n}$ is the matrix whose columns are obtained by a $p$-fold self-tensoring of each column of $X$ then the matrix $\Pi^p A$ can be computed using Algorithm~\ref{alg:main} in time $\widetilde{O}\left( p n m + p^5\epsilon^{-2} \Nnz(X) \right)$.
\end{restatable}

We can immediately apply
these theorems to \emph{kernel ridge regression} with
respect to the polynomial kernel of degree $p$.
In this problem, we are given a regularization parameter $\lambda >0$, a $d \times n$ matrix $X$, and vector
$b \in \mathbb{R}^n$ and would like
to find a $y \in \mathbb{R}^n$ so as to minimize $\|A^\top Ay-b\|_2^2 + \lambda \|Ay\|_2^2$, where
$A \in \mathbb{R}^{d^p \times n}$
is the matrix obtained from $X$ by applying the self tensoring
of degree $p$ to each column.
To solve this problem via sketching,
we choose a random matrix $\Pi^p$ according to the theorems above and
compute $\Pi^p A$.
We then solve the sketched ridge regression problem which seeks to minimize
$\left\|(\Pi^p A)^\top \Pi^p A y- b \right\|_2^2 + \lambda \|\Pi^p A x\|_2^2$ over $y \in \RR^n$.
By the above theorems, solving the sketched kernel ridge regression problem gives a $(1 \pm \epsilon)$-approximation to the original problem. In particular, if we let $y^*$ denote the optimal solution of the original kernel ridge regression problem and $\widetilde{y}^*$ be the solution to the sketched problem, then, as shown in \cite{acw17b}, $\left\| y^* - \widetilde{y}^* \right\|_{A^\top A + \lambda I_n} \le \eps \left\| y^* \right\|_{A^\top A + \lambda I_n}$\footnote{For a Positive Definite matrix $K$, $\left\| \cdot \right\|_{K}$ denotes the norm induced by $K$, i.e. $\|x\|_K^2 = x^\top K x$.}.
If we apply Theorem \ref{const-prob-thrm}, then the number of rows
of $\Pi^p$ needed to ensure success with probability $9/10$
is $\Theta(p s_{\lambda}^2 \epsilon^{-2})$. The running time to compute
$\Pi^p A$ is $O( p^2 s_\lambda^2 \eps^{-2} n  + p \Nnz(X))$, after which a ridge regression problem can be solved
in $O(n s_{\lambda}^4/ \epsilon^4)$ time via an exact closed-form
solution for ridge regression. An alternative approach to obtaining a very high-accuracy approximation is to use the sketched kernel as a preconditioner to solve the original ridge regression problem via the preconditioned conjugate gradient method, which improves the dependence on $\eps$ to $\log(1/\eps)$ \cite{acw17b}.
To obtain a higher probability of success, we can instead apply
Theorem \ref{high-prob-sketch}, which would allow us to compute the sketched matrix
$\Pi^p A$ in
$\widetilde{O}( p^5 s_\lambda \eps^{-2} n + p^5\eps^{-2} \Nnz(X))$ time. This is the first sketch to achieve the optimal dependence on $s_{\lambda}$ for the polynomial kernel, after
which we can now solve the ridge regression problem in
$\widetilde{O}(n s_{\lambda}^2 \poly{p, \epsilon^{-1}})$ time.
Importantly, both running times are polynomial in $p$, whereas all previously known methods incurred running times that were exponential in $p$.

Although there has been much work on sketching methods for kernel approximation which nearly achieve the optimal target dimension $m \approx s_\lambda$, such as Nystrom sampling \cite{mm17}, all known methods are data-dependent unless strong conditions are assumed about the kernel matrix (small condition number or incoherence). Data oblivious methods provide nice advantages, such as one-round distributed protocols and single-pass streaming algorithms. However, for kernel methods they are poorly understood and previously had worse theoretical guarantees than data-dependent methods.
Furthermore, note that the Nystrom method requires to sample at least $m = {\Omega} (s_\lambda)$ landmarks to satisfy the subspace embedding property even given an oracle access to the exact leverage scores distribution. This results in a runtime of ${\Omega}\left( s_\lambda^2 d + s_\lambda \Nnz(X) \right)$.
Whereas our method achieves a target dimension that nearly matches the best dimension possible with data-dependent Nystrom method and with strictly better running time of $\widetilde{O}(n s_\lambda + \Nnz(X) )$ (assuming $p = \poly{\log n}$). Therefore, for a large range of parameter our sketch runs in input sparsity time wheras the Nystrom methods are slower by an $s_\lambda$ factor in the best case.

\paragraph{Application: Polynomial Kernel Rank-$k$ Approximation.} Approximate matrix product and subspace emebedding are key properties for sketch matrices which imply efficient algorithms for rank-$k$ kernel approximation \cite{avron2014subspace}. The following corollary of Theorem \ref{const-prob-thrm} immediately follows from Theorem 6 of \cite{avron2014subspace}.

\begin{cor}[Rank-$k$ Approximation]
For every positive integers $k, n, p, d$, every $\eps>0$, any  $X \in \RR^{d \times n}$, if $A \in \RR^{d^p \times n}$ is the matrix whose columns are obtained by the $p$-fold self-tensoring of each column of $X$ then there exists an algorithm which finds an $n \times k$ matrix $V$ in time $O\left(p \Nnz(X) + \poly{k,p,\eps^{-1}} \right)$ such that with probability $9/10$,
$$\|A - A V V^\top\|_F^2 \le (1+\eps) \min_{\substack{U \in \RR^{d^p \times n}\\ \text{rank}(U) = k}}\|A-U\|_F^2.$$
\end{cor}
Note that this runtime improves the runtime of \cite{avron2014subspace} by exponential factors in the polynomial kernel's degree $p$.

\paragraph{Additional Applications.}
Our results also imply improved bounds for each of the applications
in \cite{avron2014subspace}, including
canonical correlation analysis (CCA), and principal component regression (PCR).
Importantly, we obtain the first sketching-based solutions for these problems
with running time polynomial rather than exponential in $p$.

\paragraph{Oblivious Subspace Embedding for the Gaussian Kernel.} One very important implication of our result is Oblivious Subspace Embedding of the Gaussian kernel. Most work in this area is related to the Random Fourier Features method \cite{rr07}. It was shown in \cite{AKMMVZ17} that one requires $\Omega(n)$ samples of the standard Random Fourier Features to obtain a subspace embedding for the Gaussian kernel, while a modified distribution for sampling frequencies yields provably better performance. The target dimension of our proposed sketch for the Gaussian kernel strictly improves upon the result of \cite{AKMMVZ17}, which has an exponential dependence on the dimension $d$. We for the first time, embed the Gaussian kernel with a target dimension which has a linear dependence on the statistical dimension of the kernel and is not exponential in the dimensionality of the data-point.

\begin{restatable}{thm}{gaussiansketch}\label{gaussian-embedding}
For every $r>0$, every positive integers $n, d$, and every $X \in \RR^{d \times n}$ such that $\| x_i \|_2 \le r$ for all $i\in[n]$, where $x_i$ is the $i^{\text{th}}$ column of $X$, suppose $G \in \RR^{n \times n}$ is the Gaussian kernel matrix -- i.e., $G_{j,k} = e^{-\|x_j - x_k\|_2^2/2}$ for all $j,k \in [n]$.
There exists an algorithm which computes $S_g(X) \in \RR^{m \times n}$ in time $\widetilde{O}\left( q^6\epsilon^{-2} n s_\lambda + q^6\epsilon^{-2} \Nnz(X) \right)$ such that for every $\eps, \lambda >0$,
$$\Pr_{S_g} \left[(1-\epsilon)(G + \lambda I_n) \preceq (S_g(X))^\top S_g(X) + \lambda I_n \preceq (1+\epsilon)(G + \lambda I_n)\right] \ge 1 - 1/\poly{n},$$
where $m = \widetilde{\Theta}\left( q^5 s_\lambda/\epsilon^2 \right)$ and $q = \Theta(r^2 + \log( n/\epsilon\lambda))$ and $s_\lambda$ is $\lambda$-statistical dimension of $G$ as in Definition \ref{stat-dim-definition}.
\end{restatable}

We remark that for datasets with radius $r = \poly{\log n}$ even if one has oracle access to the exact leverage scores for Fourier features of Gaussian kernel, in order to get subspace embedding guarantee one needs to use $m = \Omega(s_\lambda)$ features which requires $\Omega(s_\lambda \Nnz(X))$ operations to compute. Wheras our result of Theorem \ref{gaussian-embedding} runs in time $\widetilde{O}(n s_\lambda + \Nnz(X))$. Therefore, for a large range of parameters our Gaussian sketch runs in input sparsity time wheras the Fourier features method is at best slower by an $s_\lambda$ factor.

\subsection{Technical Overview}

Our goal is to design a sketching matrix $\Pi^p$ that satisfies the oblivious subspace embedding property with an optimal embedding dimension and which can be efficiently applied to vectors of the form $x^{\otimes p} \in \RR^{d^p}$\footnote{Tensor product of $x$ with itself $p$ times.}. We start by describing some natural approaches to this problem (some of which have been used before), and show why they incur an exponential loss in the degree of the polynomial kernel. We then present our sketch and outline our proof of its correctness.

We first discuss two natural approaches to tensoring classical sketches, namely the Johnson-Lindenstrauss transform and the {\sf CountSketch}. We show that both lead to an exponential dependence of the target dimension on $p$ and then present our new approach.
\paragraph{Tensoring the Johnson-Lindenstrauss Transform.} Perhaps the most natural approach to designing a sketch $\Pi^p$ is the idea of tensoring $p$ independent Johnson-Lindenstrauss matrices. Specifically, let $m$ be the target dimension. For every $r=1,\ldots, p$ let $M^{(r)}$ denote an $m\times d$ matrix with iid uniformly random $\pm 1$ entries, and let the sketching matrix $M \in \RR^{m \times d^p}$ be
$$
M=\frac1{\sqrt{m}}M^{(1)}\bullet \ldots\bullet M^{(p)},
$$
where $\bullet$ stands for the operation of tensoring the rows  of matrices $M^{(r)}$ (see Definition \ref{tensornotations}).
This would be a very efficient matrix to apply, since for every $j=1,\ldots, m$ the $j$-th entry of $Mx^{\otimes p}$ is exactly $\prod_{r=1}^p \left[M^{(r)} x\right]_j$, which can be computed in time $O(p \Nnz(x))$, giving overall evaluation time $O(p m \Nnz(x))$.
One would hope that $m=O(\eps^{-2}\log n)$ would suffice to ensure that $\|Mx^{\otimes p}\|_2^2=(1\pm\epsilon)\|x^{\otimes q}\|_2^2$. However, this is not true: we show in Appendix~\ref{sec:proof:lower} that one must have $m=\Omega(\eps^{-2}3^p \log(n)/p+\eps^{-1} (\log(n)/p)^p)$ in order to preserve the norm with high probability. Thus, the dependence on degree $p$ of the polynomial kernel must be exponential. The lower bound is provided by controlling the moments of the sketch $M$ and using Paley-Zygmund inequality. For completeness, we  show that the aforementioned bound on the target dimension $m$ is sharp, i.e., necessary and sufficient for obtaining the Johnson-Lindenstrauss property.

\paragraph{Tensoring of \textsc{CountSketch} (\textsc{TensorSketch}).} Pagh and Pham~\cite{pp13} introduced the following tensorized version of {\sf CountSketch}. For every $i=1,\ldots, p$ let $h_i:[d]\to [m]$ denote a random hash function, and $\sigma_i:[d]\to [m]$ a random sign function. Then let $S: \mathbb{R}^{d^{\otimes p}} \to \mathbb{R}^m$ be defined by
$$
S_{r, (j_1,\dots, j_p)}:=\sigma(i_1) \cdots \sigma(i_p) \, \mathbf{1}[h_1(i_1)+\ldots h_p(i_p)=r]
$$
for $r=1,\ldots, m$. For every $x\in \mathbb{R}^d$ one can compute $Sx^{\otimes p}$ in time $O(p m \log m+p \Nnz(x))$. Since the time to apply the sketch only depends linearly on the dimension $p$ (due to the Fast Fourier Transform) one might hope that the dependence of the sketching dimension on $p$ is polynomial.
However, this turns out to not be the case: the argument in~\cite{avron2014subspace} implies that $m=\widetilde{O}(3^p s_\lambda^2)$ suffices to construct a subspace embedding for a matrix with regularization $\lambda$ and statistical dimension $s_\lambda$, and we show in Appendix~\ref{appendix:lowerbound-tensorsketch} that exponential dependence on $p$ is necessary.

\paragraph{Our Approach: Recursive Tensoring.}

The initial idea behind our sketch is as follows. To apply
our sketch $\Pi^{p}$ to $x^{\otimes p}$,
for $x \in \mathbb{R}^d$, we first compute
the sketches $T_{1} x, T_2 x, \ldots, T_p x$
for independent sketching matrices $T_1, \ldots, T_p \sim \tbase$ -- see the leaves of the sketching tree in Fig.~\ref{sketchingtree}.
Note that we choose these sketches as {\sf CountSketch} \cite{charikar2002finding} or {\sf OSNAP} \cite{nelson2013osnap} to ensure that the leaf sketches can be applied in time proportional to the number of nonzeros in the input data (in the case of OSNAP this is true up to polylogarithimic factors).

Each of these is a standard sketching matrix mapping
$d$-dimensional vectors to $m$-dimensional vectors for some common value of $m$.
We refer the reader to the survey \cite{w14}. The next idea is
to choose new sketching matrices $S_1, S_2, \dots, S_{p/2} \sim \sbase$, mapping
$m^2$-dimensional vectors to $m$-dimensional vectors and apply $S_1$ to
$(T_1 x) \otimes (T_2 x)$,
as well as apply $S_2$ to
$(T_3 x) \otimes (T_4 x)$, and so on, applying
$S_{p/2}$ to $(T_{p-1} x) \otimes (T_p x)$. These sketches are denoted by $S_{base}$ -- see internal nodes of the sketching tree in Fig.~\ref{sketchingtree}. We note that in order to ensure efficiency of our construction (in particular, running time that depends only linearly on the statistical dimension $s_\lambda$) we must choose $S_{base}$ as a sketch that can be computed on tensored data without explicitly constructing the actual tensored input, i.e., $S_{base}$ supports fast matrix vector product on tensor product of vectors. We use either {\sf TensorSketch} (for results that work with constant probability)  and a new variant of the Subsampled Randomized Hadamard Transform {\sf SRHT} which supports fast multiplication for the tensoring of two vectors (for high probability bounds) -- we call the last sketch {\sf TensorSRHT}.

At this point we have reduced our number of input vectors from $p$
to $p/2$, and the dimension is $m$, which will turn out to be roughly $s_{\lambda}$. We have made progress,
as we now have fewer vectors each in roughly the same dimension we started
with. After $\log_2 p$ levels in the tree we are left with a single output
vector.

Intuitively, the reason that this construction avoids an exponential dependence on $p$ is that at every level in the tree we use target dimension $m$ larger than the statistical dimension of our matrix by a factor polynomial in $p$. This ensures that  the accumulation of error is limited, as the total number of nodes in the tree is $O(p)$. This is in contrast to the direct approaches discussed above, which use a rather direct tensoring of classical sketches, thereby incurring an exponential dependence on $p$ due to dependencies that arise.

\begin{figure*}
\centering

\scalebox{0.9}{
\begin{tikzpicture}[<-, level/.style={sibling distance=75mm/#1,level distance = 2.3cm}]
\node [arn_t] (z){}
child {node [arn_t] (a){}edge from parent [electron]
child {node [arn_t] (b){}edge from parent [electron]
}
child {node [arn_t] (e){}edge from parent [electron]
}
}
child { node [arn_t] (h){}edge from parent [electron]
child {node [arn_t] (i){}edge from parent [electron]
}
child {node [arn_t] (l){}edge from parent [electron]
}
};

\node []	at (z.south)	[label=\large{${\bf \sbase}$}]	{};

\node []	at (a.south)	[label=\large{${\bf \sbase}$}]	{};

\node []	at (b.south)	[label=\large${\bf \tbase}$]	{};
\node []	at (e.south)	[label=\large${\bf \tbase}$] {};
\node []	at (h.south)	[label=\large{${\bf \sbase}$}] {};

\node []	at (i.south)	[label=\large${\bf \tbase}$] {};
\node []	at (l.south)	[label=\large${\bf \tbase}$] {};

\draw[draw=black, ->] (2.5,0.2) -- (0.7,0.1);
\draw[draw=black, ->] (4.5,-0.5) -- (3.8,-1.7);

\node [] at (2.5,0.5) [label=right:\large{ internal nodes:}]	{}
edge[->, bend right=40] (-3.6,-1.7);
\node [] at (2.5,0) [label=right:\large{ {\sf TensorSketch} or {\sf TensorSRHT}}]	{};

\draw[draw=black, ->] (4.0,-6.1) -- (5.2,-5.2);
\draw[draw=black, ->] (2.8,-6.1) -- (2,-5.2);
\draw[draw=black, ->] (1.9,-6.1) -- (-1.6,-5.15);

\node [] at (1,-6.5) [label=right:\large{leaves: {\sf CountSketch} or {\sf OSNAP}}]	{}
edge[->, bend left=15] (-5.3,-5.1);

\end{tikzpicture}
}
\par

\caption{$\sbase$ is chosen from the family of sketches which support fast matrix-vector product for tensor inputs such as {\sf TensorSketch} and {\sf TensorSRHT}. The $\tbase$ is chosen from the family of sketches which operate in input sparsity time such as {\sf CountSketch} and {\sf OSNAP}.} \label{sketchingtree}
\end{figure*}
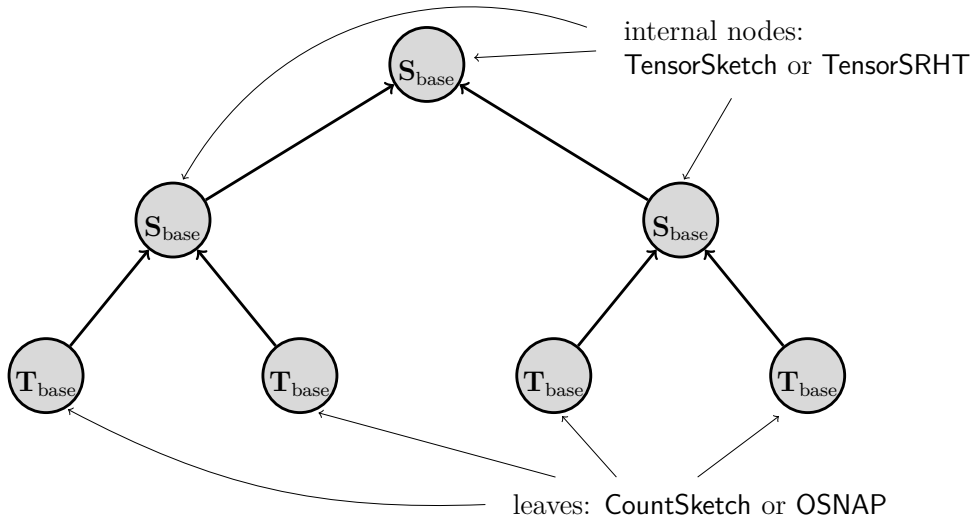

\paragraph{Showing Our Sketch is a Subspace Embedding.} In order to show that our recursive sketch is a subspace embedding, we need to argue it preserves
norms of arbitrary vectors in $\mathbb{R}^{d^p}$, not only vectors of the form $x^{\otimes p}$, i.e., $p$-fold self-tensoring of $d$-dimensional
vectors\footnote{$x^{\otimes p}$ denotes $\underbrace{x\otimes x \cdots \otimes x}_{\text{$p$ terms}}$, the $p$-fold self-tensoring of $x$.}. Indeed, all known methods for showing the subspace embedding property (see \cite{w14} for a survey) at the very least argue that the norms
of each of the columns of an orthonormal basis for the subspace in question are preserved. While our subspace may be formed by the span of vectors which are tensor
products of $p$ $d$-dimensional vectors, we are not guaranteed that there
is an orthonormal basis of this form.
Thus, we first observe that our mapping is indeed linear over
$\mathbb{R}^{d^p}$, making it well-defined on the elements of any basis for our subspace, and hence our task essentially reduces to proving that our mapping preserves norms of arbitrary vectors in $\mathbb{R}^{d^p}$.

We present two approaches to analyzing our construction. One is based on the idea of propagating moment bounds through the sketching tree, and results in a nearly linear dependence of the sketching dimension $m$ on the degree $p$ of the polynomial kernel, at the expense of a quadratic dependence on the statistical dimension $s_\lambda$. This approach is presented in Section~\ref{sec:moment-bounds}. The other approach achieves the (optimal) linear dependence on $s_\lambda$, albeit at the expense of a worse polynomial dependence on $p$. This approach uses sketches that succeed with high probability, and uses matrix concentration bounds.

\paragraph{Propagating moment bounds through the tree -- optimizing the dependence on the degree $p$.}
We analyze our recursively tensored version of the {\sf OSNAP} and {\sf CountSketch} by showing how moment bounds can be propagated through the tree structure of the sketch.
This analysis is presented in Section~\ref{sec:moment-bounds}, and results in the proof of Theorem~\ref{const-prob-thrm} as well as the first part of Theorem~\ref{high-prob-sketch}.
The analysis obtained this way give particularly sharp dependencies on $p$ and $\log1/\delta$.

The idea is to consider the unique matrix $M\in\RR^{m\times d^p}$ that acts on simple tensors in the way we have described it recursively above.
This matrix could in principle be applied to any vector $x\in\RR^{d^p}$ (though it would be slow to realise).
We can nevertheless show that this matrix has the $(\eps,\delta,t)$-JL Moment Property, which is for parameters $\eps,\delta\in[0,1],t\ge 2$, and every $x \in \RR^{d}$ with $\|x\|_2=1$ the statement $\Ep{ \left| \norm{M x}_2^2-1 \right|^t} \le \eps^t\delta$.

It can be shown that $M$ is built from our various $\sbase$ and $\tbase$ matrices using three different operations: multiplication, direct sum, and row-wise tensoring.
In other words, it is sufficient to show that if $Q$ and $Q'$ both have the $(\eps,\delta,t)$-JL Moment Property, then so does $QQ'$, $Q\oplus Q'$ and $Q\bullet Q'$.
This turns out to hold for $Q\oplus Q'$, but $QQ'$ and $Q\bullet Q'$ are more tricky.
(Here $\oplus$ is the direct sum and $\bullet$ is the composition of tensoring the rows.
See \cref{section:preliminaries} on notation.)

For multiplication, a simple union bound allows us to show that $Q^{(1)}Q^{(2)}\cdots Q^{(p)}$ has the $(p\eps,p\delta,t)$-JL Moment Property.
This would unfortunately mean a factor of $p^2$ in the final dimension.
The union bound is clearly suboptimal,
since implicitly it is assumes that all the matrices conspire to either shrink or increase the norm of a vector, while in reality with independent matrices, we should get a random walk on the real line.
Using an intricate decoupling argument, we show that this is indeed the case, and that $Q^{(1)}Q^{(2)}\cdots Q^{(p)}$ has the $(\sqrt{p}\eps,\delta,t)$-JL Moment Property, saving a factor of $p$ in the output dimension.

Finally we need to analyze $Q\bullet Q'$.
Here it is easy to show that the JL Moment Property doesn't in general propagate to $Q\bullet Q'$
(consider e.g. $Q$ being constant 0 on its first $m/2$ rows and $Q'$ having 0 on its $m/2$ last rows.)
For most known constructions of JL matrices it does however turn out that $Q\bullet Q'$ behaves well.
In particular we show this for matrices with independent sub-Gaussian entries (\cref{section:upper-subgauss}), and for the so-called Fast Johnson Lindenstrauss construction~\cite{AilonC06} (\cref{lem:tensorfastjl}).
The main tool here is a higher order version of the classical Khintchine inequality~\cite{haagerup2007best} which bounds the moments
$\Ep{\langle \sigma^{(1)} \otimes \sigma^{(2)}\otimes\cdots\otimes \sigma^{(p)}, x\rangle^t}$ when $\sigma^{(1)}, \dots \sigma^{(p)}$ are independent sub-Gaussian vectors (\cref{lem:gen-khinchine}).

\paragraph{Optimizing the dependence on $s_\lambda$.} Our proof of Theorem~\ref{high-prob-sketch} relies on instantiating our framework with {\sf OSNAP} at the leaves of the tree ($T_{base}$) and a novel version of the {\sf SRHT} that we refer to as {\sf TensorSRHT} at the internal nodes of the tree.
We outline the analysis here. In order to show that our sketch preserves norms, let $y$ be an
arbitrary vector in $\mathbb{R}^{d^p}$. Then in the bottom level of the
tree, we can view our sketch as
$T_1 \times T_2 \times \cdots \times T_p$, where $\times$
for denotes the tensor product of matrices (see Definition \ref{tensoring-matrix-vector-prod}). Then, we can reshape $y$
to be a $ d^{q-1} \times d$ matrix $Y$, and the entries of
$T_1 \times T_2 \times \cdots \times T_p y$ are in bijective
correspondence with those of $T_1 \times T_2 \times \cdots \times T_{p-1} Y T_p^\top$. By definition of $T_p$,
it preserves the Frobenius norm of $Y$, and consequently, we can replace
$Y$ with $Y T_p^\top$. We next look at $(T_1 \times T_2 \times \cdots \times T_{p-2}) Z (I_d \times T_{p-1}^\top)$, where $Z$ is the $d^{p-2} \times d^2$ matrix with entries in bijective correspondence with those of $Y T_{p}^\top$. Then we know that $T_{p-1}$ preserves the Frobenius norm of $Z$. Iterating in this fashion, this means the first layer of our tree preserves the norm of $y$,
provided we union bound over $O(p)$ events that a sketch preserves a norm
of an intermediate matrix. The core of the analysis consists of applying spectral concentration bounds based analysis to sketches that act on blocks of the input vector in a correlated fashion. We give the details in Section~\ref{sec:linear-s-lambda}.

\paragraph{Sketching the Gaussian kernel.} Our techniques yield the first oblivious sketching method for the Gaussian kernel with target dimension that does not depend exponentially on the dimensionality of the input data points. The main idea is to Taylor expand the Gaussian function and apply our sketch for the polynomial kernel to the elements of the expansion. It is crucial here that the target dimension of our sketch for the polynomial kernel depends only polynomially on the degree, as otherwise we would not be able to truncate the Taylor expansion sufficiently far in the tail (the number of terms in the Taylor expansion depends on the radius of the dataset and depends logarithmically on the regularization parameter). Overall, our Gaussian kernel sketch has optimal target dimension up to polynomial factors in the radius dataset and logarithmic factors in the dataset size. Moreover, it is the first subspace embedding of Gaussian kernel which runs in input sparsity time $\widetilde{O}(\Nnz(X))$ for datasets with polylogarithmic radius. The result is summarized in Theorem~\ref{gaussian-embedding}, and the analysis is presented in Section~\ref{sec:gaussian-ose}.


\subsection{Related Work}
Work related to sketching of tensors and explicit kernel embeddings is found in fields ranging from pure mathematics to physics and machine learning.
Hence we only try to compare ourselves with the four most common types we have found.


\paragraph{Johnson-Lindenstrauss Transform}
A cornerstone result in the field of subspace embeddings is the Johnson-Lindenstrauss lemma~\cite{Johnson1986}:
``For all $\eps \in [0, 1]$, integers $n, d \ge 1$, and $X \subseteq \RR^d$ with $\abs{X} = n$
there exists $f : \RR^d \to \RR^m$ with $m = O(\eps^{-2} \log(n))$, such that
$
(1 - \eps)\|x - y\|_2 \le \|f(x) - f(y)\|_2 \le (1 + \eps)\|x - y\|_2
$
for every $x, y \in X$.``

It has been shown in~\cite{ClarksonW13, CohenNW16} there exists a
constant $C$, so that, for any $r$-dimensional subspace $U \subseteq \RR^d$,
there exists a subset $X \subseteq U$ with $\abs{X} \le C^r$, such that
$
\max_{x \in U}\abs{\|f(x)\|_2^2 - \|x\|_2^2}
\le O(\max_{x \in X}\abs{\|f(x)\|_2^2 - \|x\|_2^2})
$.
So the Johnson-Lindenstrauss Lemma implies that there exists a subspace embedding
with $m = O(\eps^{-2} r)$.

It is not enough to know that the subspace embedding exists, we also need the
to find the dimension-reducing map $f$, and we want the map $f$ to be applied
to the data quickly. Achlioptas showed that if $\Pi \in \RR^{m \times d}$ is
random matrix with i.i.d. entries where $\Pi_{i, j} = 0$ with probability $2/3$,
and otherwise $\Pi_{i, j}$ is uniform in $\{-1, 1\}$, and $m = O(\eps^{-2}\log(1/\delta))$,
then $\|\Pi x\|_2 = (1 \pm \eps)\|x\|_2$ with probability $1 - \delta$ for any
$x \in \RR^d$~\cite{Achlioptas2003}. This gives a running time of $O(m \nnz{x})$
to sketch a vector $x \in \RR^d$. Later, the Fast Johnson Lindenstrauss Transform~\cite{AilonC06},
which exploits the Fast Fourier Transform, improved the running time for dense vectors to
$O(d \log d + m^3)$. The related Subsampled Randomized Hadamard Transform has
been extensively studied~\cite{Sarlos06, DrineasMM06, DrineasMMS11, Tropp11, DrineasMMW12, LuDFU13},
which uses $O(d \log d)$ time but obtains suboptimal dimension $O(\eps^{-2} \log(1/\delta)^2)$,
hence it can not use the above argument to get subspace embedding, but it has been proven
in~\cite{Tropp11} that if $m = O(\eps^{-2}(r + \log(1/\delta)^2))$, then one get a subspace
embedding.

The above improvements has a running time of $O(d \log d)$, which can be worse
than $O(m \nnz{x})$ if $x \in \RR^d$ is very sparse. This inspired a line of work
trying to obtain sparse Johnson Lindenstrauss transforms~\cite{DasguptaKS10, KaneN14, nelson2013osnap, c16}. They obtain
a running time of $O(\eps^{-1}\log(1/\delta)\nnz{x})$. In~\cite{nelson2013osnap} they define the ONSAP
transform and investigate the trade-off between sparsity and subspace embedding dimension.
This was further improved in~\cite{c16}.

In the context of this paper all the above mentioned methods have the same shortcoming, they
do not exploit the extra structure of the tensors. The Subsampled Randomized Hadamard Transform
have a running time of $\Omega(p d^p \log(p))$ in the model considered in this paper, and
the sparse embeddings have a running time of $\Omega(\Nnz(x)^p)$. This is clearly unsatisfactory
and inspired the TensorSketch~\cite{pp13,avron2014subspace}, which has a running time of $\Omega(p \Nnz(x))$.
Unfortunately, they need $m = \Omega(3^p \eps^{-2} \delta^{-1})$ and one of the
main contributions of this paper is get rid of the exponential dependence on $p$.

\paragraph{Approximate Kernel Expansions}
A classic result by Rahimi and Recht~\cite{rahimi2008random}
shows how to compute an embedding for any shift-invariant kernel function $k(\|x-y\|_2)$ in time $O(dm)$.
In~\cite{DBLP:journals/corr/LeSS14} this is improved to any kernel on the form $k(\langle x,y\rangle)$ and time $O((m+d)\log d)$,
however the method does not handle kernel functions that can't be specified as a function of the inner product, and it doesn't provide subspace embeddings.
See also~\cite{mm17} for more approaches along the same line. Unfortunately, these methods are unable to operate in input sparsity time and their runtime at best is off by an $s_\lambda$ factor.

\paragraph{Tensor Sparsification}
There is also a literature of tensor sparsification based on sampling~\cite{nguyen2015tensor}, however unless the vectors tensored are already very smooth (such as $\pm1$ vectors), the sampling has to be weighted by the data.
This means that these methods in aren't applicable in general to the types of problems we consider, where the tensor usually isn't known when the sketching function is sampled.

\paragraph{Hyper-plane rounding}
An alternative approach is to use hyper-plane rounding to get vectors on the form $\pm1$.
Let $\rho = \frac{\langle x,y\rangle}{\|x\|\|y\|}$, then we have
$
\langle\sign{Mx}, \sign{My}\rangle
= \sum_i \sign{M_i x}\sign{M_i y}
= \sum_i X_i
$
,
where $X_i$ are independent Rademachers with $\mu/m = E[X_i] = 1-\frac{2}{\pi}\arccos\rho = \frac{2}{\pi}\rho + O(\rho^3)$.
By tail bounds then
$\Pr[|\langle\sign{Mx}, \sign{My}\rangle - \mu|>\epsilon\mu]
\le2\exp(-\min(\frac{\epsilon^2\mu^2}{2\sigma^2}, \frac{3\epsilon\mu}{2}))
$.
Taking $m = O(\rho^{-2}\epsilon^{-2}\log1/\delta)$ then suffices with high probability.
After this we can simply sample from the tensor product using simple sample bounds.

The sign-sketch was first brought into the field of data-analysis by~\cite{charikar2002similarity} and~\cite{valiant2015finding} was the first, in our knowledge, to use it with tensoring.
The main issue with this approach is that it isn't a linear sketch, which hinders the applications we consider in this paper, such as kernel low rank approximation, CCA, PCR, and ridge regression.
It also takes $dm$ time to calculate $Mx$ and $My$ which is unsatisfactory.

\subsection{Organization}
In section \ref{section:preliminaries} we introduce basic definitions and notations that will be used throughout the paper. Section \ref{sec:sektch-construction} introduces our recursive construction of the sketch which is our main technical tool for sketching high degree tensor products. Section \ref{sec:moment-bounds} analyzes how the moment bounds propagate through our recursive construction thereby proving Theorems \ref{const-prob-thrm} and \ref{thm:high-prob-moment} which have linear dependence on the degree $q$. Section \ref{sec:linear-s-lambda} introduces a high probability Oblivious Subspace Embedding with linear dependence on the statistical dimension thereby proving Theorem \ref{high-prob-sketch}. Finally, section \ref{sec:gaussian-ose} uses the tools that we build for sketching polynomial kernel and proves that, for the first time, Gaussian kernel can be sketched without an exponential loss in the dimension with provable guarantees. Appendix \ref{sec:proof:lower} proves lower bounds.

\section{Preliminaries}\label{section:preliminaries}
In this section we introduce notation and present useful properties of tensor product of vectors and matrices as well as properties of linear sketch matrices.

We denote the tensor product of vectors $a, b$ by $a \otimes b$ which is formally defined as follows,
\begin{defn}[Tensor product of vectors]
	Given $a\in\RR^m$ and $b\in\RR^n$ we define the \emph{twofold tensor product} $a\otimes b$ to be
	\[
	a\otimes b = \begin{bmatrix} a_1 b_1 & a_1 b_2 & \cdots & a_1 b_n\\ a_2 b_1 & a_2 b_2 & \cdots & a_2 b_n\\ \vdots & \vdots &  & \vdots \\ a_m b_1 & a_m b_2 & \cdots & a_m b_n \end{bmatrix} \in \RR^{m\times n}.
	\]
	Although tensor products are multidimensional objects, it is often convenient to associate them with single-dimensional vectors. In particular, we will often associate $a\otimes b$ with the single-dimensional column vector
	$(a_1 b_1, a_2 b_1, \dots, a_m b_1, a_1 b_2, a_2 b_2, \dots, a_m b_2, \dots, a_m b_n)$.
	\
	Given $v_1 \in \RR^{d_1}, v_2 \in \RR^{d_2} \cdots v_k \in \RR^{d_k}$, we define the \emph{$k$-fold tensor product} $v_1 \otimes v_2  \cdots \otimes v_k \in \RR^{d_1 d_2 \cdots  d_k}$.
	For shorthand, we use the notation $v^{\otimes k}$ to denote $\underbrace{v\otimes v \cdots \otimes v}_{\text{$k$ terms}}$, the $k$-fold self-tensoring of $v$.
	
\end{defn}

Tensor product can be naturally extended to matrices which is formally defined as follows,
\begin{defn} \label{tensoring-matrix-vector-prod}
	Given $A_1\in \RR^{m_1\times n_1}, A_2\in\RR^{m_2\times n_2}, \cdots, A_k\in \RR^{m_k\times n_k}$, we define $A_1 \times A_2 \times \cdots \times A_k$ to be the matrix in $\RR^{ m_1 m_2 \cdots m_k \times n_1 n_2 \cdots n_k }$ whose element at row $(i_1, \cdots, i_k)$ and column $(j_1, \cdots, j_k)$ is $A_1(i_1,j_1) \cdots A_k(i_k,j_k)$. As a consequence the following holds for any $v_1 \in \RR^{n_1}, v_2 \in \RR^{n_2}, \cdots, v_k \in \RR^{n_k}$:
	$(A_1 \times A_2 \times \cdots \times A_k)(v_1\otimes v_2 \otimes \cdots \otimes v_k) = (A_1 v_1) \otimes (A_2 v_2) \otimes \cdots \otimes (A_k v_k)$.
	
\end{defn}

The tensor product has the useful \emph{mixed product property}, given in the following Claim,
\begin{claim}\label{tensor-prod-identity}
	For every matrices $A, B, C, D$ with appropriate sizes, the following holds,
	$$(A \cdot B) \times (C \cdot D) = (A \times C) \cdot (B \times D).$$
\end{claim}

We also define the column wise tensoring of matrices as follows,
\begin{defn}
	Given $A_1\in \RR^{m_1\times n}, A_2\in\RR^{m_2\times n}, \cdots, A_k\in \RR^{m_k\times n}$, we define $A_1\otimes A_2 \otimes \cdots \otimes A_k$ to be the matrix in $\RR^{m_1 m_2 \cdots m_k \times n}$ whose $j^{\text{th}}$ column is $A_1^j \otimes A_2^j \otimes \cdots \otimes A_k^j$ for every $j\in[n]$, where $A_l^j$ is the $j^{\text{th}}$ column of $A_l$ for every $l\in[k]$.
\end{defn}

Similarly the row wise tensoring of matrices are introduced in the following Definition,
\begin{defn} \label{tensornotations}
	Given $A^1\in \RR^{m\times n_1}, A^2\in\RR^{m\times n_2}, \cdots, A^k\in \RR^{m\times n_k}$, we define $A^1\bullet A^2 \bullet \cdots A^k$ to be the matrix in $\RR^{m \times n_1 n_2 \cdots n_k}$ whose $j^{\text{th}}$ row is $(A^1_j \otimes A^2_j \otimes \cdots \otimes A^k_j)^\top$ for every $j\in[m]$, where $A^l_j$ is the $j^{\text{th}}$ row of $A^l$ as a column vector for every $l\in[k]$.
\end{defn}

\begin{defn}
	Another related operation is the \emph{direct sum} for vectors: $x \oplus y = \left[\begin{smallmatrix} x \\ y \end{smallmatrix}\right]$
	and for matrices: $A \oplus B = \left[\begin{smallmatrix}A & 0 \\ 0 & B \end{smallmatrix}\right]$.
	When the sizes match up, we have $(A \oplus B)(x\oplus y) = Ax + By$.
	Also notice that if $I_k$ is the $k\times k$ identity matrix, then $I_k \otimes A = \underbrace{A\oplus\dots\oplus A}_{k \text{ times}}$.
\end{defn}

\section{Construction of the Sketch}\label{sec:sektch-construction}
In this section, we present the basic construction for our new sketch. Suppose we are given $v_1, v_2, \ldots v_q \in \Rbb^m$. Our main task is to map the tensor product $v_1 \otimes v_2 \otimes \cdots \otimes v_q$ to a vector of size $m$ using a linear sketch.

Our sketch construction is recursive in nature. To illustrate the general idea, let us first consider the case in which $q \geq 2$ is a power of two. Our sketch involves first sketching each pair ($v_1\otimes v_2), (v_3\otimes v_4), \cdots, (v_{q-1}\otimes v_q) \in \RR^{m^2}$ independently using independent instances of some linear base sketch (e.g., degree two {\sf TensorSketch}, Sub-sampled Randomized Hadamard Transform ({\sf SRHT}), {\sf CountSketch}, {\sf OSNAP}). The number of vectors after this step is half of the number of vectors that we began with. The natural idea is to recursively apply the same procedure on the sketched tensors with half as many instances of the base sketch in each successive step.

More precisely, we first choose a (randomized) base sketch $\sbase: \RR^{m^2}\to \RR^m$ that sketches twofold tensor products of vectors in $\RR^m$ (we will describe how to choose the base sketch later). Then, for any power of two $q\geq 2$, we define $Q^q: \RR^{m^q}\to\RR^m$ on $v_1 \otimes v_2 \otimes \cdots \otimes v_q$ recursively as follows:
\[
Q^q ( v_1 \otimes v_2 \otimes \cdots \otimes v_q ) = Q^{q/2} \left( S^q_1 (v_1 \otimes v_2) \otimes S^q_2(v_3 \otimes v_4) \otimes \cdots \otimes S^q_{q/2}(v_{q-1} \otimes v_q)  \right),
\]
where $S^q_1, S^q_2,\cdots, S^q_{q/2}: \RR^{m^2} \rightarrow \RR^m$ are independent instances of $\sbase$ and $Q^1: \RR^m \to \RR^m$ is simply the identity map on $\RR^m$.

The above construction of $Q^q$ has been defined in terms of its action on $q$-fold tensor products of vectors in $\RR^m$, but it extends naturally to a linear mapping from $\RR^{m^q}$ to $\RR^m$. The formal definition of $\Pi^q$ is presented below.

\begin{defn}[Sketch $Q^q$]\label{def:sketch-tilde}
Let $m \geq 2$ be a positive integer and let $\sbase: \RR^{m^2} \to \RR^m$ be a linear map that specifies some base sketch. Then, for any integer power of two $q\geq 2$, we define $Q^q: \RR^{m^q} \to \RR^m$ to be the linear map specified as follows:

$$Q^q \eqdef  S^2 \cdot S^4 \cdots S^{q/2} \cdot S^q,$$
where for each $l \in \{ 2^1, 2^2, \cdots , q/2, q \}$, $S^l$ is a matrix in $\RR^{m^{l/2} \times m^l}$ defined as
\begin{equation}\label{def:sketch-firstlevel}
S^l \eqdef S^l_{1} \times S^l_2 \times \cdots \times S^l_{l/2},
\end{equation}
where the matrices $S^l_{1}, \cdots, S^l_{l/2} \in \RR^{m\times m^2}$ are drawn independently from a base distribution $\sbase$. 

\end{defn}

This sketch construction can be best visualized using a balanced binary tree with $q$ leaves. Figure \ref{binary-tree-sketch} illustrates the construction of degree $4$, $Q^4$.

  \begin{figure*}[h]
	\centering

		\scalebox{0.9}{
			\begin{tikzpicture}[<-, level/.style={sibling distance=75mm/#1,level distance = 2.3cm}]
			\node [arn] (z){}
			child {node [arn] (a){}edge from parent [electron]
				child {node [arn_l] (b){}
				}
				child {node [arn_l] (e){}
				}
			}
			child { node [arn] (h){}edge from parent [electron]
				child {node [arn_l] (i){}
				}
				child {node [arn_l] (l){}
				}
			};
			
			\node []	at (z.south)	[label=\large{${\bf S^2_1}$}]	{};
			\node []	at (z.south)	[label=below:\large{$w_1\otimes w_2$}]	{};
			\node []	at (z.north)	[label={${\bf z}=S_1^2(w_1 \otimes w_2)$}]	{};
			
			\node []	at (a.south)	[label=\large{${\bf S_1^4}$}]	{};
			\node []	at (a.south)	[label=below:\large{$v_1\otimes v_2$}]	{};
			\node []	at (a.north)	[label={${\bf w_1}=S_1^4(v_1 \otimes v_2)$}]	{};	
			
			\node []	at (b)	[label=below:\large${\bf v_1}$]	{};
			\node []	at (e)	[label=below:\large${\bf v_2}$] {};
			\node []	at (h.south)	[label=\large{${\bf S_2^4}$}] {}; 
			\node []	at (h.south)	[label=below:\large{$v_3 \otimes v_4$}] {}; 
			\node []	at (h.north)	[label={${\bf w_2}=S_2^4(v_1 \otimes v_2)$}] {}; 
			
			\node []	at (i)	[label=below:\large${\bf v_3}$] {}; 
			\node []	at (l)	[label=below:\large${\bf v_4}$] {};

			\node [] at (6,0) [label=right:\large${ \bf S^2=S^2_1}$]	{};
			\node [] at (5.5,-2.3) [label=right:\large${ \bf S^4 = S^4_1 \times S^4_2}$]	{};
			
			\end{tikzpicture}
		}
		\par

	\caption{Visual illustration of the recursive construction of $Q^q$ for degree $q=4$. The input tensor is $v_1 \otimes v_2 \otimes v_3 \otimes v_4$ and the output is $z = Q^4 (v_1 \otimes v_2 \otimes v_3 \otimes v_4)$. The intermediate nodes sketch the tensors $w_1=S_1^4(v_1 \otimes v_2)$ and $w_1=S_2^4(v_3 \otimes v_4)$.} \label{binary-tree-sketch}
\end{figure*}
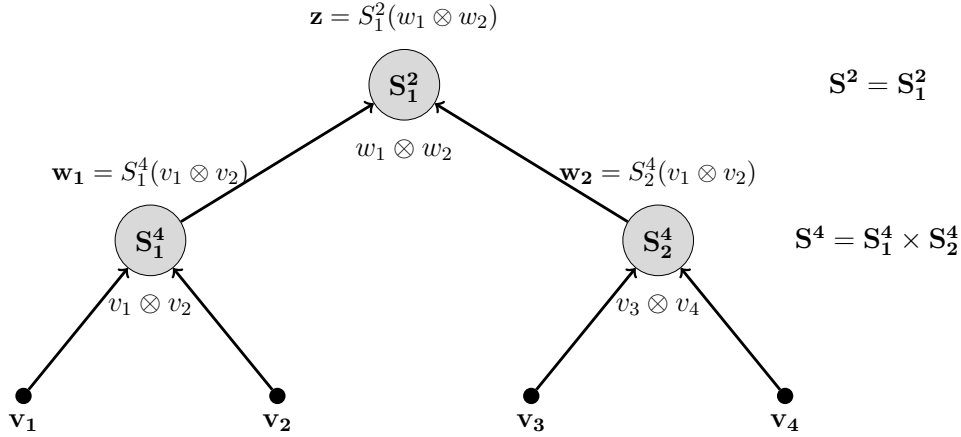

For every integer $q$ which is a power of two, by definition of $S^q$ in \eqref{def:sketch-firstlevel} of Definition \ref{def:sketch-tilde}, $S^q = S^q_{1} \times \cdots \times S^q_{q/2}$. Hence, by claim \ref{tensor-prod-identity} we can write,

$$ S^q = S^q_{1} \times \cdots \times S^q_{q/2} = \left( S^q_{1} \times \cdots \times S^q_{q/2-1} \times I_{m} \right) \cdot \left( I_{m^{q-2}} \times S^q_{q/2} \right). $$
By multiple applications of Claim \ref{tensor-prod-identity} we have the following claim,

\begin{claim}\label{claim-tensor2product-reduction}
	For every power of two integer $q$ and any positive integer $m$, if $S^q$ is defined as in \eqref{def:sketch-firstlevel} of Definition \ref{def:sketch-tilde}, then
$$ S^q = M_{q/2} M_{q/2-1} \cdots M_{1},$$
	where $M_j = I_{m^{q-2j}} \times S^{q}_{q/2 - j +1} \times I_{m^{j-1}}$ for every $j \in [q/2]$.
\end{claim}

\paragraph{Embedding $\RR^{d^q}$:}
So far we have constructed a sketch $Q^q$ for sketching tensor product of vectors in $\RR^m$. However, in general the data points can be in a space $\RR^d$ of arbitrary dimension. A natural idea is to reduce the dimension of the vectors by a mapping from $\RR^d$ to $\RR^m$ and then apply $Q^q$ on the tensor product of reduced data points.
The dimensionality reduction defines a linear mapping from $\RR^{d^q}$ to $\RR^{m^d}$ which can be represented by a matrix. We denote the dimensionality reduction matrix by $T^q \in \RR^{m^q \times d^q}$ formally defined as follows. 

\begin{defn}\label{def-dim-reduction}
Let $m, d$ be positive integers and let $T_{\text{base}}: \RR^{d} \rightarrow \RR^m$ be a linear map that specifies some base sketch. Then for any integer power of two $q$ we define $T^q$ to be the linear map specified as follows,
$$T^q = T_1 \times T_2 \times \cdots \times T_q,$$
where the matrices $T_1, \cdots, T_q$ are drawn independently from $\tbase$. 

{\bf Discussion:} Similar to Claim \ref{claim-tensor2product-reduction}, the transform $T^q$ can be expressed as the following product of $q$ matrices,
$$T^q = M_q M_{q-1} \cdots M_1,$$
where $M_j = I_{d^{q-j}} \times T_{q-j+1} \times I_{m^{j-1}}$ for every $j \in [q]$.

\end{defn}

Now we define the final sketch $\Pi^q: \RR^{d^q} \rightarrow \RR^{m}$ for arbitrary $d$ as the composition of $Q^q \cdot T^q$. Moreover, to extend the definition to arbitrary $q$ which is not necessarily a power of two we tensor the input vector with a standard basis vector a number of times to make the input size compatible with the sketch matrices. The sketch $\Pi^q$ is formally defined below,
\begin{defn}[Sketch $\Pi^p$]\label{def:sketch}
	Let $m , d$ be positive integers and let $\sbase: \RR^{m^2} \to \RR^m$ and $T_{\text{base}}: \RR^{d} \rightarrow \RR^m$ be linear maps that specify some base sketches. Then, for any integer $p\geq 2$ we define $\Pi^p: \RR^{d^p} \to \RR^m$ to be the linear map specified as follows:
	\begin{enumerate}
		\item If $p$ is a power of two, then $\Pi^p$ is defined as
		$$\Pi^p = Q^p \cdot T^p,$$
		where $Q^p \in \RR^{m \times m^p}$ and $T^p \in \RR^{m^p \times d^p}$ are sketches as in Definitions \ref{def:sketch-tilde} and \ref{def-dim-reduction} respectively.
		
		\item If $p$ is not a power of two, then let $q = 2^{\lceil \log_2p \rceil}$ be the smallest power of two integer that is greater than $p$ and we define $\Pi^p$ as
		$$\Pi^p(v) = {\Pi}^q \left( v \otimes e_1^{\otimes (q-p)} \right),$$
		for every $v \in \RR^{d^p}$, where $e_1 \in \RR^d$ is the standard basis column vector with a 1 in the first coordinate and zeros elsewhere, and ${\Pi}^q$ is defined as in the first part of this definition. 
	\end{enumerate}
\end{defn}

Algorithm \ref{alg:main} sketches $x^{\otimes p}$ for any integer $p$ and any input vector $x \in \RR^d$ using the sketch $\Pi^p$ as in Definition \ref{def:sketch}, i.e., computes $\Pi^p ( x^{\otimes p})$.

\begin{algorithm}[h]
	\caption{\algoname{Sketch for the Tensor $x^{\otimes p}$}}
	{\bf input}: vector $x \in \RR^d$, dimension $d$, degree $p$, number of buckets $m$, base sketches $\sbase \in \RR^{m \times m^2}$ and $\tbase \in \RR^{m \times d}$ \\
	{\bf output}: sketched vector $z \in \RR^m$
	\begin{algorithmic}[1]
		\State{Let $q = 2^{\lceil \log_2p \rceil}$}
		\State{Let $T_1, \cdots T_q$ be independent instances of the base sketch $\tbase : \RR^d \rightarrow \RR^m$}
		\State{For every $j\in \{1, 2, \cdots, p\}$, let $Y^0_j = T_j \cdot x$}\label{y-0-q}
		\State{For every $j\in \{p+1, \cdots, q \}$, let $Y^0_j = T_j \cdot e_1$, where $e_1$ is the standard basis vector in $\RR^d$ with value $1$ in the first coordinate and zero elsewhere}\label{y-0-p}
		
		\For{$l =  1 $ to $\log_2q$}\label{alg-loop}
		\State{Let $S^{q/2^{l-1}}_{1}, \cdots, S^{q/2^{l-1}}_{q/2^l}$ be independent instances of the base sketch $\sbase: \RR^{m^2} \rightarrow \RR^m$}
		\State{For every $j \in \{ 1, \cdots, q/2^{l} \}$ let $Y^l_j = S^{q/2^{l-1}}_{j} \left( Y^{l-1}_{2j-1} \otimes Y^{l-1}_{2j} \right)$}\label{Y-l-j}
		\EndFor
		
		\\\Return{$z = Y^{\log_2q}_1$}
	\end{algorithmic}
	\label{alg:main}
\end{algorithm}

We show the correctness of Algorithm \ref{alg:main} in the next lemma.
\begin{lem}\label{algo-sketch-correctness}
For any positive integers $d$, $m$, and $p$, any distribution on matrices $\sbase: \RR^{m^2} \to \RR^m$ and $T_{\text{base}}: \RR^{d} \rightarrow \RR^m$ which specify some base sketches, any vector $x \in \RR^{d}$, Algorithm \ref{alg:main} computes $\Pi^p (x^{\otimes p})$ as in Definition \ref{def:sketch}.
\end{lem}
\begin{proof}

For every input vector $x \in \RR^d$ to Algorithm~\ref{alg:main}, the vectors $Y^0_1, \cdots, Y^0_{p}$, are computed in lines~\ref{y-0-q} and \ref{y-0-p} of algorithm as
$Y^0_{j} = T_j \cdot x,$
for all $j \in \{ 1, \cdots, p \}$, and,
$Y^0_{j'} = T_{j'} \cdot e_1,$
for all $j \in \{ q + 1, \cdots, q \}$.
Therefore, as shown in Definition \ref{tensoring-matrix-vector-prod}, the following holds,

$$Y^0_1 \otimes \cdots \otimes Y^0_p = T_1 \times \cdots \times T_q \cdot \left(x^{\otimes p} \otimes e_1^{\otimes(q-p)}\right).$$
From the definition of sketch $T^q$ as per Definition \ref{def-dim-reduction} it follows that,
\begin{equation}\label{alg-osnap}
Y^0_1 \otimes \cdots \otimes Y^0_q = T^q \cdot \left(x^{\otimes p} \otimes e_1^{\otimes(q-p)}\right).
\end{equation}
The algorithm computes $Y^l_1, \cdots Y^l_{q/2^{l}}$ in line~\ref{Y-l-j} as, $Y^l_j = S^{q/2^{l-1}}_{j} \left( Y^{l-1}_{2j-1} \otimes Y^{l-1}_{2j} \right),$
for every $j \in \{ 1,\cdots, q/2^{l} \}$ and every $l \in \{ 1, \cdots, \log_2q \}$ in a for loop. Therefore, by Claim \ref{tensor-prod-identity},
$$Y^l_1 \otimes \cdots \otimes Y^l_{q/2^{l}} = \left(S^{q/2^{l-1}}_1 \times \cdots \times S^{q/2^{l-1}}_{q/2^{l}}\right) \cdot Y^{l-1}_1 \otimes \cdots \otimes Y^{l-1}_{q/2^{l-1}}.$$
By the definition of the sketch $S^{q/2^{l-1}}$ in \eqref{def:sketch-firstlevel} of Definition~\ref{def:sketch-tilde} we have that for every $l \in \{1, \cdots, \log_2q\}$,

$$Y^l_1 \otimes \cdots \otimes Y^l_{q/2^{l}} = S^{q/2^{l-1}} \cdot Y^{l-1}_1 \otimes \cdots \otimes Y^{l-1}_{q/2^{l-1}}.$$
Therefore, by recursive application of the above identity we get that,

$$Y^{\log_2p}_1 = S^{2} \cdot S^4 \cdots S^{q/2} \cdot S^q \cdot Y^0_1 \otimes \cdots \otimes Y^0_q.$$
From the definition of sketch $Q^q$ as in Definition \ref{def:sketch-tilde} it follows that,
$$Y^{\log_2q}_1 = Q^q \cdot Y^0_1 \otimes \cdots \otimes Y^0_q.$$
Substituting $ Y^0_1 \otimes \cdots \otimes Y^0_q$ from \eqref{alg-osnap} in the above gives,
$z = (Q^q \cdot T^q) \cdot \left( x^{\otimes p} \otimes e_1^{\otimes(q-p)} \right),$
where by Definition \ref{def:sketch} we have that,
$z = \Pi^p (x^{\otimes p}).$
\end{proof}

\paragraph{Choices of the Base Sketches $\sbase$ and $\tbase$:}

We present formal definitions for various choices of the base sketches $\sbase$ and $\tbase$ that will be used for our sketch construction $\Pi^q$ of Definition \ref{def:sketch}. We start by briefly recalling the {\sf CountSketch} \cite{charikar2002finding}.

\begin{defn}[{\sf CountSketch} transform] \label{def-countsketch}
	Let $h: [d] \rightarrow [m]$ be a 3-wise independent hash function and also let $\sigma : [d] \rightarrow \{-1, +1\}$ be a 4-wise independent random sign function. Then, the {\sf CountSketch} transform, $S : \RR^{d} \rightarrow \RR^m$, is defined as follows; for every $i \in [d]$ and every $r \in [m]$,
	$$S_{r,i} = \sigma(i) \cdot \mathbbm{1}\left[ h(i)= r \right].$$
\end{defn}

Another base sketch that we consider is the {\sf TensorSketch} of degree two \cite{p13} defined as follows.

\begin{defn}[degree two {\sf TensorSketch} transform] \label{deg-2-tensorsketch}
Let $h_1, h_2 : [d] \rightarrow [m]$ be 3-wise independent hash functions and also let $\sigma_1, \sigma_2 : [d] \rightarrow \{-1, +1\}$ be 4-wise independent random sign functions. Then, the degree two {\sf TensorSketch} transform, $S : \RR^{d} \times \RR^{d} \rightarrow \RR^m$, is defined as follows; for every $i,j \in [d]$ and every $r \in [m]$,
$$S_{r,(i,j)} = \sigma_1(i) \cdot \sigma_2(j) \cdot \mathbbm{1}\left[ h_1(i) + h_2(j) = r \mod m \right].$$

{\bf Remark:} $S (x^{\otimes 2})$ can be computed in $O(m\log m + \text{nnz}(x))$ time using the Fast Fourier Transform.
\end{defn}

Now let us briefly recall the SRHT~\cite{AilonC06}.
\begin{defn}[Subsampled Randomized Hadamard Transform ({\sf SRHT})]\label{def:fastjl}
	Let $D$ be a $d\times d$ diagonal matrix with independent Rademacher random variables along the diagonal. Also, let $P \in \{0,1\}^{m \times d}$ be a random sampling matrix in which each row contains a $1$ at a uniformly distributed coordinate and zeros elsewhere, and let $H$ be a $d\times d$ Hadamard matrix. Then, the {\sf SRHT}, $S \in \RR^{m \times d}$,  is $S = \frac{1}{\sqrt{m}} P H D$.
\end{defn}

We now define a variant of the {\sf SRHT} which is very efficient for sketching $x^{\otimes 2}$ which we call the \emph{TensorSRHT}.

\begin{defn}[Tensor Subsampled Randomized Hadamard Transform ({\sf TensorSRHT})]\label{def:tensorfastjl}
	Let $D_1$ and $D_2$ be two independent ${d}\times{d}$ diagonal matrices, each with diagonal entries given by independent Rademacher variables. Also let $P \in \{0,1\}^{m \times d^2}$ be a random sampling matrix in which each row contains exactly one uniformly distributed nonzero element which has value one, and let $H$ be a ${d}\times{d}$ Hadamard matrix. Then, the {\sf TensorSRHT} is defined to be $S: \RR^{{d}} \times \RR^{{d}} \rightarrow \RR^m$ given by $S = \frac{1}{\sqrt{m}} P \cdot \left( H D_1 \times H D_2 \right)$.
	
{\bf Remark:} $S(x^{\otimes 2})$ can be computed in time $O(d \log d + m)$ using the FFT algorithm.
\end{defn}

Another sketch which is particularly efficient for sketching sparse vectors with high probability is the {\sf OSNAP} transform \cite{nelson2013osnap}, defined as follows.

\begin{defn}[{\sf OSNAP} transform]\label{def:osnap}
	For every sparsity parameter $s$, target dimension $m$, and positive integer $d$, the {\sf OSNAP} transform with sparsity parameter $s$ is defined as,
	
	$$S_{r,j}= \sqrt{\frac{1}{s}} \cdot \delta_{r,j} \cdot \sigma_{r,j},$$
	for all $r \in [m]$ and all $j \in [d]$, where $\sigma_{r,j} \in \{-1, +1\}$ are independent and uniform Rademacher random variables and $\delta_{r,j}$ are Bernoulli random variables satisfying,
	\begin{enumerate}
	\item For every $i \in [d]$, $\sum_{r \in [m]} \delta_{r,i} = s$. That is, each column of $S$ has exactly $s$ non-zero entries.
	\item  For all $r \in [m]$ and all $i \in [d]$, $\Ep{\delta_{r,i}} = s/m$.
	\item The $\delta_{r,i}$'s are negatively correlated: $\forall T \subset [m] \times [d]$, $\Ep{\prod_{(r,i) \in T} \delta_{r,i}} \le \prod_{(r,i) \in T} \Ep{ \delta_{r,i}} = (\frac{s}{m})^{|T|}$.
	\end{enumerate}
\end{defn}

\section{Linear Dependence on the Tensoring Degree \texorpdfstring{$p$}{p}}\label{sec:moment-bounds}

There are various desirable properties that we would like a linear sketch to satisfy.
One such property which is central to our main results is the \emph{JL Moment Property}. 
In this section we prove Theorem \ref{const-prob-thrm} and Theorem \ref{thm:high-prob-moment} by propagating the \emph{JL Moment Property} through our recursive construction from Section \ref{sec:sektch-construction}.
The \emph{JL Moment Property} captures a bound on the moments
of the difference between the Euclidean norm of a vector and its Euclidean norm after applying
the sketch on it. The JL Moment Property proves to be a powerful property for a sketch and we will show that it implies the Oblivious Subspace Embedding as well as the Approximate Matrix Product property for linear sketches.

In \cref{sec:second-moment} we choose $\sbase$ and $\tbase$ to be {\sf TensorSketch} and {\sf CountSketch} respectively. Then we propagate the second JL Moment through the sketch construction $\Pi^p$ and thereby prove Theorem \ref{const-prob-thrm}. In \cref{sec:higher-moments} we propagate the higher JL Moments through our recursive construction $\Pi^p$ as per Definition \ref{def:sketch} with {\sf TensorSRHT} at the internal nodes ($\sbase$) and {\sf OSNAP} at the leaves ($\tbase$), thereby proving Theorem \ref{thm:high-prob-moment}.

To make the notation less heavy we will use $\PNorm{X}{t}$ for the $t^{\text{th}}$
moment of a random variable $X$. This is formally defined below.
\begin{defn}
	For every integer $t \ge 1$ and any random variable $X\in \RR$, we write
	$$\PNorm{X}{t} = \left( E \left[ |X|^t \right] \right)^{1/t}.$$
	Note that $\PNorm{X+Y}{t} \le \PNorm{X}{t} + \PNorm{Y}{t}$ for any random variables
	$X, Y$ by the Minkowski Inequality.
\end{defn}

We now formally define the JL Moment Property of sketches.
\begin{defn}[JL Moment Property] \label{def:jl-moment}
	For every positive integer $t$ and every $\delta, \eps \ge 0$,
	we say a distribution over random matrices $S \in \RR^{m\times d}$ has the
	$(\epsilon, \delta, t)$-JL-moment property, when
	\begin{align*}
		\PNorm{\|Sx\|_2^2 - 1}{t} \le \epsilon \delta^{1/t} \quad\text{and}\quad \Ep{\|Sx\|_2^2} = 1
	\end{align*}
	for all $x\in \RR^d$ such that $\|x\|=1$.
\end{defn}

The JL Moment Property directly implies the following moment bound for the inner product of two vectors:
\begin{lem}[Two vector JL Moment Property]
	\label{lem:bamm}
	For any $x,y\in \RR^d$, if $S$ has the $(\eps,\delta, t)$-JL Moment Property, then
	\begin{align}
		\PNorm{(Sx)^\top(Sy) - x^\top y}{t} \le \eps\delta^{1/t}\|x\|_2\|y\|_2.
	\label{eq:better-mat-mul}
	\end{align}
\end{lem}
\begin{proof}
	We can assume by linearity of the norms that $\|x\|_2=\|y\|_2=1$.
	We then use that $\|x-y\|^2_2 = \|x\|^2_2 + \|y\|^2_2 - 2x^\top y$
	and $\|x+y\|^2_2 = \|x\|^2_2 + \|y\|^2_2 + 2x^\top y$
	such that $x^\top y = (\|x+y\|^2_2-\|x-y\|^2_2)/4$.
	Plugging this into the left hand side of \eqref{eq:better-mat-mul} gives
	\begin{align*}
	\PNorm{(Sx)^\top(Sy) - x^\top y}{t}
	&=
	\PNorm{\|Sx+Sy\|^2_2 - \|x+y\|^2_2 - \|Sx-Sy\|^2_2 + \|x-y\|^2_2 }{t}/4
	\\&\le
	\left(\PNorm{\|S(x+y)\|^2_2 - \|x+y\|^2_2}{t} + \PNorm{\|S(x-y)\|^2_2 - \|x-y\|^2_2}{t}\right)/4
	\\&\le
	\eps\delta^{1/t} (\|x+y\|^2_2+\|x-y\|^2_2)/4 \quad\text{(JL moment property)}
	\\&=
	\eps\delta^{1/t} (\|x\|^2_2+\|y\|^2_2)/2
        \\&= \eps\delta^{1/t}.
	\end{align*}
\end{proof}

We will also need the Strong JL Moment Property, which is a sub-Gaussian bound on
the difference between the Euclidean norm of a vector
and its Euclidean norm after applying the sketch on it.
\begin{defn}[Strong JL Moment Property]\label{defn:strong-moment-property}
    For every $\eps,\delta>0$ we say a distribution over random matrices
    $M \in \RR^{m \times d}$ has the Strong $(\eps, \delta)$-JL
    Moment Property when
    \begin{align*}
        \PNorm{\norm{Mx}_{2}^2 - 1}{t}
			\le \frac{\eps}{e}\sqrt{\frac{t}{\log(1/\delta)}}
		\quad\text{and}\quad \Ep{\norm{Mx}_2^2} = 1
		\; ,
    \end{align*}
    for all $x \in \RR^d$, $\norm{x}_{2} = 1$ and every integer
    $t \le \log(1/\delta)$.
\end{defn}
\begin{rem}
	It should be noted that if a matrix $M \in \RR^{m \times d}$
	has the Strong $(\eps, \delta)$-JL Moment Property then it has
	the $(\eps, \delta, \log(1/\delta))$-JL Moment Property, since
	\[
		\PNorm{\norm{Mx}_{2}^2 - 1}{\log(1/\delta)}
			\le \frac{\eps}{e}\sqrt{\frac{\log(1/\delta)}{\log(1/\delta)}}
			= \frac{\eps}{e}
			= \eps \delta^{1/\log(\frac{1}{\delta})}
		\; .
	\]
\end{rem}

The following two lemmas together show that if we want to prove that $\Pi^p$ is
an Oblivious Subspace Embedding and that $\Pi^p$ has the Approximate Matrix Multiplication Property,
then it suffices to prove that $\Pi^q$ has the JL Moment Property, for $q$ which is the smallest power of two integer such that $q \ge p$, as in \Cref{def:sketch}. This reduction will be the main component of
the proofs of \Cref{const-prob-thrm} and \Cref{thm:high-prob-moment}.
\begin{lem}\label{lem:reduce-p-to-q}
	For every positive integers $n, p, d$, every $\eps, \delta \in [0, 1]$, and every $\mu \ge 0$.
	Let $q = 2^{\lceil\log_2(p)\rceil}$ and let $\Pi^p \in \RR^{m \times d^p}$ and
	$\Pi^q \in \RR^{m \times d^q}$ be defined as in \Cref{def:sketch}, for some base
	sketches $\sbase \in \RR^{m \times m^2}$ and $\tbase \in \RR^{d \times d}$.

	If $\Pi^q$ is an $(\eps, \delta, \mu, d^q, n)$-Oblivious Subspace Embedding then
	$\Pi^p$ is an $(\eps, \delta, \mu, d^p, n)$-Oblivious Subspace Embedding. Also if
	$\Pi^q$ has the $(\eps, \delta)$-Approximate Matrix Multiplication Property then
	$\Pi^p$ has the $(\eps, \delta)$-Approximate Matrix Multiplication Property.
\end{lem}
\begin{proof}
	We will prove a correspondence between $\Pi^p$ and $\Pi^q$. Let $E_1 \in \RR^{d \times n}$
	be a matrix whose first row is equal to one and is zero everywhere else. By \Cref{def:sketch}
	we have that for any matrix $A \in \RR^{d^p \times n}$ that
	$\Pi^p A = \Pi^q (A \otimes E_1^{\otimes (q - p)})$. A simple calculation shows that for
	any matrices $A, B \in \RR^{d^p \times n}$ then
	\begin{align*}
		(A \otimes E_1^{\otimes (q - p)})^\top (B \otimes E_1^{\otimes (q - p)})
			= A^\top B \circ (E_1^{\otimes (q - p)})^\top E_1^{\otimes (q - p)}
			= A^\top B
		\; ,
	\end{align*}
	where $\circ$ denotes the Hadamard product, and the last equality follows since
	$(E_1^{\otimes (q - p)})^\top E_1^{\otimes (q - p)}$ is an all ones matrix. This implies
	that $\norm{A \otimes E_1^{\otimes (q - p)}}_{F} = \norm{A}_{F}$ and
	$s_\lambda((A \otimes E_1^{\otimes (q - p)})^\top A \otimes E_1^{\otimes (q - p)}) = s_\lambda(A^\top A)$.

	Now assume that $\Pi^q$ is an $(\eps, \delta, \mu, n)$-Oblivious
	Subspace Embedding, and let $A \in \RR^{d^p \times n}$ and $\lambda \ge 0$ be such that $s_\lambda(A) \le \mu$.
	Define $A' = A \otimes E_1^{\otimes (q - p)}$, then
	\begin{align*}
		&\Prp{(1 - \eps)(A^\top A + \lambda I_n) \preceq (\Pi^p A)^\top \Pi^p A + \lambda I_n \preceq (1 + \eps) (A^\top A + \lambda I_n)}
			\\&\quad= \Prp{(1 - \eps)(A'^\top A' + \lambda I_n) \preceq (\Pi^q A')^\top \Pi^q A' + \lambda I_n \preceq (1 + \eps) (A'^\top A' + \lambda I_n)}
			\\&\quad\ge 1 - \delta
		\; ,
	\end{align*}
	where we have used that $s_\lambda(A'\top A') = s_\lambda(A^\top A) \le \mu$. This shows that $\Pi^p$ is an
	$(\eps, \delta, \mu, n)$-Oblivious Subspace Embedding.

	Assume that $\Pi^q$ has $(\eps, \delta)$-Approximate Matrix Multiplication Property, and
	let $C, D \in \RR^{d^p \times n}$. Define $C' = C \otimes E_1^{\otimes (q - p)}$ and
	$D' = D \otimes E_1^{\otimes (q - p)}$, then 
	\begin{align*}
		\Prp{\norm{(\Pi^p C)^\top \Pi^p D - C^\top D}_{F} \ge \eps \norm{C}_{F} \norm{D}_{F}}
			&= \Prp{\norm{(\Pi^q C')^\top \Pi^q D' - C'^\top D'}_{F} \ge \eps \norm{C'}_{F} \norm{D'}_{F}}
			\\&\le \delta
		\; ,
	\end{align*}
	where we have used that $\norm{C'}_{F} = \norm{C}_{F}$, $\norm{D'}_{F} = \norm{D}_{F}$,
	and $C'^\top D' = C^\top D$. This show that $\Pi^p$ has
	$(\eps, \delta)$-Approximate Matrix Multiplication Property.
\end{proof}

\begin{lem}\label{lem:jl-moment-ose}
	For any $\eps, \delta \in [0, 1]$, $t \ge 1$,
	if $M \in \RR^{m \times d}$ is a random matrix with $(\eps, \delta, t)$-JL Moment Property
	then $M$ has the $(\eps, \delta)$-Approximate Matrix Multiplication Property.
	
	Furthermore, for any $\mu > 0$, if $M \in \RR^{m \times d}$
	is a random matrix with $(\eps/\mu, \delta, t)$-JL Moment Property then for every positive
	integer $n \in \ZZ$, $M$ is a $(\eps, \delta, \mu, d, n)$-OSE.
\end{lem}
\begin{proof}
	~\paragraph{Approximate Matrix Multiplication}
	Let $C, D \in \RR^{d \times n}$. We will prove that
	\begin{align}\label{eq:amm-moment}
		\PNorm{\norm{(M C)^\top M D - C^\top D}_{F}}{t}
			\le \eps \delta^{1/t} \norm{C}_{F} \norm{D}_{F}
		\; .		
	\end{align}
	Then Markov's inequality will give us the result. Using the triangle inequality
	together with \Cref{lem:bamm} we get that:
	\begin{align*}
		\PNorm{\norm{(M C)^\top M D - C^\top D}_{F}}{t}
			&= \PNorm{\norm{(M C)^\top M D - C^\top D}_{F}^2}{t/2}^{1/2}
			\\&= \PNorm{\sum_{i, j \in [n]} \left((M C_i)^\top M D_j - C_i^\top D_j\right)^2}{t/2}^{1/2}
			\\&\le \sqrt{\sum_{i, j \in [n]} \PNorm{(M C_i)^\top M D_j - C_i^\top D_j}{t}^2}
			\\&\le \sqrt{\sum_{i, j \in [n]} \eps^2 \delta^{2/t} \norm{C_i}_{2}^2 \norm{D_j}_2^2 }
			\\&= \eps \delta^{1/t} \norm{C}_{F} \norm{D}_{F}
		\; .
	\end{align*}
	Using Markov's inequality we now get that
	\begin{align*}
		\Prp{\norm{(M C)^\top M D - C^\top D}_{F} \ge \eps \norm{C}_{F} \norm{D}_{F}}
			\le \frac{\PNorm{\norm{(M C)^\top M D - C^\top D}_{F}}{t}^t}{\eps^t \norm{C}_{F}^t \norm{D}_{F}^t}
			\le \delta
		\; .
	\end{align*}

	~\paragraph{Oblivious Subspace Embedding.}
	We will prove that for any $\lambda \ge 0$ and any matrix $A \in \RR^{d \times n}$, 
	\begin{align}\label{eq:ose-spectral}
		(1 - \eps)(A^\top A + \lambda I_n) \preceq (M A)^\top M A + \lambda I_n \preceq (1 + \eps) (A^\top A + \lambda I_n)
		\; ,		
	\end{align}
	holds with probability at least $1 - \left(\frac{s_\lambda(A^\top A)}{\mu}\right)^t\delta$, which will imply our result.

	We will first consider $\lambda > 0$. Then $A^\top A + \lambda I_n$ is positive definite.
	Thus, by left and right multiplying \eqref{eq:ose-spectral} by $(A^\top A + \lambda I_n)^{-1/2}$,
	we see that \eqref{eq:ose-spectral} is equivalent to
	\[
		(1 - \eps)I_n \preceq \left(M A (A^\top A + \lambda I_n)^{-1/2} \right)^\top M A (A^\top A + \lambda I_n)^{-1/2} + \lambda (A^\top A + \lambda I_n)^{-1} \preceq (1 + \eps)I_n
		\; .
	\]
	which, in turn, is implied by the following:
	\[
		\left\| \left(M A (A^\top A + \lambda I_n)^{-1/2} \right)^\top M A (A^\top A + \lambda I_n)^{-1/2} + \lambda (A^\top A + \lambda I_n)^{-1} - I_n \right\|_{op} \le \eps
		\; .
	\]
	Note that $(A^\top A + \lambda I_n)^{-1/2}A^\top A (A^\top A + \lambda I_n)^{-1/2} = I_n - \lambda (A^\top A + \lambda I_n)^{-1}$.
	Letting $Z = A (A^\top A + \lambda I_n)^{-1/2}$, we note that it suffices to establish,
	\begin{equation*}
		\left\| \left(M Z\right)^\top M Z - Z^\top Z \right\|_{op} \le \eps \; .
	\end{equation*}
	Using \eqref{eq:amm-moment} together with Markov's inequality we get that
	\[
		\Prp{\left\| \left(M Z\right)^\top M Z - Z^\top Z \right\|_{op} \ge \eps}
			\le \Prp{\left\| \left(M Z\right)^\top M Z - Z^\top Z \right\|_F \ge \eps}
			\le \left(\frac{\norm{Z}_{F}^2}{\mu}\right)^t \delta
			= \left(\frac{s_\lambda(A^\top A)}{\mu}\right)^t \delta
		\; ,
	\]
	where the last equality follows from 
	\begin{align*}
		\norm{Z}_{F}^2
			&= \tr\left(Z^\top Z\right)
			\\&= \tr\left(\left(A (A^\top A + \lambda I_n)^{-1/2}\right)^\top A (A^\top A + \lambda I_n)^{-1/2}\right)
			\\&= \tr\left(A^\top A (A^\top A + \lambda I_n)^{-1}\right)
			\\&= s_\lambda(A^\top A)
		\; .
	\end{align*}

	To prove the result for $\lambda = 0$ we will use Fatou's lemma.
	\begin{align*}
		&\Prp{\left((1 - \eps)A^\top A \preceq (M A)^\top M A \preceq (1 + \eps) A^\top A\right)^C}
			\\&\quad\le \liminf_{\lambda \to 0^+} \Prp{\left((1 - \eps)(A^\top A + \lambda I_n) \preceq (M A)^\top M A + \lambda I_n \preceq (1 + \eps) (A^\top A + \lambda I_n)\right)^C}
			\\&\quad\le \liminf_{\lambda \to 0^+} \frac{s_\lambda(A^\top A)}{\mu} \delta
			\\&\quad= \frac{s_0(A^\top A)}{\mu} \delta
		\; ,
	\end{align*}
	where the last equality follows from continuity of $\lambda \mapsto s_\lambda(A^\top A)$.
\end{proof}

Our next important observation is that $\Pi^q$ can be written as the product
of $2q - 1$ independent random matrices, which all have a special structure
which makes them easy to analyse.
\begin{lem}\label{lem:pi^q-decomposition}
	For any integer $q$ which is a power of two, $\Pi^q: \RR^{m^q} \to \RR^m$ be defined as in
    \Cref{def:sketch} for some base sketches
	$\sbase: \RR^{m^2}\to \RR^m$ and $\tbase: \RR^d \to \RR^m$.
	Then there exist matrices $(M^{(i)})_{i \in [q - 1]}, (M'^{(j)})_{j \in [q]}$
	and integers $(k_i)_{i \in [q - 1]}, (k'_1)_{i \in [q - 1]}, (l_j)_{j \in [q]}, (l'_j)_{j \in [q]}$,
	such that,
	\[
		\Pi^q = M^{(q - 1)} \cdot \ldots M^{(1)} \cdot M'^{(q)} \cdot \ldots \cdot M'^{(1)} \; ,
	\]
	and $M^{(i)} = I_{k_i} \times \sbase^{(i)} \times I_{k'_i}$,
	$M'^{(j)} = I_{l_j} \times \tbase^{(j)} \times I_{l'_j}$,
	where $\sbase^{(i)}$ and $\tbase^{(j)}$ are independent instances
	of $\sbase$ and $\tbase$,
	for every $i \in [q - 1]$, $j \in [q]$. 
\end{lem}
\begin{proof}
	We have that $\Pi^q = Q^q T^q$ by \Cref{def:sketch}. By \Cref{def:sketch-tilde}
	we have that $Q^q = S^2 S^4 \cdots S^q$. \Cref{claim-tensor2product-reduction} shows that for
    every $l \in \{ 2, 4, \cdots q \}$ we can write,
    \begin{equation}\label{eq:s-matrix-decomposition}
        S^l = M^l_{l/2} M^l_{l/2-1} \cdot \ldots \cdot M^l_1, 
    \end{equation}
    where $M^l_j = I_{m^{l-2j}} \times S^{l}_{l/2 - j +1} \times I_{m^{j-1}}$ for every
	$j \in [l/2]$. From the discussion in \Cref{def-dim-reduction} it follows that,
    \begin{equation}
        T^q = M'^{(q)} \cdot \ldots \cdot M'^{(1)}, \label{eq:t-matrix-decomposition}
    \end{equation}
    where $M'^{(j)} = I_{d^{q-j}} \times T_{q-j+1} \times I_{m^{j-1}}$ for every $j \in [q]$.
    Therefore by combining \eqref{eq:s-matrix-decomposition} and
	\eqref{eq:t-matrix-decomposition} we get the result.
\end{proof}

We want to show that $I_{k} \times M \times I_{k'}$ inherits the JL properties
of $M$. The following simple fact does just that. 
\begin{lem}
	Let $t \in \mathbb{N}$ and $\alpha \ge 0$. If
	$P \in \RR^{m_1 \times d_1}$ and $Q \in \RR^{m_2 \times d_2}$
	are two random matrices (not necessarily independent), such
	that,
	\begin{align*}
		\PNorm{\norm{Px}_{2}^2-\|x\|_2^2}{t} \le \alpha \norm{x}_{2}^2 \quad\text{and}\quad \Ep{\| Px\|_2^2} = \|x\|_2^2
		\; ,\\
		\PNorm{\norm{Qy}_{2}^2-\|y\|_2^2}{t} \le \alpha \norm{y}_{2}^2 \quad\text{and}\quad \Ep{\| Qy\|_2^2} = \|y\|_2^2
		\; ,
	\end{align*}
	for any vectors $x \in \RR^{d_1}$ and $y \in \RR^{d_2}$, then
	\[
		\PNorm{ \norm{(P \oplus Q)z}_2^2 - \|z\|_2^2 }{t} \le \alpha \norm{z}_{2}^2 \quad\text{and}\quad \Ep{\| (P \oplus Q)z\|_2^2} = \|z\|_2^2
		\; ,
	\]
	for any vector $z \in \RR^{d_1 + d_2}$.
\end{lem}
\begin{proof}
	Let $z \in \RR^{d_1 + d_2}$ and choose $x \in \RR^{d_1}$ and
	$y \in \RR^{d_2}$, such that, $z = x \oplus y$.
	Using the triangle inequality,
	\begin{align*}
	\PNorm{\norm{(P \oplus Q)z}_{2}^2 - \norm{z}_{2}^2}{t}
	&= \PNorm{\norm{Px}_{2}^2 + \norm{Qy}_{2}^2
		- \norm{x}_{2}^2 - \norm{y}_{2}^2}{t}
	\\&\le \PNorm{\norm{Px}_{2}^2 - \norm{x}_{2}^2}{t}
	+ \PNorm{\norm{Qy}_{2}^2 - \norm{y}_{2}^2}{t}
	\\&\le \alpha \norm{x}_{2}^2 + \alpha \norm{y}_{2}^2 \\&= \alpha\norm{z}_{2}^2.
	\end{align*}
	We also see that
	\[
		\Ep{\| (P \oplus Q)z\|_2^2}
			= \Ep{\|Px\|_2^2} + \Ep{\|Qy\|_2^2}
			= \|x\|_2^2 + \|y\|_2^2
			= \|z\|_2^2.
	\]
\end{proof}
An easy consequence of this lemma is that for any matrix,
$S$, with the $(\eps, \delta, t)$-JL Moment Property,
$I_k \times S$ has the $(\eps, \delta, t)$-JL Moment Property.
This follows simply from
$
I_k \times S =
\underbrace{S \oplus S \oplus \ldots \oplus S}
_{k\text{ times}}
$. 
Similarly, $S \times I_k$ has the $(\eps, \delta, t)$-JL Moment
Property, since $S \times I_k$ is just a reordering of the rows
of $I_k \times S$, which trivially does not affect the JL Moment Property.
The same arguments show that if $S$ has the Strong $(\eps, \delta)$-JL Moment Property
then $I_k \times S$ and $S \times I_k$ has the Strong $(\eps, \delta)$-JL Moment Property.
So we conclude the following

\begin{lem}\label{claim:jl-moment-tensor-identity}
	If the matrix $S$ has the $(\eps, \delta, t)$-JL Moment Property, then for any positive integers $k, k'$,
	the matrix $M = I_{k} \times S \times I_{k'}$ has the 
	$(\eps, \delta, t)$-JL Moment Property.

	Similarly, if the matrix $S$ has the Strong $(\eps, \delta)$-JL Moment Property,
	then for any positive integers $k, k'$, the matrix $M = I_{k} \times S \times I_{k'}$ has
	the Strong $(\eps, \delta)$-JL Moment Property.
\end{lem}

Now if we can prove that the product of matrices with the JL Moment Property has the JL Moment Property,
then \Cref{claim:jl-moment-tensor-identity} and \Cref{lem:pi^q-decomposition}
would imply that $\Pi^q$ has the JL Moment Property, which again implies that $\Pi^p$ is an
Oblivious Subspace Embedding and has the Approximate Matrix Multiplication Property, by
\Cref{lem:jl-moment-ose} and \Cref{lem:reduce-p-to-q}. This is exactly what we will do:
in \Cref{sec:second-moment} we prove that the product of $k$ independent matrices with the 
$\left(\frac{\eps}{\sqrt{2k}},\delta,2\right)$-JL Moment Property results in a matrix with the 
$(\eps, \delta, 2)$-JL Moment Property, which will give us the proof of \Cref{const-prob-thrm},
and in \Cref{sec:higher-moments} we prove that the product of $k$ independent matrices
with the Strong $\left(O\left(\frac{\eps}{\sqrt{k}}\right), \delta\right)$-JL Moment Property results
in a matrix with the Strong $(\eps, \delta)$-JL Moment Property, which will give us the proof
of \Cref{thm:high-prob-moment}.

\subsection{Second Moment of \texorpdfstring{${\Pi}^q$ (analysis for $\tbase:$ {\sf CountSketch} and $\sbase:$ {\sf TensorSketch})}{TensorSketch}}\label{sec:second-moment}
In this section we prove Theorem \ref{const-prob-thrm} by instantiating our recursive construction from Section \ref{sec:sektch-construction} with {\sf CountSketch} at the leaves and {\sf TensorSketch} at the internal nodes of the tree.
The proof proceeds by showing the second moment property -- i.e., $(\eps,\delta,2)$-JL Moment Property, for our recursive construction.
We prove that our sketch $\Pi^q$ satisfies the $(\eps,\delta,2)$-JL Moment Property as per Definition \ref{def:jl-moment} as long as the base sketches $\sbase, \tbase$ are chosen from a distribution which satisfies the second moment property. We show that this is the case for {\sf CountSketch} and {\sf TensorSketch}.

\Cref{claim:jl-moment-tensor-identity}
together with \Cref{lem:pi^q-decomposition} show that if the base sketches
$\sbase, \tbase$ have the JL Moment Property then $\Pi^q$ is the product of $2q - 1$
independent random matrices with the JL Moment Property. Therefore, understanding how matrices
with the JL Moment Property compose is crucial.
The following lemma shows that composing independent random matrices which have the
JL Moment Property results in matrix which has the JL Moment Property with a small loss in the parameters.

\begin{lem}[Composition lemma for the second moment]\label{lem:second-moment-composition}
For any $\eps , \delta \ge 0$ and any integer $k$ if $M^{(1)} \in \RR^{d_2 \times d_1}, \cdots M^{(k)} \in \RR^{d_{k+1} \times d_k}$ are independent random matrices with the $\left( \frac{\eps}{\sqrt{2k}}, \delta, 2 \right)$-JL-moment property then the product matrix $M = M^{(k)} \cdots M^{(1)}$ satisfies the $(\eps, \delta, 2)$-JL-moment property.
\end{lem}
\begin{proof}
Let $x \in \RR^{d_1}$ be a fixed unit norm vector.
We note that for any $i \in [k]$ we have that
    \begin{align}
\Epcond{\norm{M^{(i)} \cdot \ldots \cdot M^{(1)} x}_{2}^2}{M^{(1)}, \ldots, M^{(i - 1)}}
= \norm{M^{(i - 1)} \cdot \ldots \cdot M^{(1)} x}_{2}^2 \label{eq:jl-prod-ep}
\; .
\end{align}
Now we will prove by induction on $i \in [k]$ that,
\begin{align}
\Var\left[{\norm{M^{(i)} \cdot \ldots \cdot M^{(1)} x}_{2}^2}\right] 
\le \left( 1 + \frac{\eps^2\delta}{2k} \right)^i - 1 \label{eq:jl-prod-induc}.
\end{align}
For $i = 1$ the result follows from the fact that $M^{(1)}$
has the $(\eps/\sqrt{2k}, \delta, 2)$-JL moment property. Now assume
that \eqref{eq:jl-prod-induc} is true for $i - 1$. By the law
of total variance we get that
\begin{align}
\Var\left[{\norm{M^{(i)} \cdot \ldots \cdot M^{(1)} x}_{2}^2}\right] &= \mathbb{E}\left[{\Varcond{\norm{M^{(i)} \cdot \ldots \cdot M^{(1)} x}_{2}^2}{M^{(1)}, \ldots, M^{(i - 1)}}}\right]\nonumber
\\&\quad+ \Var\left[{\Epcond{\norm{M^{(i)} \cdot \ldots \cdot M^{(1)} x}_{2}^2}{M^{(1)}, \ldots, M^{(i - 1)}}}\right]\label{totalvar:lem12}
\end{align}
Using \eqref{eq:jl-prod-ep} and the induction hypothesis we get that,
\begin{align}
\Var\left[{\Epcond{\norm{M^{(i)} \cdot \ldots \cdot M^{(1)} x}_{2}^2}{M^{(1)}, \ldots, M^{(i - 1)}}}\right]
&= \Var\left[{\norm{M^{(i - 1)} \cdot \ldots M^{(1)} x}_{2}^2}\right]\nonumber
\\&\le \left( 1 + \frac{\eps^2\delta}{2k} \right)^{i-1} - 1. \label{eq:varexp:lem12}
\end{align}
Using that $M^{(i)}$ has the $(\eps/\sqrt{2k}, \delta, 2)$-JL moment property,
\eqref{eq:jl-prod-ep}, and the induction hypothesis we get that,
\begin{align}
&\mathbb{E}\left[{\Varcond{\left\|{M^{(i)} \cdot \ldots \cdot M^{(1)} x} \right\|_{2}^2}{M^{(1)}, \ldots, M^{(i - 1)}}}\right]\nonumber
\\&\quad\le \mathbb{E}\left[{\frac{\eps^2}{2k} \delta \left\|{M^{(i - 1)} \cdot \ldots M^{(1)} x}\right\|_{2}^4}\right]\nonumber
\\&\quad= \frac{\eps^2 \delta}{2k}\left(
\Var\left[{\left\|{M^{(i - 1)} \cdot \ldots M^{(1)} x}\right\|_{2}^2}\right]
+ \mathbb{E}\left[{\left\|{M^{(i - 1)} \cdot \ldots M^{(1)} x}\right\|_{2}^2}\right]^2
\right)\nonumber
\\&\quad\le \frac{\eps^2\delta}{2k} \left( \left(1 + \frac{\eps^2\delta}{2k} \right)^{i-1} - 1 + 1\right) = \frac{\eps^2\delta}{2k} \left(1 + \frac{\eps^2\delta}{2k} \right)^{i-1}.\label{expvar:lem12}
\end{align}
Plugging \eqref{eq:varexp:lem12} and \eqref{expvar:lem12} into \eqref{totalvar:lem12} gives,
\[
\Var\left[{\left\|{M^{(i)} \cdot \ldots \cdot M^{(1)} x}\right\|_{2}^2}\right]
\le \frac{\eps^2\delta}{2k} \left(1 + \frac{\eps^2\delta}{2k} \right)^{i-1} + \left(1 + \frac{\eps^2\delta}{2k} \right)^{i-1} - 1
= \left(1 + \frac{\eps^2\delta}{2k} \right)^{i} -1 \; .
\]
Hence,
\begin{align*}
\Var\left[{\norm{Mx}_{2}^2}\right]
\le \left(1 + \frac{\eps^2\delta}{2k} \right)^{k} - 1
\le \exp(\eps^2 \delta/2) - 1
\le \eps^2 \delta
\; ,
\end{align*}
which proves that $M$ has the $(\eps, \delta, 2)$-JL moment property.
\end{proof}

Equipped with the composition lemma for the second moment, we now establish the second moment property for our recursive sketch $\Pi^q$:
\begin{cor}[Second moment property for $\Pi^q$]\label{cor:second-moment-sketch}
For any power of two integer $q$ let $\Pi^q: \RR^{m^q} \to \RR^m$ be defined as in Definition \ref{def:sketch}, where both of the common distributions $\sbase: \RR^{m^2}\to \RR^m$ and $\tbase: \RR^d \to \RR^m$, satisfy the $\left( \frac{\eps}{\sqrt{4q+2}}, \delta, 2 \right)$-JL-moment property. Then it follows that $\Pi^q$ satisfies the $(\eps, \delta, 2)$-JL-moment property.
\end{cor}
\begin{proof}
This follows from \Cref{lem:pi^q-decomposition} and \Cref{lem:second-moment-composition}.

\end{proof}

Now we are ready to prove \Cref{const-prob-thrm}.
Recall that $k(x,y) = \langle x, y \rangle^q$ is the polynomial kernel of degree $q$.
One can see that $k(x,y) = \langle x^{\otimes q} , y^{\otimes q} \rangle$. Let $x_1, x_2, \cdots x_n \in \RR^m$ be an arbitrary dataset of $n$ points in $\RR^m$. We represent the data points by matrix $X \in \RR^{m \times n}$ whose $i^{\text{th}}$ column is the vector $x_i$. Let $A \in \RR^{m^q \times n}$ be the matrix whose $i^{\text{th}}$ column is $x_i^{\otimes q}$ for every $i \in [n]$. For any regularization parameter $\lambda>0$, the statistical dimension of $A^\top A$ is defined as $s_\lambda := \mathbf{tr}\left( (A^\top A) (A^\top A + \lambda I_n )^{-1} \right)$.

\constprobthrm*

\begin{proof} 
	Throughout the proof, let $\delta = \frac{1}{10}$ denote the failure probability, let $q = 2^{\lceil \log_2 p \rceil}$, and let $e_1 \in \RR^d$ be the column vector with a $1$ in the first coordinate and zeros elsewhere. Let $\Pi^p \in \RR^{m\times d^p}$ be the sketch defined in Definition \ref{def:sketch}, where the base distributions $\sbase \in \RR^{m\times m^2}$ and $\tbase \in \RR^{m \times d}$ are respectively the standard {\sf TensorSketch} of degree two and standard {\sf CountSketch}.
	It is shown in \cite{avron2014subspace} and \cite{clarkson2017low} that for these choices of base sketches, $\sbase$ and $\tbase$ are both unbiased and satisfy the $\left( \frac{\eps}{\sqrt{4q+2}}, \delta, 2 \right)$-JL-moment property as long as $m = \Omega(\frac{q}{\eps^2 \delta})$ (see Definition~\ref{def:jl-moment}).
	
    \paragraph{Oblivious Subspace Embedding}
    Let $m = \Omega \left(\frac{q s_\lambda^2}{\delta \epsilon^2}\right)$ be an integer. 
    Then $\sbase$ and $\tbase$ has the $\left(\frac{\eps}{\sqrt{4q + 2} s_\lambda}, \delta, 2 \right)$-JL
    Moment Property. Thus using \Cref{cor:second-moment-sketch} we conclude that $\Pi^q$ has the 
    $\left(\frac{\eps}{s_\lambda}, \delta, 2\right)$-JL Moment Property.
    Thus, \Cref{lem:jl-moment-ose} implies that $\Pi^q$ is an
    $(\eps, \delta, s_\lambda, d^q, n)$-Oblivious Subspace Embedding,
    and by \Cref{lem:reduce-p-to-q} we get that $\Pi^p$ is an
    $(\eps, \delta, s_\lambda, d^p, n)$-Oblivious Subspace Embedding.

	\paragraph{Approximate Matrix Multiplication.}
	Let $m = \Omega \left(\frac{q}{\delta \epsilon^2}\right)$. 
	Then $\sbase$ and $\tbase$ have the $\left(\frac{\eps}{\sqrt{4q + 2}}, \delta, 2 \right)$-JL
    Moment Property. Thus, using \Cref{cor:second-moment-sketch} we conclude that $\Pi^q$ has the 
 	$\left(\eps, \delta, 2\right)$-JL Moment Property.
	Thus, \Cref{lem:jl-moment-ose} implies that $\Pi^q$ has the 
    $(\eps, \delta)$-Approximate Matrix Multiplication Property,
    and by \Cref{lem:reduce-p-to-q} we get that $\Pi^p$ has the 
    $(\eps, \delta)$-Approximate Matrix Multiplication Property.

	\paragraph{Runtime of Algorithm \ref{alg:main} when the base sketch $\sbase$ is {\sf TensorSketch} of degree two and $\tbase$ is {\sf CountSketch}:} We compute the time of running Algorithm \ref{alg:main} on a vector $x$.
	Computing $Y^0_j$ for each $j$ in lines~\ref{y-0-q} and \ref{y-0-p} of algorithm requires applying a {\sf CountSketch} on either $x$ or $e_1$ which takes time $O( \text{nnz}(x))$. Therefore computing all $Y^0_j$'s takes time $O(q \cdot \text{nnz}(x))$.
	
	Computing each of $Y^l_j$'s for $l \ge 1$ in line~\ref{Y-l-j} of Algorithm \ref{alg:main} amounts to applying a degree two {\sf TensorSketch} of input dimension $m^2$ and target dimension of $m$ on $Y^{l-1}_{2j-1} \otimes Y^{l-1}_{2j}$. This takes time $O( m \log m)$. Therefore computing $Y^l_j$ for all $l, j \ge 1$ takes time $O(q \cdot m \log m)$. Note that $q \le 2p$ and hence the total running time of Algorithm \ref{alg:main} on one vector $x$ is $O(p \cdot m\log_2m + p\cdot \text{nnz}(w))$. Sketching $n$ columns of a matrix $X \in \RR^{d \times n}$ takes time $O(p (nm\log_2m +  \text{nnz}(X)))$.
\end{proof}

\subsection{Higher Moments of \texorpdfstring{$\Pi^q$ (analysis for $\tbase:$ {\sf OSNAP} and $\sbase:$ {\sf TensorSRHT})}{TensorSketch}}\label{sec:higher-moments}
In this section we prove Theorem \ref{thm:high-prob-moment} by instantiating our recursive construction of Section \ref{sec:sektch-construction} with {\sf OSNAP} at the leaves and {\sf TensorSRHT} at the internal nodes.

The proof proceeds by showing the Strong JL Moment Property for our sketch $\Pi^q$.
If a sketch satisfies the Strong JL Moment Property then it straightforwardly is an OSE and has the approximate matrix product property. 
This section has two goals: first is to show that {\sf SRHT}, and {\sf TensorSRHT} as well as {\sf OSNAP} transform all satisfy the Strong JL Moment Property. The second goal of
this section is to prove that our sketch construction $\Pi^q$ inherits the strong JL moment property from the base sketches $\sbase, \tbase$.

In this section we will need Khintchine's inequality.
\begin{lem}[Khintchine's inequality~\cite{haagerup2007best}] \label{lem:khintchine}
	Let $t$ be a positive integer, $x\in \RR^d$, and $(\sigma_i)_{i\in[d]}$
	be independent Rademacher $\pm1$ random variables.
	Then
	\begin{align*}
	\PNorm{\sum_{i=1}^d \sigma_i x_i}{t} \,\le C_t\, \norm{x}_2,
	\end{align*}
	where $C_t\le \sqrt2 \left(\frac{\Gamma((t+1)/2)}{\sqrt\pi}\right)^{1/t}\le\sqrt{t}$ for all $t\ge 1$.
	
	One may replace $(\sigma_i)$ with an arbitrary independent sequence of random variables $(\varsigma_i)$ with $\Ep{\varsigma_i}=0$ and $\PNorm{\varsigma_i}{r} \le \sqrt{r}$ for any $1 \le r \le t$,
	and the lemma still holds up to a universal constant factor on the r.h.s.
\end{lem}

First we note that the {\sf OSNAP} transform satisfies the strong JL moment property.
\begin{lem}\label{lem:osnap-strong-moment}
	There exists a universal constant $L$, such that, the following holds.
	Let $M \in \RR^{m \times d}$ be a {\sf OSNAP} transform with sparsity parameter $s$.
	Let $x \in \RR^d$ be any vector with $\norm{x}_{2} = 1$ and $t \ge 1$, then
	\begin{align}\label{eq:sparse-jl-moment}
	\PNorm{\norm{Mx}_{2}^2 - 1}{t} \le L\left(\sqrt{\frac{t}{m}} + \frac{t}{s}\right)
	\; .
	\end{align}
	
	Setting $m = \Omega(\eps^{-2}\log(1/\delta))$ and $s = \Omega(\eps^{-1}\log(1/\delta))$
	then $M$ has the Strong $(\eps, \delta)$-JL Moment Property (\Cref{defn:strong-moment-property}).
\end{lem}
\begin{proof}
	The proof of \eqref{eq:sparse-jl-moment} follows from analysis in \cite{cohen18sparse}. They only prove it for $t = \log(1/\delta)$ but their proof
	is easily extended to the general case.
	
	Now if we set $m = 4L^2 e^2 \cdot \eps^{-2}\log(1/\delta)$
	and $s = 2L e \cdot \eps^{-1}\log(1/\delta)$
	then we get that
	\[
	\PNorm{\norm{Mx}_{2}^2 - 1}{t} \le L \sqrt{\frac{t}{4L^2 e^2 \cdot \eps^{-2}\log(1/\delta)}} + L\frac{t}{2L e \cdot \eps^{-1}\log(1/\delta)}
	\le \frac{\eps}{e}\sqrt{\frac{t}{\log(1/\delta)}}
	\; ,
	\]
	for every $1 \le t \le \log(1/\delta)$, which proves the result.
\end{proof}

We continue by proving that {\sf SRHT} and {\sf TensorSRHT} sketches satisfy the strong JL moment property. We will do this by
proving that a more general class of matrices satisfies the strong JL moment property. More precisely, let $k \in \ZZ_{> 0}$
be a positive integer and $(D^{(i)})_{i \in [k]} \in \prod_{i \in [k]} \RR^{d_i \times d_i}$ be independent matrices,
each with diagonal entries given by independent Rademacher variables. Let $d = \prod_{i \in [k]} d_i$,
and $P \in \{0, 1\}^{m \times d}$ be a random sampling matrix in which each row contains exactly one
uniformly distributed nonzero element which has value one. Then we will prove that the matrix
$M = \frac{1}{\sqrt{m}}PH(D_1 \times \ldots \times D_k)$ satisfies the strong JL moment property,
where $H$ is a $d \times d$ Hadamard matrix. If $k = 1$ then $M$ is just a {\sf SRHT}, and if $k = 2$ then $M$ is a {\sf TensorSRHT}.

In order to prove this result we need a couple of lemmas. The first lemma can be seen as a version of Khintchine's
inequality for higher order chaos.
\begin{lem}\label{lem:gen-khinchine}
	Let $t \ge 1$, $k \in \ZZ_{> 0}$, and $(\sigma^{(i)})_{i \in [k]} \in \prod_{i \in [k]} \RR^{d_i}$
	be independent vectors each satisfying the Khintchine inequality $\PNorm{\langle\sigma^{(i)},x\rangle}{t} \le C_t\norm{x}_{2}$ for $t \ge 1$ and any vector $x \in \RR^{d_i}$.
	Let $(a_{i_1, \ldots, i_{k}})_{i_1 \in [d_j], \ldots, i_k \in [d_k]}$
	be a tensor in $\RR^{d_1 \times \ldots \times d_k}$, then
	\begin{align*}
	\PNorm{
		\sum_{i_1\in[d_1],\dots,i_k\in[d_k]}
		\left(\prod_{j \in [k]} \sigma^{(j)}_{i_j}\right)
		a_{i_1, \ldots, i_{k}}
	}{t}
	\le C_t^k\left(\sum_{i_1\in[d_1],\ldots,i_k\in[d_k]} a_{i_1, \ldots, i_{k}}^2\right)^{1/2}
	\; ,
	\end{align*}
	for $t \ge 1$.
	Or, considering $a\in\RR^{d_1\cdots d_k}$ a vector, then simply
	$\PNorm{\langle \sigma^{(1)}\otimes\dots\otimes\sigma^{(k)}, a\rangle}{t}
	\le C_t^k\, \norm{a}_2$, for $t \ge 1$.
\end{lem}
This is related to Lata{\l}a's estimate for Gaussian chaoses~\cite{latala2006estimates}, but more simple in the case where $a$ is not assumed to have special structure.
Note that this implies the classical bound on the fourth moment of products of 4-wise independent hash
functions~\cite{DBLP:conf/stacs/BravermanCLMO10, indyk2008declaring, DBLP:journals/jacm/PatrascuT12}, since $C_4=3^{1/4}$
for Rademachers we have $\Ep{\langle \sigma^{(1)}\otimes\dots\otimes\sigma^{(k)}, x\rangle^4} \le 3^k \norm{x}_2^4$ for four-wise independent $(\sigma^{(i)})_{i \in [k]}$.

\begin{proof}
	The proof will be by induction on $k$. For $k = 1$ then the result is by assumption.
	So assume that the result is true for every value up to $k-1$.
	Let $B_{i_1,\dots,i_{k-1}} = \sum_{i_{k} \in [d_{k}]} \sigma^{(k)}_{i_{k}} a_{i_1, \ldots, i_{k}}$.
	We then pull it out of the left hand term in the theorem:
	\begin{align*}
	\PNorm{
		\sum_{i_1\in[d_1],\dots,i_k\in[d_k]}
		\left(\prod_{j \in [k]} \sigma^{(j)}_{i_j}\right)
		a_{i_1, \ldots, i_{k}}
	}{t}
	&=
	\PNorm{
		\sum_{i_1\in[d_1],\dots,i_{k-1}\in[d_{k-1}]}
		\left(\prod_{j \in [k - 1]} \sigma^{(j)}_{i_j}\right)
		B_{i_1,\dots,i_{k-1}}
	}{t}
	\\&\le
	C_t^{k-1}\PNorm{\left(
		\sum_{i_1\in[d_1],\dots,i_{k-1}\in[d_{k-1}]}
		B_{i_1,\dots,i_{k-1}}^2
		\right)^{1/2}}{t}
	\numberthis\label{eq:gen-khin-ih}
	\\&=
	C_t^{k-1}\PNorm{
		\sum_{i_1\in[d_1],\dots,i_{k-1}\in[d_{k-1}]}
		B_{i_1,\dots,i_{k-1}}^2
	}{t/2}^{1/2}
	\\&\le
	C_t^{k-1}\left(
	\sum_{i_1\in[d_1],\dots,i_{k-1}\in[d_{k-1}]}
	\PNorm{B_{i_1,\dots,i_{k-1}}^2}{t/2}
	\right)^{1/2}
	\numberthis\label{eq:gen-khin-triangle}
	\\&=
	C_t^{k-1}\left(
	\sum_{i_1\in[d_1],\dots,i_{k-1}\in[d_{k-1}]}
	\PNorm{B_{i_1,\dots,i_{k-1}}}{t}^2
	\right)^{1/2}
	\; .
	\end{align*}
	Here \eqref{eq:gen-khin-ih} is the inductive hypothesis and \eqref{eq:gen-khin-triangle} is the triangle inequality.
    It remains to bound $\PNorm{B_{i_1,\dots,i_{k-1}}}{t}^2 \le C_t^2\sum_{i_c\in[d_k]}a_{i_1,\dots,i_k}^2$ by Khintchine's inequality, which finishes the induction step and hence the proof.
\end{proof}

The next lemma we will be using is a type of Rosenthal inequality, but which mixes large and small moments in a careful way.
It bears similarity to the one sided bound in~\cite{boucheron2013concentration} (Theorem 15.10) derived from the Efron Stein inequality,
and the literature has many similar bounds, but we still include a proof here based on first principles. 
\begin{lem}\label{lem:rosenthal-variant}
	There exists a universal constant $L$, such that, for
	$t \ge 1$ if $X_1, \ldots, X_k$ are independent non-negative random variables
	with $t$-moment, then
	\[
	\PNorm{\sum_{i \in [k]} (X_i - \Ep{X_i})}{t}
	\le L\left(\sqrt{t} \PNorm{\max_{i \in [k]} X_i}{t}^{1/2} \sqrt{\sum_{i \in [k]} \Ep{X_i}}
	+ t\PNorm{\max_{i \in [k]} X_i}{t}\right).
	\]
\end{lem}
\begin{proof}
	Throughout these calculations $L_1$, $L_2$ and $L_3$ will be universal constants.
	\begin{align*}
	\PNorm{\sum_{i \in [k]} (X_i - \Ep{X_i})}{t}
	&\le
	L_1\PNorm{\sum_{i \in [k]} \sigma_i X_i}{t}
	\quad\text{(Symmetrization)}
	\\&\le
	L_2\sqrt{t}\PNorm{\sum_{i \in [k]} X_i^2}{t/2}^{1/2}
	\quad\text{(Khintchine's inequality)}
	\\&\le
	L_2\sqrt{t}\PNorm{\max_{i\in[k]} X_i \cdot \sum_{i \in [k]} X_i}{t/2}^{1/2}
	\quad\text{(Non-negativity)}
	\\&\le
	L_2\sqrt{t}\PNorm{\max_{i\in[k]} X_i}{t}^{1/2} \cdot \PNorm{\sum_{i \in [k]} X_i}{t}^{1/2}
	\quad\text{(Cauchy-Schwartz)}
	\\&\le
	L_2\sqrt{t}\PNorm{\max_{i\in[k]} X_i}{t}^{1/2}\left(\sqrt{\sum_{i \in [k]} \Ep{X_i}}
	+ L_2\PNorm{\sum_{i \in [k]} (X_i - \Ep{X_i})}{t}^{1/2}\right)
	\; .
	\end{align*}
	Now let $C = \PNorm{\sum_{i \in [k]} (X_i - \Ep{X_i})}{t}^{1/2}$,
	$B = L_2\sqrt{\sum_{i \in [k]} \Ep{X_i}}$,
	and $A=\sqrt{t}\PNorm{\max_{i\in[k]} X_i}{t}^{1/2}$.
	then we have shown $C^2 \le A(B+C)$.
	That implies $C$ is smaller than the largest of the roots of the quadratic.
	Solving this quadratic inequality gives
	$C^2 \le L_3(AB+A^2)$
	which is the result.
\end{proof}

We can now prove that {\sf SHRT} and {\sf TensorSRHT} has the Strong JL Moment Property.

\begin{lem}\label{lem:tensorfastjl}
	There exists a universal constant $L$, such that, the following holds.
	Let $k \in \ZZ_{> 0}$, and $(D^{(i)})_{i \in [k]} \in \prod_{i \in [k]} \RR^{d_i \times d_i}$
	be independent diagonal matrices with independent Rademacher variables. Define
	$d = \prod_{i \in [k]} d_i$ and $D = D_1  \times D_2 \times \cdots D_k \in \RR^{d\times d}$.
	Let $P \in \RR^{m \times d}$ be an independent sampling matrix
	which samples exactly one coordinate per row, and define $M = PHD$ where $H$
	is a $d \times d$ Hadamard matrix.
	Let $x \in \RR^d$ be any vector with $\norm{x}_2 = 1$ and $t \ge 1$, then
	\begin{align*}
	\PNorm{\tfrac{1}{m}\norm{PHDx}_{2}^2 - 1}{t} \le L\left(\sqrt{\frac{tr^k}{m}} + \frac{tr^k}{m} \right)
	\; ,
	\end{align*}
	where $r=\max\{t,\log m\}$.
	
	There exists a universal constant $L'$, such that, setting $m = \Omega\left(\eps^{-2}\log(1/\delta)(L'\log(1/\eps\delta))^k\right)$,
	we get that $\frac{1}{\sqrt{m}}PHD$ has Strong $(\eps, \delta)$-JL Moment Property.
\end{lem}
Note that setting $k=1$, this matches the Fast Johnson Lindenstrauss analysis in~\cite{cohen16stablerank}.

\begin{proof}
	Throughout the proof $C_1, C_2$ and $C_3$ will denote universal constants.
	
	For every $i \in [m]$ we let $P_i$ be the random variable that says which
	coordinate the $i$'th row of $P$ samples, and we define the random variable
	$Z_i = M_i x = H_{P_i} D x$. We note that since the variables $(P_i)_{i \in [m]}$
	are independent then the variables $(Z_i)_{i \in [m]}$ are conditionally
	independent given $D$, that is, if we fix $D$ then $(Z_i)_{i \in [m]}$ are
	independent.
	
	We use \Cref{lem:rosenthal-variant}, the triangle inequality, and Cauchy-Schwartz to get that
	\begin{align*}\label{eq:FJLT-proof-eq}
	&\PNorm{\tfrac{1}{m}\sum_{i \in [m]} Z_i^2 - 1}{t}
	\\&\quad= \PNorm{\Epcond{\left(\tfrac{1}{m}\sum_{i \in [m]} Z_i^2 - 1\right)^t}{D}^{1/t}}{t}
	\\&\quad\le C_1 \PNorm{\frac{\sqrt{t}}{m} \Epcond{\left(\max_{i \in [m]} Z_i^2 \right)^{t}}{D}^{1/(2t)} \sqrt{\sum_{i \in [m]}\Epcond{Z_i^2}{D}} 
		+ \frac{t}{m} \Epcond{\left(\max_{i \in [m]} Z_i^2 \right)^{t}}{D}^{1/t}}{t}
	\\&\quad\le C_1 \frac{\sqrt{t}}{m} \PNorm{\Epcond{\left(\max_{i \in [m]} Z_i^2 \right)^{t}}{D}^{1/(2t)} \sqrt{\sum_{i \in [m]}\Epcond{Z_i^2}{D}}}{t}
	+ C_1 \frac{t}{m}\PNorm{\max_{i \in [m]} Z_i^2}{t}
	\\&\quad\le C_1 \frac{\sqrt{t}}{m} \PNorm{\max_{i \in [m]} Z_i^2}{t}^{1/2} \PNorm{\sum_{i \in [m]}\Epcond{Z_i^2}{D}}{t}^{1/2}
	+ C_1 \frac{t}{m}\PNorm{\max_{i \in [m]} Z_i^2}{t}
	\; .
	\end{align*}
	
	By orthogonality of $H$ we have $\norm{HDx}_{2}^2 = d\norm{x}_{2}^2$ independent of $D$.
	Hence 
	\[
	\sum_{i \in [m]} \Epcond{Z_i^2}{D} = \sum_{i \in [m]} \norm{x}_{2}^2 = m
	\; .
	\]
	To bound $\PNorm{\max_{i \in [m]} Z_i^2}{t}$ we first use \cref{lem:gen-khinchine} to show
	\begin{align*}
	\PNorm{Z_i^2}{r}
	= \PNorm{H_{P_i} D x}{2r}^2
	= \PNorm{D x}{2r}^2
	\le r^{k} \norm{x}_{2}^2
	\; .
	\end{align*}
	We then bound the maximum using a sufficiently high powered sum:
	\begin{align*}
	\PNorm{\max_{i \in [m]} Z_i^2}{t}
	\le \PNorm{\max_{i \in [m]} Z_i^2}{r}
	\le \left(\sum_{i \in [m]} \PNorm{Z_i^2}{r}^r\right)^{1/r}
	\le m^{1/r} r^{k} \norm{x}_{2}^2
	\le e r^{k}
	\; ,
	\end{align*}
	where the last inequality follows from $r \ge \log m$.
	This gives us that
	\[
	\PNorm{\tfrac{1}{m}\sum_{i \in [m]} Z_i^2 - \norm{x}_{2}^2}{t}
	\le C_2\sqrt{\frac{tr^k}{m}} + C_2\frac{tr^k}{m}
	\; ,
	\]
	which finishes the first part of the proof.
	
	We set $m = 4e^2 C_2^2 \eps^{-2} \log(1/\delta) (C_3 \log(1/(\delta \eps)))^k$,
	such that, $r \le C_3 \log(1/(\delta \eps))$. Hence $m \ge 4e^2C_2^2 \eps^{-2} \log(1/\delta) r^k$.
	We then get that
	\begin{align*}
	\PNorm{\norm{PHDx}_{2}^2 - 1}{t}
	\le C_2\sqrt{\frac{tr^k}{4e^2C_2^2 \eps^{-2} \log(1/\delta) r^k}} + C_2\frac{tr^k}{4e^2C_2^2 \eps^{-2} \log(1/\delta) r^k}
	\le \frac{\eps}{e} \sqrt{\frac{t}{\log 1/\delta}}
	\; ,
	\end{align*}
	for all $1 \le t  \le \log(1/\delta)$ which finishes the proof.
\end{proof}

Now we have proved that the Strong JL Moment Property is satisfied by the {\sf SRHT}, the {\sf TensorSRHT} as well as {\sf OSNAP} transform,
but we still need to prove the usefulness of the property. Our next result remedies this and show that the Strong JL Moment
Property is preserved under multiplication. We will use the following decoupling lemma which first appeared in \cite{hitczenko1994domination},
but the following is roughly taken from \cite{de2012decoupling}, which we also recommend for readers interested in more general versions.
\begin{lem}[General decoupling, \cite{de2012decoupling} Theorem 7.3.1, paraphrasing]\label{lem:gen_decoup}
	There exists an universal constant $C_0$, such that,
	given any two sequences $(X_i)_{i \in [n]}$ and $(Y_i)_{i \in [n]}$ of random variables,
	satisfying
	\begin{enumerate}
		\item $\Prpcond{Y_i > t}{(X_j)_{j \in [i - 1]}} = \Prpcond{X_i > t}{(X_j)_{j \in [i - 1]}}$ for every $t \in \RR$ and for every $i \in [n]$.
		\item The sequence $(Y_i)_{i \in [n]}$ is conditionally independent given $(X_i)_{i \in [n]}$.
		\item $\Prpcond{Y_i > t}{(X_j)_{j \in [i - 1]}} = \Prpcond{Y_i > t}{(X_j)_{j \in [n]}}$ for every $t \in \RR$ and for every $i \in [n]$.
	\end{enumerate}
	Then for all $t \ge 1$,
	\begin{align*}
	\PNorm{\sum_{i \in [n]} X_i}{t} \le C_0\PNorm{\sum_{i \in [n]} Y_i}{t}
	\end{align*}
\end{lem}
We are now ready to state and prove the main lemma of this section.
This basically says that if you take $k$ independent JL transforms, that all have the Strong $(\eps/\sqrt{k},\delta)$-JL Moment Property, SJLMP, then the result has the $(\eps,\delta)$ SJLMP.
A simple union bound would give the same result where each matrix has the $(\eps/k,\delta)$ SJLMP, but that would ultimately result in a higher dependency on the tensoring dimension.
A simple change to the proof shows that we only need the $i$th JL transform to have the $(\eps/\sqrt{i},\delta)$ SJLMP, but that ultimately makes no difference for our construction.
\begin{lem}\label{lem:jl-product}
	There exists a universal constant $L$, such that, 
	for any constants $\eps, \delta \in [0,1]$ and positive
	integer $k \in \ZZ_{> 0}$. If
	$
	M^{(1)} \in \RR^{d_2 \times d_1}, \ldots, 
	M^{(k)} \in \RR^{d_{k+1} \times d_c}
	$
	are independent random matrices with the Strong
	$(\eps/(L\sqrt{k}), \delta)$-JL Moment Property,
	then the matrix $M = M^{(k)} \cdot \ldots \cdot M^{(1)}$
	has the Strong $(\eps, \delta)$-JL Moment Property.
\end{lem}
\begin{proof}
	Let $x \in \RR^d_1$ be an arbitrary, fixed unit vector, and fix $1 < t \le \log(1/\delta)$.
	We define $X_i = \norm{M^{(i)} \cdot \ldots \cdot M^{(1)}x}_{2}^2$ and
	$Y_i = X_i - X_{i - 1}$ for every $i \in [k]$. By telescoping we then have that
	$X_i - 1 = \sum_{j \in [i]} Y_i$. We let $(T^{(i)})_{i \in [k]}$ be
	independent copies of $(M^{(i)})_{i \in [k]}$ and define
	\[
	Z_i
	= \norm{T^{(i)} \cdot M^{(i - 1)} \cdot \ldots \cdot M^{(1)}x}_{2}^2 - \norm{M^{(i - 1)} \cdot \ldots \cdot M^{(1)}x}_{2}^2
	\; ,
	\]
	for every $i \in [k]$. We get the following three properties:
	\begin{enumerate}
		\item $\Prpcond{Z_i > t}{(M^{(j)})_{j \in [i - 1]}} = \Prpcond{Y_i > t}{(M^{(j)})_{j \in [i - 1]}}$ for every $t \in \RR$ and every $i \in [k]$.
		\item The sequence $(Z_i)_{i \in [k]}$ is conditionally independent given $(M^{(i)})_{i \in [k]}$.
		\item $\Prpcond{Z_i > t}{(M^{(j)})_{j \in [i - 1]}} = \Prpcond{Z_i > t}{(M^{(j)})_{j \in [k]}}$ for every $t \in \RR$ and for every $i \in [k]$.
	\end{enumerate}
	This means we can use \Cref{lem:gen_decoup} to get
	\begin{align}\label{eq:jl-prod-decoup}
	\PNorm{\sum_{j \in [i]} Y_j}{t} \le C_0\PNorm{\sum_{j \in [i]} Z_j}{t}   
	\; .
	\end{align}
	for every $i \in [k]$.
	
	We will prove by induction on $i \in [k]$ that
	\begin{align}\label{eq:jl-moment-prod-induc}
	\PNorm{X_i - 1}{t}
	\le \frac{\eps}{e}\sqrt{\frac{t}{\log(1/\delta)}}
	\le 1
	\; .
	\end{align}
	For $i = 1$ we use that $M^{(1)}$ has the Strong $(\eps/(L\sqrt{k}), \delta)$-JL
	Moment Property and get that
	\[
	\PNorm{\norm{M^{(1)}x}_{2}^2 - 1}{t}
	\le \frac{\eps}{e L \sqrt{k}} \sqrt{\frac{t}{\log(1/\delta)}}
	\le \frac{\eps}{e}\sqrt{\frac{t}{\log(1/\delta)}}
	\; .
	\]
	Now assume that \eqref{eq:jl-moment-prod-induc} is true for $i - 1$.
	Using \eqref{eq:jl-prod-decoup}
	we get that $\PNorm{X_i - 1}{t} = \PNorm{\sum_{j \in [i]} Y_j}{t} \le C_0 \PNorm{\sum_{j \in [i]} Z_j}{t}$.
	By using that $(T^{(j)})_{j \in [i]}$ has
	the Strong $(\eps/(L\sqrt{k}), \delta)$-JL Moment Property together with Khintchine's inequality
	(\Cref{lem:khintchine}), we get that
	\begin{align*}
	\PNorm{\sum_{j \in [i]} Z_j}{t}
	&= \PNorm{\Epcond{\left( \sum_{j \in [i]} Z_j \right)^t}{(M^{(j)})_{j \in [i]}}^{1/t}}{t}
	\\&\le C_1 \PNorm{\frac{\eps}{e L \sqrt{k}}\sqrt{\frac{t}{\log(1/\delta)}}\sqrt{\sum_{j \in [i]} X_j^2 }}{t}
	\\&= C_1 \frac{\eps}{e} \sqrt{\frac{t}{\log(1/\delta)}} \cdot \frac{1}{L \sqrt{k}} \sqrt{\PNorm{\sum_{j \in [i]}X_j^2}{t/2}}
	\\&\le C_1 \frac{\eps}{e} \sqrt{\frac{t}{\log(1/\delta)}} \cdot \frac{1}{L \sqrt{k}} \sqrt{\sum_{j \in [i]} \PNorm{X_j}{t}^2}
	\; ,
	\end{align*}
	where the last inequality follows from the triangle inequality.
	Using the triangle inequality and \eqref{eq:jl-moment-prod-induc} we get that
	\[
	\PNorm{X_j}{t}
	\le 1 + \PNorm{X_j - 1}{t}
	\le 2
	\; ,
	\]
	for every $j \in [i]$. Setting $L = 2C_0C_1$ we get that
	\begin{align}
	\PNorm{\sum_{j \in [i]} Y_j}{t}
	&\le \frac{\eps}{e} \sqrt{\frac{t}{\log(1/\delta)}} \cdot \frac{C_0 C_1}{L \sqrt{k}} \sqrt{\sum_{j \in [i]} \PNorm{X_j}{t}^2}
	\\&\le \frac{\eps}{e}\sqrt{\frac{t}{\log(1/\delta)}} \cdot \frac{C_0 C_1}{L \sqrt{k}} \cdot 2\sqrt{i}
	\\&\le \frac{\eps}{e}\sqrt{\frac{t}{\log(1/\delta)}}
	\; ,
	\end{align}
	which finishes the induction. Now we have that 
	$\PNorm{\norm{Mx}_{2}^2 - 1}{t} \le \frac{\eps}{e} \sqrt{\frac{t}{\log(1/\delta)}}$
	so we conclude that $M$ has Strong $(\eps, \delta)$-JL Moment Property.
\end{proof}

A simple corollary of this result is a sufficient condition for our recursive sketch $\Pi^q$
to have the Strong JL Moment Property.
\begin{cor}[Strong JL Moment Property for $\Pi^q$]\label{cor:strong-moment-sketch}
	For any integer $q$ which is a power of two, let $\Pi^q: \RR^{m^q} \to \RR^m$ be defined as in
	Definition \ref{def:sketch}, where both of the common distributions
	$\sbase: \RR^{m^2}\to \RR^m$ and $\tbase: \RR^d \to \RR^m$, satisfy the
	Strong $\left( O\left(\frac{\eps}{\sqrt{q}}\right), \delta\right)$-JL Moment Property.
	Then it follows that $\Pi^q$ satisfies the Strong $(\eps, \delta)$-JL Moment Property.
\end{cor}
\begin{proof}
	The proof follows from using \Cref{lem:pi^q-decomposition} and \Cref{lem:jl-product}.
\end{proof}

We conclude this section by proving \Cref{thm:high-prob-moment}.
\highprobmoment*
\begin{proof}
	Let $\delta = \frac{1}{\poly{n}}$ denote the failure probability.
	Define $q = \lceil \log_2(p) \rceil$ and let $\Pi^p \in \RR^{m \times d^p}$
	and $\Pi^q \in \RR^{m \times d^q}$ be the sketches defined in \Cref{def:sketch},
	where $\sbase \in \RR^{m \times m^2}$ is a {\sf TensorSRHT} sketch and
	$\tbase \in \RR^{m \times d}$ is an {\sf OSNAP} sketch with sparsity parameter
	$s$, which will be set later. 
	
	\paragraph{Oblivious Subspace Embedding}
	Let $m = \Theta\left(\frac{p s_{\lambda}^2 \log(1/(\eps \delta))^3}{\eps^2} \right)$
	and $s = \Theta\left(\frac{\sqrt{p} s_\lambda \log(1/\delta)}{\eps}\right)$ be integers,
	then \Cref{lem:tensorfastjl} and \Cref{lem:osnap-strong-moment} implies that
	$\sbase$ and $\tbase$ has the Strong $\left(O\left(\frac{\eps}{\sqrt{q} s_\lambda}\right), \delta \right)$-JL
	Moment Property, thus using \Cref{cor:strong-moment-sketch} we conclude that $\Pi^q$ has the Strong
	$\left(\frac{\eps}{s_\lambda}, \delta\right)$-JL Moment Property and in particular
	it has the $\left(\frac{\eps}{s_\lambda}, \delta, \log(1/\delta)\right)$-JL Moment Property.
	By \Cref{lem:jl-moment-ose} we then get that $\Pi^q$ is an
	$(\eps, \delta, s_\lambda, d^q, n)$-Oblivious Subspace Embedding,
	and by \Cref{lem:reduce-p-to-q} we get that $\Pi^p$ is an
	$(\eps, \delta, s_\lambda, d^p, n)$-Oblivious Subspace Embedding.
	
	\paragraph{Approximate Matrix Multiplication}
	Let $m = \Theta\left(\frac{p \log(1/(\eps \delta))^3}{\eps^2} \right)$
	and $s = \Theta\left(\frac{\sqrt{p} \log(1/\delta)}{\eps}\right)$ be integers. 
	Then \Cref{lem:tensorfastjl} and \Cref{lem:osnap-strong-moment} implies that
	$\sbase$ and $\tbase$ has the Strong $\left(O\left(\frac{\eps}{\sqrt{q} s_\lambda}\right), \delta \right)$-JL
	Moment Property. Thus, using \Cref{cor:strong-moment-sketch} we conclude that $\Pi^q$ has the Strong
	$(\eps, \delta)$-JL Moment Property and in particular
	it has the $(\eps, \delta, \log(1/\delta))$-JL Moment Property.
	By \Cref{lem:jl-moment-ose} we then get that $\Pi^q$ has
	the $(\eps, \delta)$-Approximate Matrix Multiplication Property,
	and by \Cref{lem:reduce-p-to-q} we get that $\Pi^p$ has
	the $(\eps, \delta)$-Approximate Matrix Multiplication Property.
	
	\paragraph{Runtime of Algorithm \ref{alg:main} when the base sketch $\sbase$ is a {\sf TensorSRHT} sketch and $\tbase$ is an {\sf OSNAP} sketch with sparsity parameter $s$:}
	We compute the time of running Algorithm \ref{alg:main} on a vector $x$.
	Computing $Y^0_j$ for each $j$ in lines~\ref{y-0-q} and \ref{y-0-p} of algorithm
	requires applying an {\sf ONSAP} sketch on either $x$ or $e_1$ which takes time
	$O( s \cdot \text{nnz}(x))$. Therefore computing all $Y^0_j$'s takes time $O(q s \cdot \text{nnz}(x))$.
	
	Computing each of $Y^l_j$'s for $l \ge 1$ in line~\ref{Y-l-j} of Algorithm \ref{alg:main}
	amounts to applying a {\sf TensorSRHT} sketch of input dimension $m^2$ and target dimension
	of $m$ on $Y^{l-1}_{2j-1} \otimes Y^{l-1}_{2j}$. This takes time $O( m \log m)$.
	Therefore computing $Y^l_j$ for all $l, j \ge 1$ takes time $O(q \cdot m \log m)$.
	Note that $q \le 2p$ hence the total running time of Algorithm \ref{alg:main} on one vector
	$x$ is $O(p m\log_2m + p s \cdot \text{nnz}(w))$. Sketching $n$ columns of a matrix
	$X \in \RR^{d \times n}$ takes time $O(p (nm\log_2m +  s \cdot \text{nnz}(X)))$.
	
	In the setting of {\bf (1)} we have that $s = O\left( \frac{\sqrt{p} s_\lambda \log(1/\delta)}{\eps} \right)$,
	hence we get a runtime of \\
	$
	O\left(p nm\log_2m + \frac{p^{3/2} s_\lambda \log(1/\delta))}{\eps}\text{nnz}(X)\right)
	= \tilde{O}\left(pnm + \frac{p^{3/2} s_\lambda}{\eps}\text{nnz}(X)\right).
	$
\end{proof}

\section{Linear Dependence on the Statistical Dimension {$s_\lambda$}}\label{sec:linear-s-lambda}

In this section, we show that if one chooses the internal nodes and the leaves of our recursive construction from Section \ref{sec:sektch-construction} to be {\sf TensorSRHT} and {\sf OSNAP} transform respectively, then the recursive construction $\Pi^q$ as in Definition \ref{def:sketch} yields a high probability OSE with target dimension $\widetilde{O}(p^4 s_\lambda)$. Thus, we prove Theorem \ref{high-prob-sketch}. This sketch is very efficiently computable for high degree tensor products because the {\sf OSNAP} transform is computable in input sparsity time and the {\sf TensorSRHT} supports fast matrix vector multiplication for tensor inputs.

We start by defining the \emph{Spectral Property} for a sketch. We use the notation $\|\cdot\|_{op}$ to denote the operator norm of matrices.
\begin{defn}[Spectral Property]\label{def:spectral-prp}
	For any positive integers $m, n, d$ and any $\eps, \delta, \mu_F, \mu_2\ge 0$ we say that a random matrix $S \in \RR^{m \times d}$ satisfies the \emph{$(\mu_F, \mu_2 , \epsilon , \delta, n)$-spectral property} if, for every fixed matrix $U \in \RR^{d \times n}$ with $\|U\|_F^2 \le \mu_F$ and $\| U \|_{op}^2 \le \mu_2$,
	$$\Pr_{S} \left[ \left\| U^\top S^\top S U - U^\top U \right\|_{op}\le \epsilon\right] \ge 1 - \delta.$$
\end{defn}

The \emph{spectral property} is a central property of our sketch construction from Section \ref{sec:sektch-construction} when leaves are {\sf OSNAP} and internal nodes are {\sf TensorSRHT}.
This is a powerful property which implies that any sketch which satisfies the \emph{spectral property}, is an \emph{Oblivious Subspace Embedding}.
The {\sf SRHT}, {\sf TensorSRHT}, as well as {\sf OSNAP} sketches (Definitions \ref{def:fastjl}, \ref{def:tensorfastjl}, \ref{def:osnap} respectively) with target dimension $m = \Omega\left( (\frac{\mu_F \mu_2}{\epsilon^2}) \cdot \poly{\log (nd/\delta)} \right)$ and sparsity parameter $s = \Omega (\poly{\log (nd/\delta)})$, all satisfy the above-mentioned spectral property \cite{Sarlos06, Tropp11, nelson2013osnap}. 

In \cref{subsec:matrix-tools} we recall the tools from the literature which we use to prove the spectral property for our construction $\Pi^q$. Then in \cref{sec:spectral-prop-pi^q} we show that our recursive construction in Section \ref{sec:sektch-construction} satisfies the Spectral Property of Definition \ref{def:spectral-prp} as long as $I_{d^{q}} \times \tbase$ and $I_{m^{q}} \times \sbase$ satisfy the Spectral Property. Therefore, we analyze the Spectral Property of $I_{d^{q}} \times $ {\sf OSNAP} and $I_{m^{q}} \times $ {\sf TensorSRHT} in \cref{sec:identity-tensorsrht} and \cref{sec:identity-osnap} respectively. Finally we put everything together in \cref{sec:putt-together-highprob-ose} and prove that when the leaves are {\sf OSNAP} and the internal nodes are {\sf TensorSRHT} in our recursive construction of Section \ref{sec:sektch-construction}, the resulting sketch $\Pi^q$ satisfies the Spectral Property thereby proving Theorem \ref{high-prob-sketch}.

\subsection{Matrix Concentration Tools} \label{subsec:matrix-tools}
In this section we present the definitions and tools which we use for proving concentration properties of random matrices.

\begin{claim}\label{claim:tensor-reordering}
	For every $\epsilon,\delta>0$ and any sketch $S\in\RR^{m\times d}$ such that $I_k \times S$ satisfies $(\mu_F, \mu_2, \epsilon, \delta, n)$-spectral property, the sketch $S \times I_k$ also satisfies the $(\mu_F, \mu_2, \epsilon, \delta, n)$-spectral property.
\end{claim}
\begin{proof}
	Suppose $U\in \RR^{dk \times n}$. Then, note that there exists $U'\in \RR^{dk \times n}$ formed by permuting the rows of $U$ such that $(S\times I_k) U$ and $(I_k \times S)U'$ are identical up to a permutation of the rows. (In particular, $U'$ is the matrix such that the $(d,k)$-reshaping of any column $U^j$ of $U'$ is the transpose of the $(k,d)$-reshaping of the corresponding column $U'^j$ of $U'$.) Then, observe that
	\[
	U^\top U = U'^\top U'.
	\]
	and
	\[
	U^\top (S\times I_k)^\top (S\times I_k) U  = U'^\top (I_k \times S)^\top (I_k\times S) U'.
	\]
	Therefore,
	\[
	\|{U}^\top (S \times I_k)^\top (S \times I_k) U - {U}^\top U\|_{op} = \|{U'}^\top (S \times I_k)^\top (S \times I_k) U' - {U'}^\top U'\|_{op}.
	\]
	Moreover, since $U$ and $U'$ are identical up to a permutation of the rows, we have $\| U\|_{op} = \| U'\|_{op}$ and $\|U\|_F = \|U'\|_F$. The desired claim now follows easily.
\end{proof}

We will use matrix Bernstein inequalities to show spectral guarantees for sketches, 
\begin{lem}[Matrix Bernstein Inequality (Theorem 6.1.1 in \cite{Tropp15})] 
\label{lem:Bernstein}
Consider a finite sequence ${Z_i}$ of independent, random matrices with dimensions $d_1 \times d_2$. Assume that each random matrix satisfies $\mathbb{E}[Z_i] = 0$ and $\|Z_i\|_{op} \leq B$ almost surely. Define $\sigma^2=\max\{\|\sum_i \mathbb{E}[Z_iZ_i^{*}]\|_{op}, \|\sum_i \mathbb{E}[Z_i^{*}Z_i]\|_{op}\}$. Then for all $t>- 0$, \[\mathbb{P}\left[ \left\|\sum_i Z_i \right\|_{op} \geq t\right] \leq (d_1+d_2)\cdot \exp\left(\frac{-t^2 / 2}{\sigma^2+Bt / 3}\right) \text{.}\]
\end{lem}

\begin{lem}[Restatement of Corollary 6.2.1 of \cite{Tropp15}]\label{bernstein-matrix}
	Let $B$ be a fixed $n \times n$ matrix. Construct an $n \times n$ matrix $R$ that satisfies,
	$$\mathbb{E} [R]=B \text{~~~and~~~} \|R\|_{op} \le L,$$
	almost surely. Define $M=\max\{\| \mathbb{E}[RR^{*}]\|_{op}, \|\mathbb{E}[R^{*}R]\|_{op}\}$.
	Form the matrix sampling estimator,
	$$\bar{R} = \frac{1}{m} \sum_{k=1}^m R_k,$$
	where each $R_k$ is an independent copy of $R$. Then,
	$$\Pr \left[ \|\bar{R} - B\|_{op} \ge \epsilon \right] \le 8n \cdot \exp\left( \frac{-m\epsilon^2/2}{M + 2L\epsilon/3} \right).$$
\end{lem}

To analyze the performance of {\sf SRHT} we need the following claim which shows that with high probability individual entries of the Hadamard transform of a vector with random signs on its entries do not ``overshoot the mean energy'' 
by much.
\begin{claim}\label{l-infinity-bound-HD}
Let $D$ be a $d\times d$ diagonal matrix with independent Rademacher random variables along the diagonal. Also, let $H$ be a $d\times d$ Hadamard matrix. Then, for every $x \in \RR^d$,
$$\Pr_D\left[ \| H  D \cdot x \|_\infty \le 2\sqrt{\log_2(d/\delta)} \cdot \|x\|_2 \right] \ge 1 - \delta.$$
\end{claim}
\begin{proof}
By Khintchine's inequality, Lemma \ref{lem:khintchine} we have that for every $t\ge 1$ and every $j \in [d]$ the $j^{\text{th}}$ element of $H D x$ has a bounded $t^{\text{th}}$ moment as follows,

$$\PNorm{(HDx)_j}{t} \le \sqrt{t} \cdot \|x\|_2.$$
Hence by applying Markov's inequality to the $t^{\text{th}}$ moment of $|(HDx)_j|$ for $t = \log_2(d/\delta)$ we get that,
$$\Pr \left[ |(HDx)_j| \ge 2 \sqrt{\log_2(d/\delta)} \cdot \|x\|_2 \right] \le \delta/d.$$
The claim follows by a union bound over all entries $j \in [d]$.
\end{proof}

\begin{claim}\label{i-infinity-HD-HD}
Let $D_1, D_2$ be two independent ${d}\times {d}$ diagonal matrices, each with diagonal entries given by independent Rademacher random variables. Also, let $H$ be a ${d}\times {d}$ Hadamard matrix. Then, for every $x \in \RR^{d^2}$,
$$\Pr_{D_1,D_2}\left[ \| \left((H D_1) \times (H D_2)\right) \cdot x \|_\infty \le 4\log_2(d/\delta) \cdot \|x\|_2 \right] \ge 1 - \delta.$$
\end{claim}
\begin{proof}
By Claim \ref{tensor-prod-identity} we can write that, 
$$(H D_1) \times (H D_2) = (H \times H)  (D_1 \times D_2),$$
where $H \times H$ is indeed a Hadamard matrix of size $d^2 \times d^2$ which we denote by $H'$. The goal is to prove 
$$\Pr_{D_1,D_2}\left[ \|  H' (D_1 \times D_2) \cdot x \|_\infty \le 4\log_2(d/\delta) \cdot \|x\|_2 \right] \ge 1 - \delta.$$
By Lemma \ref{lem:gen-khinchine} we have that for every $t\ge 1$ and every $j \in [d^2]$ the $j^{\text{th}}$ element of $H' (D_1 \times D_2) x$ has a bounded $t^{\text{th}}$ moment as follows,

$$\PNorm{(H' (D_1 \times D_2)x)_j}{t} \le {t} \cdot \|x\|_2.$$
Hence by applying Markov's inequality to the $t^{\text{th}}$ moment of $|(H'(D_1 \times D_2)x)_j|$ for $t = \log_2(d/\delta)$ we get that,
$$\Pr \left[ |(H'(D_1 \times D_2)x)_j| \ge 4 {\log_2(d/\delta)} \cdot \|x\|_2 \right] \le \delta/d^2.$$
The claim follows by a union bound over all entries $j \in [d^2]$.

\end{proof}

\subsection{Spectral Property of the sketch \texorpdfstring{$\Pi^q$}{}}\label{sec:spectral-prop-pi^q}
In this section we show that the sketch $\Pi^q$ presented in Definition \ref{def:sketch} inherits the spectral property (see Definition \ref{def:spectral-prp}) from the base sketches $\sbase$ and $\tbase$.
We start by the following claim which proves that composing two random matrices with spectral property results in a matrix with spectral property.

\begin{claim} \label{spectral-prop-composition}
For every $\epsilon, \epsilon', \delta, \delta' > 0$, suppose that $S \in\RR^{m \times t}$ is a sketch which satisfies the $( (\mu_F + 1)(1+\epsilon'), \mu_2 + 1 + \epsilon', \epsilon, \delta, n)$-spectral property and also suppose that the sketch $T \in \RR^{t \times d}$ satisfies the $(\mu_F + 1, \mu_2 + 1, \epsilon', \delta'/n, n)$-spectral property. Then $S \cdot T$ satisfies the $\left( \mu_F + 1 , \mu_2 + 1, \epsilon + \epsilon' , \delta + \delta'(1 + 1/n), n \right)$-spectral property.
\end{claim}
\begin{proof}
Suppose $S$ and $T$ are matrices satisfying the hypothesis of the claim. Consider an arbitrary matrix $U \in \RR^{d \times n}$ which satisfies $\| U \|_F^2 \le \mu_F + 1$ and $\| U \|_{op}^2 \le \mu_2 + 1$.
We want to prove that for every such $U$,
$$\Pr\left[ \| U^\top (S \cdot T)^\top (S \cdot T) U - U^\top U \|_{op} \le \epsilon + \epsilon' \right] \ge 1 - \delta - \delta'(1+1/n).$$
Let us define the event $\E$ as follows,
$$\E := \left\{ \| T \cdot U \|_F^2 \le \left( 1 + \epsilon' \right) \| U\|_F^2 \text{ and } \left\|U^\top T^\top T U - U^\top U \right\|_{op} \le \epsilon' \right\}.$$
We show that this event holds with probability $1-\delta'(1+1/n)$ over the random choice of sketch $T$. The spectral property of $T$ implies that for every column $U^j$ of matrix $U$,
$$\| T U^j \|_2^2 = \left( 1 \pm \epsilon' \right) \|U^j\|_2^2,$$
with probability $1- \frac{\delta'}{n}$. By a union bound over all $j \in [n]$, we have the following,
$$\Pr_T\left[ \| T \cdot U \|_F^2 \le \left( 1 + \epsilon' \right) \|U\|_F^2 \right] \ge 1 - {\delta'}.$$
Also,
$$\Pr_T \left[ \left\|U^\top T^\top T U - U^\top U \right\|_{op} \le \epsilon' \right] \ge 1 - \delta'/n.$$
Therefore by union bound,
$$\Pr_T [\E] \ge 1 - \delta'(1+1/n).$$
We condition on $T \in \E$ in the rest of the proof.
Since $S$ satisfies the $( (\mu_F + 1)(1+\epsilon'), \mu_2 + 1 + \epsilon', \epsilon, \delta, n)$-spectral property,
$$\Pr_{S}\left[ \left\| (TU)^\top S^\top S (TU) - (TU)^\top (TU) \right\|_{op} \le \epsilon \right] \ge 1 - \delta.$$
Therefore,
\begin{align*}
&\Pr_{T,S}\left[ \left\| U^\top (S \cdot T)^\top (S \cdot T) U - U^\top U \right\|_{op} \le \epsilon + \epsilon' \right] \\
&\qquad\ge \Pr_{S}\left[ \left\| U^\top (S \cdot T)^\top (S \cdot T) U - U^\top U \right\|_{op} \le \epsilon + \epsilon' \, \Big| \, T \in \mathcal{E} \right] - \Pr_T[\bar{\mathcal{E}}]\\
&\qquad \ge \Pr_S \left[ \left\| (TU)^\top S^\top S (TU) - U^\top U \right\|_{op} \le \epsilon + \epsilon' \, \Big| \, T \in \mathcal{E} \right] - \delta'(1+1/n)\\
&\qquad \ge \Pr_S \left[ \left. \left\| (TU)^\top S^\top S (TU) - (TU)^\top (TU) \right\|_{op} + \left\| (TU)^\top (TU) - U^\top U \right\|_{op} \le \epsilon + \epsilon' \right| T \in \mathcal{E} \right] - \delta'(1+\frac{1}{n})\\
&\qquad \ge \Pr_S \left[ \left\| (TU)^\top S^\top S (TU) - (TU)^\top (TU) \right\|_{op} \le \epsilon \, \Big| \, T \in \mathcal{E} \right] - \delta'(1+1/n)\\
&\qquad\ge 1 - \delta - \delta'(1+1/n).
\end{align*}
This completes the proof.
\end{proof}

In the following lemma we show that composing independent random matrices with spectral property preserves the spectral property.

\begin{lem}\label{lem:spectral-property-composition}
	For any $\eps, \delta,  \mu_F, \mu_2 >0$ and every positive integers $k, n$, if $M^{(1)} \in \RR^{d_2 \times d_1}, \cdots M^{(k)} \in \RR^{d_{k+1} \times d_k}$ are independent random matrices with the $(2\mu_F+2 , 2\mu_2+2, O(\epsilon/k) , O(\delta/n k), n)$-spectral property then the product matrix $M = M^{(k)} \cdots M^{(1)}$ satisfies the $(\mu_F+1 , \mu_2 +1 , \epsilon , \delta, n)$-spectral property.
\end{lem}
\begin{proof}
	Consider a matrix $U \in \RR^{d_1 \times n}$ which satisfies $\| U \|_F^2 \le \mu_F + 1$ and $\| U \|_{op}^2 \le \mu_2 + 1$.
	We want to prove that for every such $U$,
	
	$$\Pr\left[ \| U^\top M^\top M U - U^\top U \|_{op} \le \epsilon \right] \ge 1 - \delta,$$
	where $M = M^{(k)} \cdots M^{(1)}$.
 	
 	By the assumption of the lemma the matrices $M^{(1)}, \cdots M^{(k)}$ satisfy the $(2\mu_F+2 , 2\mu_2+2, O(\epsilon/k) , O(\delta/nk), n)$-spectral property.
	For every $j \in [k]$, let us define the set $\E_j$ as follows,
	$$\E_j := \left\{ \left( M^{(1)},\cdots,M^{(j)} \right) : \begin{cases}
	1. \left\| \left( M^{(j)}  \cdots M^{(1)} \right) U \right\|_F^2 \le \left( 1 + \frac{\epsilon}{10k} \right)^j \|U\|_F^2\\
	2. \left\| U^\top \left( M^{(j)} \cdots M^{(1)} \right)^\top \left( M^{(j)} \cdots M^{(1)} \right) U - U^\top U \right\|_{op} \le \frac{ \epsilon j}{3k}
	\end{cases} \right\}.$$
	First we prove that for every $j \in \{ 1,\cdots,k-1\}$,
	
	$$\Pr_{M^{(j+1)}} \left[ \left. \left( M^{(1)}, \cdots, M^{(j+1)} \right) \in \E_{j+1}\right| \left( M^{(1)},\cdots,M^{(j)} \right) \in \E_{j} \right] \ge 1 - \frac{\delta}{2k}.$$
	Let us denote $\left( M^{(j)} \cdots M^{(1)} \right) \cdot U$ by $U'$. The condition $\left( M^{(1)},\cdots,M^{(j)} \right) \in \E_{j}$ implies that, $\|U'\|_F^2 \le (1+\epsilon /(10k))^j\|U\|_F^2$ and $\|{U'}^\top U' - U^\top U\|_{op} \le \frac{\epsilon j}{3k}$ and therefore by triangle inequality we have $\|U'\|_{op}^2 \le \left(\|U\|_{op} + \frac{\epsilon j}{3k} \right)^2$. The assumptions $\| U \|_F^2 \le \mu_F + 1$ and $\| U \|_{op}^2 \le \mu_2 + 1$ imply that $\| U' \|_F^2 \le 2\mu_F + 2$ and $\| U' \|_{op}^2 \le 2\mu_2 + 2$.
	Now note that by the assumption of the lemma, $M^{(j+1)}$ satisfies the $(2\mu_F+2 , 2\mu_2+2, O(\epsilon/k) , O(\delta/nk), n)$-spectral property. Therefore, 
	
	$$\Pr_{M^{(j+1)}} \left[ \left. \left\| \left(M^{(j+1)} U'\right)^\top M^{(j+1)} U' - {U'}^\top {U'} \right\|_{op} \le \frac{\epsilon}{3k} \right| \left( M^{(1)},\cdots,M^{(j)} \right) \in \E_{j} \right] \ge 1 - \delta/(4nk).$$
	Combining the above with $\|{U'}^\top U' - U^\top U\|_2 \le \frac{\epsilon j}{3k}$ gives,
	
	\begin{equation}\label{compose-lem-op-norm-ind-step}
	\Pr_{M^{(j+1)}} \left[ \left. \left\| \left(M^{(j+1)} U'\right)^\top M^{(j+1)} U' - {U}^\top {U} \right\|_{op} \le \epsilon\frac{j+1}{3k} \right| \left( M^{(1)},\cdots,M^{(j)} \right) \in \E_{j} \right] \ge 1 - \delta/(4nk).
	\end{equation}
	Also from the spectral property of $M^{(j+1)}$ it follows that for every column ${U'}^i$ of matrix $U'$,
	$$\| M^{(j+1)} {U'}^i \|_2^2 = \left( 1 \pm \epsilon/(10k) \right) \|{U'}^i\|_2^2,$$
	with probability $1- \frac{\delta}{4nk}$. By a union bound over all $i \in [n]$, we have the following,
	$$	\Pr_{M^{(j+1)}}\left[ \left. \| M^{(j+1)} \cdot U' \|_F^2 \le \left( 1 + \epsilon/(10k) \right) \|U'\|_F^2\right| \left( M^{(1)},\cdots,M^{(j)} \right) \in \E_{j} \right] \ge 1 - \frac{\delta}{4k}.$$
	Combining the above with $\|U'\|_F^2 \le (1+\epsilon /(10k))^j\|U\|_F^2$ gives,
	\begin{equation}\label{compose-lemfro-norm-ind-step}
	\Pr_{M^{(j+1)}}\left[ \left. \| M^{(j+1)} \cdot U' \|_F^2 \le \left( 1 + \frac{\epsilon}{10k} \right)^{j+1} \|U\|_F^2\right| \left( M^{(1)},\cdots,M^{(j)} \right) \in \E_{j} \right] \ge 1 - \frac{\delta}{4k}.
	\end{equation}
	A union bound on \eqref{compose-lem-op-norm-ind-step} and \eqref{compose-lemfro-norm-ind-step} gives,
	$$\Pr_{M^{(j+1)}} \left[ \left. \left( M^{(1)}, \cdots, M^{(j+1)} \right) \in \E_{j+1}\right| \left( M^{(1)}, \cdots, M^{(j)} \right) \in \E_{j} \right] \ge 1 - \frac{\delta}{4nk} - \frac{\delta}{4k} \ge 1 - \frac{\delta}{2k}.$$
	
	We also show that,
	$$\Pr_{M^{(1)}}[M^{(1)} \in \E_1] \ge 1 - \delta/{2k}.$$
	By the assumption of lemma we know that $M^{(1)}$ satisfies the $\left( 2\mu_F+2 , 2\mu_2+2, \frac{\epsilon}{10k} , \frac{\delta}{4nk}, n \right)$-spectral property. Therefore, 
	
	\begin{equation}\label{ompose-lem-op-norm-base}
	\Pr_{M^{(1)}} \left[ \|(M^{(1)} U)^\top M^{(1)} U - U^\top U\|_{op} \le \frac{\epsilon}{10k} \right] \ge 1 - \frac{\delta}{4nk}.
	\end{equation}
	Also for every column $U^i$ of matrix $U$,
	$$\| M^{(1)} U^i \|_2^2 = \left( 1 \pm \epsilon/(10k) \right) \|U^i\|_2^2,$$
	with probability $1- \frac{\delta}{4nk}$. By a union bound over all $i \in [n]$, we have the following,
	\begin{equation}\label{ompose-lem-fro-norm-base}
	\Pr_{M^{(1)}}\left[ \| M^{(1)} \cdot U \|_F^2 \le \left( 1 + \epsilon/(10k) \right) \|U\|_F^2 \right] \ge 1 - \frac{\delta}{4k}.
	\end{equation}
	A union bound on \eqref{ompose-lem-op-norm-base} and \eqref{ompose-lem-fro-norm-base} gives,
	$$\Pr_{T_1}[T_1 \in \E_1] \ge 1 - \frac{\delta}{4nk} - \frac{\delta}{4k} \ge 1 - \frac{\delta}{2k}.$$
	By the chain rule for events we have,
	\begin{align*}
	&\Pr_{M^{(1)}, \cdots, M^{(k)}} \left[ \left( M^{(1)}, \cdots, M^{(k)} \right) \in \E_{k} \right] \\
	&\qquad\ge \prod_{j = 2}^{k} \Pr_{M^{(j)}} \left[ \left. \left( M^{(1)}, \cdots M^{(j)} \right) \in \E_{j}\right| \left( M^{(1)}, \cdots M^{(j-1)} \right) \in \E_{j-1} \right] \cdot \Pr_{M^{(1)}}[M^{(1)} \in \E_1] \\
	&\qquad\ge (1-\frac{\delta}{2k})^{k} \ge 1-\delta,
	\end{align*}
	which completes the proof of the lemma.

\end{proof}

The following lemma shows that our sketch construction $\Pi^q$ presented in \cref{def:sketch} inherits the spectral property of Definition \ref{def:spectral-prp} from the base sketches, that is, if $\sbase$ and $\tbase$ are such that $I_{m^{q-2}} \times \sbase$ and $I_{d^{q-1}} \times \tbase$ satisfy the spectral property, then the sketch $\Pi^q$ satisfies the spectral property.

\begin{lem} \label{spectral-sketch-pi-t}
	For every positive integers $n, d, m$, any power of two integer $q$, any base sketch $\tbase: \RR^{d} \rightarrow \RR^m$ such that $I_{d^{q-1}} \times \tbase$ satisfies the $(2\mu_F+2 , 2\mu_2+2, O(\epsilon/q) , O(\delta/nq), n)$-spectral property, any $\sbase: \RR^{m^2} \rightarrow \RR^m$ such that $I_{m^{q-2}} \times \sbase$ satisfies the $(2\mu_F + 2 , 2\mu_2 + 2, O(\epsilon/q) , O(\delta/nq), n)$-spectral property, the sketch $\Pi^q $ defined as in Definition \ref{def:sketch} satisfies the $(\mu_F+1 , \mu_2+1, \eps , \delta, n)$-spectral property.
\end{lem}
\begin{proof}
	We wish to show that $\Pi^q = Q^q T^q$ as per Definition \ref{def:sketch}, satisfies the $(\mu_F+1 , \mu_2+1, \eps , \delta, n)$-spectral property. By Definition \ref{def:sketch-tilde} $Q^q = S^2 S^4 \cdots S^q$. Claim \ref{claim-tensor2product-reduction} shows that for every $l \in \{ 2, 4, \cdots q \}$ we can write,
	
	\begin{equation}
	S^l = M^l_{l/2} M^l_{l/2-1} \cdots M^l_1, \label{spectral-q-matrix-decomposition}
	\end{equation}
	where $M_j = I_{m^{q-2j}} \times S^{q}_{q/2 - j +1} \times I_{m^{j-1}}$ for every $j \in [q/2]$. From the discussion  in Definition \ref{def-dim-reduction} it follows that,
	
	\begin{equation}
	T^q = M'_q \cdots M'_1, \label{spectral-t-matrix-decomposition}
	\end{equation}
	where $M'_j = I_{d^{q-j}} \times T_{q-j+1} \times I_{m^{j-1}}$ for every $j \in [q]$. Therefore by combining \eqref{spectral-q-matrix-decomposition} and \eqref{spectral-t-matrix-decomposition} we get that,
	
	$$\Pi^q = M^{(2q+1)} M^{(2q)} \cdots M^{(1)},$$
	where $M^{(i)}$ matrices are independent and by the assumption of the lemma about the spectral property of $I_{m^{q-2}} \times \sbase$ and $I_{d^{q-1}} \times \tbase$ together with Claim \ref{claim:tensor-reordering} it follows that $M^{(i)}$ matrices satisfy the $(2\mu_F + 2 , 2\mu_2 + 2, O(\epsilon/q) , O(\delta/nq), n)$-spectral property.
	Therefore, the Lemma readily follows by invoking Lemma \ref{lem:spectral-property-composition} with $k=2q+1$.
\end{proof}

\subsection{Spectral Property of Identity\texorpdfstring{$ \mathbf{ \times}$}{ times }TensorSRHT}\label{sec:identity-tensorsrht}
In this section, we show that tensoring an identity operator with a {\sf TensorSRHT} sketch results in a transform that satisfies the spectral property defined in Definition~\ref{def:spectral-prp} with nearly optimal target dimension.

\begin{lem}\label{tensor-fast-jl-identity}
	Suppose $\epsilon, \delta, \mu_2,\mu_F > 0 $ and $n$ is a positive integer. If $m = \Omega\left( \log(\frac{n}{\delta}) \log^2(\frac{ndk}{\epsilon\delta}) \cdot \frac{\mu_F\mu_2}{\epsilon^2} \right)$ and $S \in \RR^{m \times d}$ is a {\sf TensorSRHT}, then the sketch $I_k \times S$ satisfies $(\mu_F, \mu_2 , \epsilon , \delta, n)$-spectral property.
\end{lem}

\begin{proof}
	Fix a matrix $U \in \RR^{kd \times n}$ with $\|U\|_F^2 \leq \mu_F$ and $\|U\|_{op}^2 \leq \mu_2$. Partition $U$ by rows into $d\times n$ submatrices $U_1, U_2, \dots, U_k$ such that $U^\top = \begin{bmatrix} U_1^\top & U_2^\top & \cdots & U_k^\top \end{bmatrix}$. Note that
	\[ U^\top (I_k \times S)^\top (I_k \times S) U = (U_1)^\top S^\top S U_1 + \cdots (U_k)^\top S^\top S U_k. \]
	The proof first considers the simpler case of a {\sf TensorSRHT} sketch of rank 1 and then applies the matrix Bernstein inequality from Lemma \ref{bernstein-matrix}. Let $R$ denote a rank one {\sf TensorSRHT} sketch. $R$ is a $1\times d$ matrix defined in Definition~\ref{def:tensorfastjl} by setting $m=1$ as follows, 
	$$R = P \cdot \left( H D_1 \times H D_2 \right),$$
	where $P \in \{0,1\}^{1 \times d}$ has one non-zero element whose position is uniformly distributed over $[d]$. Note that $S^\top S \in \RR^{d \times d}$, is the average of $m$ independent samples from $R^\top R$, i.e., $S^\top S = \frac{1}{m}\sum_{i\in [m]} R_i^\top R_i$, for i.i.d. $R_1, R_2, \dots, R_m \sim R$, and therefore, 
	$$U^\top (I_k \times S)^\top (I_k \times S) U = \frac{1}{m}\sum_{i\in [m]} U^\top (I_k \times R_i)^\top (I_k \times R_i) U.$$
	Therefore in order to use matrix Bernstein, Lemma \ref{bernstein-matrix}, we need to bound the maximum operator norm of $U^\top (I_k \times R)^\top (I_k \times R) U$ as well as the operator norm of its second moment.
	
	We proceed to upper bound the operator norm of $U^\top (I_k \times R)^\top (I_k \times R) U$. First, define the set
	$$
	\E:=\left\{(D_1,D_2): \left\| ( H D_1 \times HD_2) U_j^i \right\|_\infty^2 \leq 16 {\log^2(\frac{nd\mu_Fk}{\epsilon\delta})}) \cdot \|U_j^i\|_2^2\text{~for all~}j\in[k] \text{ and all } i\in[n]\right\},
	$$
	
	where $U^j_i$ is the $i$th column of $U^j$.
	By Claim \ref{i-infinity-HD-HD}, for every $i \in [n]$ and $j \in [k]$,
	$$\Pr_{D_1,D_2} \left[ \left\| ( H D_1 \times HD_2) U^j_i \right\|_\infty^2 \le 16 {\log^2(nd k/\delta)} \|U^j_i\|_2^2 \right] \ge 1-\epsilon\delta/(nk\mu_Fd).$$
	Thus, by a union bound over all $i \in [n]$ and $j \in [k]$, it follows that $\E$ occurs with probability at least $1-\epsilon\delta/(d\mu_F)$,
	$$\Pr_{D_1,D_2} [ (D_1, D_2) \in \E ] \ge 1 - \epsilon\delta/(d\mu_F),$$
	where the probability is over the random choice of $D_1,D_2$.
	
	From now on, we fix $(D_1,D_2)\in\E$ and proceed having conditioned on this event.
	
	\paragraph{Upper bounding $\left\| U^\top (I_k \times R)^\top (I_k \times R) U \right\|_{op}$.} From the fact that we have conditioned on $(D_1,D_2) \in \E$, note that
	\begin{align*}
	L \eqdef \left\| U^\top (I_k \times R)^\top (I_k \times R) U \right\|_{op} &= \|(U^1)^\top R^\top R U_1 + \cdots (U_k)^\top R^\top R U_k\|_{op}\\
	&\le \left\|(U_1)^\top R^\top R U_1 \right\|_{op} + \cdots + \left\|(U_k)^\top R^\top R U_k \right\|_{op}\\
	&= \left\| R U_1 \right\|_2^2 + \cdots + \left\| R U_k \right\|_2^2\\
	&\le 16 {\log^2(nd\mu_Fk/\epsilon\delta)} \cdot (\|U_1\|_F^2 + \cdots + \|U_k\|_F^2) \\
	&\leq 16 {\log^2(nd\mu_Fk/\epsilon\delta)} \cdot \|U\|_F^2 \\
	&= 16\mu_F \cdot {\log^2(nd\mu_Fk/\epsilon\delta)}),
	\end{align*}
	where the equality on the third line above holds because the matrices $(U^i)^\top R^\top R U^i$ are rank one.

	\paragraph{Upper bounding $\left\|\mathbb{E}_P\left[\left( U^\top (I_k \times R)^\top (I_k \times R) U \right)^2\right]\right\|_{op}$.}  For every $x\in \mathbb{R}^d$ with $\|x\|_2=1$, we have
	\begin{align}
	x^T\mathbb{E}_P\left[ \left( U^\top (I_k \times R)^\top (I_k \times R) U \right)^2 \right]x&= \mathbb{E}_P \left[ \sum_{j,j' \in [k]} x^T(U_j)^\top R^\top R U_j \cdot (U_{j'})^\top R^\top R U_{j'}x \right]\nonumber\\
	&\leq  \mathbb{E}_P \left[ \sum_{j,j' \in [k]} |R U_jx| \|R U_j\|_2 |R U_{j'} x|\|R U_{j'}\|_2 \right]\nonumber\\
	&=  \mathbb{E}_P \left[ \left(\sum_{j \in [k]} |R U_jx| \|R U_j\|_2\right)^2 \right]\nonumber\\
	&\leq  \mathbb{E}_P \left[ \left(\sum_{j\in [k]} (R U_jx)^2\right) \left(\sum_{j\in [k]} \|R U_j\|_2^2\right) \right],\nonumber
	\end{align}
	where the second and fourth lines follow from the Cauchy-Schwarz inequality. Using the fact that we conditioned on $(D_1,D_2) \in \E$, we get
	
	\begin{align}
	x^T\mathbb{E}_P\left[ \left( U^\top (I_k \times R)^\top (I_k \times R) U \right)^2 \right]x &\leq  16 \log^2 (nd\mu_Fk/\epsilon\delta) \left(\sum_{j\in[k]} \|U_j\|_F^2\right) \ex_P\left[ \sum_{j\in[k]} (R U_jx)^2 \right] \nonumber\\
	&= 16\log^2 (nd\mu_Fk/\epsilon\delta) \left(\sum_{j\in[k]} \|U_j\|_F^2\right) \sum_{j\in[k]} \mathbb{E}_P\left[ (P (HD_1 \times HD_2) U_j x)^2 \right] \nonumber\\
	&=  16\log^2 (nd\mu_Fk/\epsilon\delta)\cdot \|U\|_F^2 \sum_{j\in[k]} \left\| U_jx\right\|_2^2 \nonumber\\
	&=  16 \log^2 (nd\mu_Fk/\epsilon\delta)\cdot \|U\|_F^2 \|U x\|_2^2\nonumber\\
	&\le  16 \log^2 (nd\mu_Fk/\epsilon\delta)\cdot \mu_F \mu_2,\nonumber
	\end{align}
	since $ \mathbb{E}_P\left[ (P (HD_1 \times HD_2) U_j x)^2 \right] = \frac{1}{d} \| (HD_1 \times HD_2) U_jx \|^2 = \|U_jx\|_2^2$ for all $x$.

	Since the matrix $\mathbb{E}_P\left[ \left( U^\top (I_k \times R)^\top (I_k \times R) U \right)^2 \right]$ is positive semi-definite for any fixed $D_1$ and $D_2$, it follows that
	\begin{align*}
	M \eqdef \left\| \ex_P \left[ \left( U^\top (I_k \times R)^\top (I_k \times R) U \right)^2 \right] \right\|_{op} &\le 16\log^2 (nd\mu_Fk/\epsilon\delta)\cdot \mu_F \mu_2.
	\end{align*}
	
	\paragraph{Combining one-dimensional TensorSRHT sketches.}
	To conclude, we note that the Gram matrix of a {\sf TensorSRHT}, $S^\top S \in \RR^{d \times d}$, is the average of $m$ independent samples from $R^\top R$, i.e., $S^\top S = \frac{1}{m}\sum_{i\in [m]} R_i^\top R_i$, for i.i.d. $R_1, R_2, \dots, R_m \sim R$, and therefore, 
	$$(I_k \times S)^\top (I_k \times S) = \frac{1}{m}\sum_{i\in [m]} (I_k \times R_i)^\top (I_k \times R_i).$$
	Recall that $(D_1,D_2) \in \E$ occurs with probability at least $1-\epsilon\delta/(d\mu_F)$, therefore we have the following for the conditional expectation $\mathbb{E}\left[ \left. U^\top (I_k \times R)^\top (I_k \times R) U \right| (D_1,D_2)\in \E \right]$,
	$$\mathbb{E}\left[ \left. U^\top (I_k \times R)^\top (I_k \times R) U \right| (D_1,D_2) \in \E \right] \preceq \frac{\mathbb{E} \left[ U^\top (I_k \times R)^\top (I_k \times R) U \right]}{\Pr[(D_1,D_2) \in \E]} \preceq \frac{U^\top U}{1 - \epsilon\delta/(d \mu_F)}.$$
	And also by Cauchy-Schwarz we have,
	\begin{align*}
	&\mathbb{E}\left[ \left. U^\top (I_k \times R)^\top (I_k \times R) U \right| (D_1,D_2) \in \E \right]\\
	&\qquad \succeq \mathbb{E} \left[ U^\top (I_k \times R)^\top (I_k \times R) U \right] - \mathbb{E} \left[ \left. U^\top (I_k \times R)^\top (I_k \times R) U \right| (D_1,D_2) \in \bar{\E} \right] \cdot \Pr[\bar{\E}]\\
	&\qquad \succeq U^\top U - d\|U\|_F^2  \Pr[\bar{\E}] \cdot I_n\\
	&\qquad \succeq U^\top U - d\|U\|_F^2 \cdot  \epsilon\delta/(d\mu_F) \cdot I_n\\
	&\qquad \succeq U^\top U - (\epsilon/2) \cdot I_n.
	\end{align*}
	These two bounds together imply that,
	$$\left\|\mathbb{E} \left[ \left. U^\top (I_k \times R)^\top (I_k \times R) U \right| (D_1,D_2) \in \E \right] - U^\top U \right\|_{op} \le \epsilon/2.$$
	
	Now note that the random variables $R_i^\top R_i$ are independent conditioned on $(D_1,D_2) \in \E$. Hence, using the upper bounds $L \le 16\mu_F\cdot {\log^2(nd\mu_Fk/\epsilon\delta)}$ and $M \le 16\mu_F \mu_2\cdot\log^2 (nd\mu_Fk/\epsilon\delta)$, which hold when $(D_1,D_2) \in \E$, we have the following by Lemma \ref{bernstein-matrix}, (here we drop the subscript from $I_k$ for ease of notation)
	\begin{align*}
	&\Pr_{P,D_1,D_2} \left[ \left\| U^\top (I \times S)^\top (I \times S) U - U^\top U \right\|_{op} \ge \epsilon \right]\\ 
	&\qquad \le \Pr_P \left[ \left. \left\| U^\top (I \times S)^\top (I \times S) U - \mathbb{E} \left[ \left. U^\top (I \times R)^\top (I \times R) U \right| (D_1,D_2) \in \E \right] \right\|_{op} \ge \epsilon/2 \, \right| (D_1,D_2) \in \E \right]\\ 
	&\qquad\qquad + \Pr_{D_1,D_2}[\bar{\mathcal{E}}]\\
	&\qquad\le 8n \cdot \exp\left(- \frac{m \epsilon^2 /2}{M + 2\epsilon L/3}\right) + \delta/2\\
	&\qquad\le \delta,
	\end{align*}
	where the last inequality follows by setting $m = \Omega\left( \log(n/\delta) \log^2(ndk/\epsilon\delta) \cdot \mu_F \mu_2  /\epsilon^2 \right)$. This shows that $I_k \times S$ satisfies the $(\mu_F, \mu_2, \epsilon, \delta, n)$-spectral property.
\end{proof}

\subsection{Spectral property of Identity\texorpdfstring{$\mathbf{ \times}$}{ times }OSNAP}\label{sec:identity-osnap}
In this section, we show that tensoring identity operator with {\sf OSNAP} sketch (Definition \ref{def:osnap}) results in a transform which satisfies the spectral property (Definition \ref{def:spectral-prp}) with nearly optimal target dimension as well as nearly optimal application time. This sketch is particularly efficient for sketching sparse vectors. We use a slightly different sketch than the original {\sf OSNAP} to simplify the analysis, defined as follows.

\begin{defn}[{\sf OSNAP} transform]\label{def:osnapprim}
	For every sparsity parameter $s$, target dimension $m$, and positive integer $d$, the {\sf OSNAP} transform with sparsity parameter $s$ is defined as,

$$S_{r,j}= \sqrt{\frac{1}{s}} \cdot \delta_{r,j} \cdot \sigma_{r,j},$$
for all $r \in [m]$ and all $j \in [d]$, where $\sigma_{r,j} \in \{-1, +1\}$ are independent and uniform Rademacher random variables and $\delta_{r,j}$ are independent Bernoulli random variables satisfying, $\Ep{\delta_{r,i}} = s/m$ for all $r \in [m]$ and all $i \in [d]$.
\end{defn}

\begin{lem}\label{osnap-identity}
	Suppose $\epsilon, \delta, \mu_2, \mu_F > 0$ and $n$ is a positive integer. If $S \in \RR^{m \times d}$ is a {\sf OSNAP} sketch with sparsity parameter $s$, then the sketch $I_k \times S$ satisfies the $(\mu_F, \mu_2 , \epsilon , \delta, n)$-spectral property, provided that $s = \Omega\left( \log^2(n d k /\epsilon\delta) \log(n/\delta) \cdot \frac{\mu_2^2}{\epsilon^2} \right)$ and $m = \Omega\left( (\mu_F \mu_2/\epsilon^2) \cdot \log^2(n d k /\epsilon\delta) \right)$.
\end{lem}

\begin{proof}
	Fix a matrix $U \in \RR^{kd \times n}$ with $\|U\|_F^2 \leq \mu_F$ and $\|U\|_{op}^2 \leq \mu_2$. Partition $U$ by rows into $d\times n$ sub-matrices $U_1, U_2, \dots, U_k$ such that $U^T = \begin{bmatrix} U_1^\top & U_2^\top & \cdots & U_k^\top \end{bmatrix}$. Note that
	\[ U^\top (I_k \times S)^\top (I_k \times S) U = (U_1)^\top S^\top S U_1 + \cdots (U_k)^\top S^\top S U_k. \]
	The proof first considers the simpler case of an {\sf OSNAP} sketch of rank 1 and then applies the matrix Bernstein bound. Let $R$ denote a rank one {\sf OSNAP} sketch. $R$ is a $1\times d$ matrix defined as follows, 
	\begin{equation}
R_{i}= \sqrt{\frac{m}{s}} \cdot \delta_{i} \sigma_{i}, \label{osnap-rank1-def}
	\end{equation}
	where $\sigma_{i}$ for all $i \in [d]$ are independent Rademacher random variables and also, $\delta_{i}$ for all $i \in [d]$ are independent Bernoulli random variables for which the probability of being one is equal to $\frac{s}{m}$.
	
	We proceed to upper bound the operator norm of $U^\top (I_k \times R)^\top (I_k \times R) U$. First, define the set
	
	$$\E:=\left\{R: (R  U_j)^\top R  U_j \preceq C \left(\frac{m}{s} {\log^2(\frac{n d k \mu_F}{\epsilon\delta})} \cdot U_j^\top U_j + \log(\frac{n d k \mu_F}{\epsilon\delta})\|U_j\|_F^2 \cdot I_n\right) \text{~for all~}j=1,\ldots, k \right\},$$
	where $C > 0$ is a large enough constant. We show that,
	$$\Pr [R \in \E ] \ge 1 - \epsilon\delta/{(dm\mu_F)},$$
	where the probability is over the random choices of $\{\sigma_{i}\}_{i\in[d]}$ and $\{\delta_{i}\}_{i\in[d]}$. To show this we first prove the following claim,
	\begin{claim}\label{claim:rank1-osnap}
		For every matrix $Z \in \RR^{d \times n}$, if we let $R$ be defined as in \eqref{osnap-rank1-def}, then,
		$$\Pr\left[ Z^\top R^\top R Z \preceq C \left(\frac{m}{s} \cdot {\log^2(n/\delta)} Z^\top Z + \log(n/\delta)\|Z\|_F^2 I_n\right) \right] \ge 1 - \delta. $$
	\end{claim}

	\begin{proof}
		The proof is by Matrix Bernstein inequality, Lemma \ref{lem:Bernstein}. For any matrix $Z$ let $A = Z(Z^\top Z + \mu I_n)^{-1/2}$, where $\mu = \frac{s}{m} \frac{1}{\log(n/\delta)} \|Z\|_F^2$. We can write $R A = \sqrt{\frac{m}{s}}\sum_{i \in [d]} \delta_i \sigma_i A_i$, where $A_i$ is the $i$th row of $A$. Note that $\mathbb{E}[\delta_i \sigma_i A_i] = 0$ and $\|\delta_i \sigma_i A_i\|_2 \le \|A_i\|_2 \le \|A\|_{op}$. Also note that 
		$$ \sum_{i \in [d]} \mathbb{E} [(\delta_i \sigma_i A_i) (\delta_i \sigma_i A_i)^*] = \sum_{i \in [d]} \frac{s}{m}\|A_i\|_2^2 = \frac{s}{m}\|A\|_F^2$$
		and, 
		$$\sum_{i \in [d]}\mathbb{E} [(\delta_i \sigma_i A_i)^* (\delta_i \sigma_i A_i)] = \sum_{i \in [d]}\frac{s}{m}A_i^* A_i = \frac{s}{m} A^\top A.$$
		Therefore, 
		$$\max\left\{ \left\| \sum_{i \in [d]} \mathbb{E} [(\delta_i \sigma_i A_i) (\delta_i \sigma_i A_i)^*] \right\|_{op}, \left\| \sum_{i \in [d]}\mathbb{E} [(\delta_i \sigma_i A_i)^* (\delta_i \sigma_i A_i)] \right\|_{op} \right\} \le \frac{s}{m}\|A\|_F^2.$$
		By Lemma \ref{lem:Bernstein},
		$$\Pr\left[ \left\|\sum_{i \in [d]} \delta_i \sigma_i A_i \right\|_{op} \ge t \right] \le (n+1) \cdot \exp \left( \frac{-t^2/2}{\frac{s}{m}\|A\|_F^2 + \|A\|_{op} t/3} \right),$$
		hence if $t = C'/2 \cdot \left( \sqrt{\frac{s}{m} \log(n/\delta)}\|A\|_F + \log(n/\delta) \|A\|_{op} \right)$, then $\Pr\left[ \left\|\sum_{i \in [d]} \delta_i \sigma_i A_i \right\|_{op} \ge t \right] \le \delta$. By plugging $\|R A\|_2^2 = \frac{m}{s} \cdot \| \sum_{i \in [d]} \delta_i \sigma_i A_i \|_2^2$ into the above we get the following,
		$$\Pr\left[ \| R A \|_{op}^2 \le C'^2/2 \left(\frac{m}{s} \cdot {\log^2(n/\delta)}\|A\|_{op}^2 + \log(n/\delta)\|A\|_F^2 \right) \right] \ge 1 - \delta.$$
		Now note that for the choice of $A = Z(Z^\top Z + \mu I_n)^{-1/2}$, we have $\|A\|_{op}^2 \le \frac{\|Z^\top Z\|_{op}}{\|Z^\top Z\|_{op}^2 + \mu} \le 1$ and also $\|A\|_F^2 = \sum_{i} \frac{\lambda_i(Z^\top Z)}{\lambda_i(Z^\top Z) + \mu} \le \frac{\sum_{i} \lambda_i(Z^\top Z)}{\mu} = \frac{m}{s} \log(n/\delta)$. By plugging these into the above we get that,
		$$\Pr\left[ \left\| R Z(Z^\top Z + \mu I_n)^{-1/2} \right\|_{op}^2 \le C'^2 \frac{m}{s} \cdot {\log^2(n/\delta)} \right] \ge 1 - \delta.$$
		Hence,
		$$(Z^\top Z + \mu I_n)^{-1/2} Z^\top R^\top R Z(Z^\top Z + \mu I_n)^{-1/2} \preceq C \frac{m}{s} \cdot {\log^2(n/\delta)} I_n,$$
		with probability $1-\delta$, where $C = C'^2$. Multiplying both sides of the above from left and right by the positive definite matrix $(Z^\top Z + \mu I_n)^{1/2}$ gives (recall that $\mu=\frac{s}{m} \cdot \frac{\|Z\|_F^2}{\log(n/\delta)}$),
		$$Z^\top R^\top R Z \preceq C \left(\frac{m}{s} \cdot {\log^2(n/\delta)} Z^\top Z + \log(n/\delta)\|Z\|_F^2 I_n\right).$$
	\end{proof}
	By applying Claim \ref{claim:rank1-osnap} with failure probability of $\epsilon\delta/{(dk\mu_F)}$ on each of $U_j$'s and then applying a union bound, we get the following,
	$$\Pr [R \in \E ] \ge 1 - \epsilon\delta/{(dm\mu_F)}.$$
	
	From now on, we fix $R\in\E$ and proceed having conditioned on this event. 
	\paragraph{Upper bounding $\left\| U^\top (I_k \times R)^\top (I_k \times R) U \right\|_{op}$.} From the fact that we have conditioned on $R\in\E$, note that,
	\begin{align*}
	L \eqdef \left\| U^\top (I_k \times R)^\top (I_k \times R) U \right\|_{op} &= \|(U_1)^\top R^\top R U_1 + \cdots (U_k)^\top R^\top R U_k\|_{op}\\
	&\le \left\| \sum_{i \in [k]} C \left(\frac{m}{s} \cdot {\log^2(n d k \mu_F/\epsilon\delta)} \cdot U_j^\top U_j + \log(n d k \mu_F/\epsilon\delta)\|U_j\|_F^2 \cdot I_n\right) \right\|_{op}\\
	&= \left\| C \left(\frac{m}{s} \cdot {\log^2(n d k \mu_F/\epsilon\delta)} \cdot U^\top U + \log(n d k \mu_F/\epsilon\delta)\|U\|_F^2 \cdot I_n\right) \right\|_{op}\\
	&\le C \left(\frac{m}{s} \cdot {\log^2(n d k \mu_F/\epsilon\delta)} \cdot \|U\|_{op}^2 + \log(n d k \mu_F/\epsilon\delta)\|U\|_F^2 \right) \\
	&\leq C\left(\frac{m}{s} \mu_2 \cdot {\log^2(n d k \mu_F/\epsilon\delta)} + \mu_F \cdot \log(n d k \mu_F/\epsilon\delta) \right).
	\end{align*}

	\paragraph{Upper bounding $\left\|\mathbb{E}\left[\left( U^\top (I_k \times R)^\top (I_k \times R) U \right)^2\right]\right\|_{op}$.}  From the condition $R\in\E$, it follows that
	\begin{align}
	&\mathbb{E}\left[ \left( U^\top (I_k \times R)^\top (I_k \times R) U \right)^2 \right] \nonumber\\
	&\qquad\preceq \mathbb{E} \left[ C \left(\frac{m}{s} \cdot {\log^2(n d k \mu_F/\epsilon\delta)} \cdot U^\top U + \log(n d k \mu_F/\epsilon\delta)\|U\|_F^2 \cdot I_n \right) \left( U^\top (I_k \times R)^\top (I_k \times R) U \right) \right]\nonumber\\
	&\qquad\preceq C \left(\frac{m}{s} \cdot {\log^2(n d k \mu_F/\epsilon\delta)} \cdot U^\top U + \log(n d k \mu_F/\epsilon\delta)\|U\|_F^2 \cdot I_n \right)\mathbb{E} \left[ \left( U^\top (I_k \times R)^\top (I_k \times R) U \right) \right]\nonumber\\
	&\qquad\preceq C \left(\frac{m}{s} \cdot {\log^2(n d k \mu_F/\epsilon\delta)} \cdot U^\top U + \log(n d k \mu_F/\epsilon\delta)\|U\|_F^2 \cdot I_n \right) \cdot \frac{U^\top U}{1-\epsilon\delta/(dm\mu_F)} \nonumber
	\end{align}
	where the last line follows from the fact that the random variable $U^\top (I_k \times R)^\top (I_k \times R) U $ is positive semidefinite and the conditional expectation can be upper bounded by its unconditional expectation as follows, 
	$$\mathbb{E} \left[ \left. U^\top (I_k \times R)^\top (I_k \times R) U \right| R \in \E \right] \preceq \frac{\mathbb{E} \left[ U^\top (I_k \times R)^\top (I_k \times R) U \right]}{\Pr[R \in \E]}.$$
	Therefore we can bound the operator norm of the above as follows,
	\begin{align}
	M &\eqdef \left\|\mathbb{E}\left[\left( U^\top (I_k \times R)^\top (I_k \times R) U \right)^2\right]\right\|_{op}\nonumber\\ 
	&\leq 2\left\|C \left(\frac{m}{s} \cdot {\log^2(n d k \mu_F/\epsilon\delta)} \cdot (U^\top U)^2 + \log(n d k \mu_F/\epsilon\delta)\|U\|_F^2 \cdot U^\top U \right) \right\|_{op} \nonumber\\
	&\le  2C\left(\frac{m}{s} \cdot {\log^2(n d k \mu_F/\epsilon\delta)} \cdot \|U^\top U \|_{op}^2 + \log(n d k \mu_F/\epsilon\delta)\|U\|_F^2 \cdot \| U^\top U \|_{op} \right) \nonumber\\
	&=  2C\left(\frac{m}{s} \cdot {\log^2(n d k \mu_F/\epsilon\delta)} \cdot \mu_2^2 + \log(n d k \mu_F/\epsilon\delta)\mu_F \mu_2 \right).\nonumber
	\end{align}
	
	\paragraph{Combining one-dimensional OSNAP transforms.}
	To conclude, we note that the Gram matrix of an {\sf OSNAP} sketch, $S^\top S \in \RR^{d\times d}$, is the average of $m$ independent samples from $R^\top R$ with $R$ defined as in \eqref{osnap-rank1-def} -- i.e., $S^\top S = \frac{1}{m}\sum_{i\in [m]} R_i^\top R_i$ for i.i.d. $R_1, R_2, \dots, R_m \sim R$, and therefore, 
	$$(I_k \times S)^\top (I_k \times S) = \frac{1}{m}\sum_{i\in [m]} (I_k \times R_i)^\top (I_k \times R_i).$$
	Note that by a union bound $R_i \in \E$ simultaneously for all $i \in [m]$ with probability at least $1-\epsilon\delta/(d\mu_F)$. Now note that the random variables $R_i^\top R_i$ are independent conditioned on $R_i \in \E$ for all $i \in [m]$. Also note that the conditional expectation $\mathbb{E} \left[ \left. U^\top (I_k \times R)^\top (I_k \times R) U \right| R \in \E \right]$ satisfies the following,
	\begin{align*}
	&\mathbb{E} \left[ \left. U^\top (I_k \times R)^\top (I_k \times R) U \right| R \in \E \right]\\ 
	&\qquad \succeq \mathbb{E} \left[ U^\top (I_k \times R)^\top (I_k \times R) U \right] - \mathbb{E} \left[ \left. U^\top (I_k \times R)^\top (I_k \times R) U \right| R \in \bar{\E} \right] \cdot \Pr[\bar{\E}]\\
	&\qquad \succeq U^\top U - d\|U\|_F^2  \Pr[\bar{\E}] \cdot I_n\\
	&\qquad \succeq U^\top U - d\|U\|_F^2 \cdot  \epsilon\delta/(d\mu_F) \cdot I_n\\
	&\qquad \succeq U^\top U - d\|U\|_F^2 \cdot  \epsilon/2 \cdot I_n.
	\end{align*}
	We also have,
	$$\mathbb{E} \left[ \left. U^\top (I_k \times R)^\top (I_k \times R) U \right| R \in \E \right] \preceq \frac{\mathbb{E} \left[ U^\top (I_k \times R)^\top (I_k \times R) U \right]}{\Pr[R \in \E]} \preceq \frac{U^\top U}{1 - \epsilon\delta/(d \mu_F)}.$$
	These two bounds together imply that,
	$$\left\|\mathbb{E} \left[ \left. U^\top (I_k \times R)^\top (I_k \times R) U \right| R \in \E \right] - U^\top U \right\|_{op} \le \epsilon/2.$$
	
	Now, using the upper bounds $L \leq C\left(\frac{m}{s} \mu_2 \cdot {\log^2(n d k \mu_F/\epsilon\delta)} + \mu_F \cdot \log(n d k \mu_F/\delta) \right)$ and $M \le 2C\left(\frac{m}{s} \cdot {\log^2(n d k \mu_F/\delta)} \cdot \mu_2^2 + \log(n d k \mu_F/\delta)\mu_F \mu_2 \right)$, which hold when $R \in \E$, we have that by Lemma \ref{bernstein-matrix},
	\begin{align*}
	&\Pr \left[ \left\| U^\top (I_k \times S)^\top (I_k \times S) U - U^\top U \right\|_{op} \ge \epsilon \right]\\ 
	&\qquad \le \Pr \left[ \left\| U^\top (I_k \times S)^\top (I_k \times S) U - \mathbb{E} \left[ \left. U^\top (I_k \times R)^\top (I_k \times R) U \right| R \in \E \right] \right\|_{op} \ge \epsilon/2 \,\big\vert\, \mathcal{E} \right] + \Pr_D[\bar{\mathcal{E}}]\\
	&\qquad\le 8n \cdot \exp\left(- \frac{m \epsilon^2 /8}{M + \epsilon L/3}\right) + \delta/2 \le \delta,
	\end{align*}
	where the last inequality follows by setting $s = \Omega\left( \log^2(n d k \mu_F/\epsilon\delta) \log(nd/\delta) \cdot \frac{\mu_2^2}{\epsilon^2} \right)$ and $m = \Omega\left( \mu_F \mu_2/\epsilon^2 \cdot \log^2(n d k \mu_F/\epsilon\delta) \right)$. This shows that $I_k \times S$ satisfies the $(\mu_F, \mu_2, \epsilon, \delta, n)$-spectral property.
\end{proof}

\subsection{High Probability OSE with linear dependence on \texorpdfstring{$s_\lambda$}{s-lambda}}\label{sec:putt-together-highprob-ose}
We are ready to prove Theorem \ref{high-prob-sketch}. We prove that if we instantiate $\Pi^p$ from Definition \ref{def:sketch} with $\tbase: $ {\sf OSNAP} and $\sbase: $ {\sf TensorSRHT}, it satisfies the statement of Theorem \ref{high-prob-sketch}.

\highprobsketch*
\begin{proof} 
	Let $\delta = \frac{1}{\poly{n}}$ denote the failure probability. Let $m \approx p^4 \log_2^3(\frac{nd}{\eps\delta}) \cdot \frac{s_\lambda}{\eps^2}$ and $s \approx \frac{p^4}{\eps^2} \cdot \log_2^3(\frac{nd}{\eps\delta})$ be integers.
	Let $\Pi^p \in \RR^{m\times m^p}$ be the sketch defined in Definition \ref{def:sketch}, where $\sbase \in \RR^{m\times m^2}$ is a {\sf TensorSRHT} sketch and $\tbase \in \RR^{m \times d}$ is an {\sf OSNAP} sketch with sparsity parameter $s$.
	
	Let $q = 2^{\lceil \log_2(p) \rceil}$. By Lemma \ref{lem:reduce-p-to-q}, it is sufficient to show that $\Pi^q$ is a $(\eps, \delta, s_\lambda, d^q, n)$-Oblivious Subspace Embedding.
	Consider arbitrary $A \in \RR^{d^q \times n}$ and $\lambda > 0$. Let us denote the statistical dimension of $A$ by $s_\lambda = s_\lambda(A^\top A)$. Let $U = A \left( A^\top A + \lambda I_n\right)^{-1/2}$. Therefore, $\|U\|_2 \le 1$ and $\|U\|_F^2 = s_\lambda$. Since $q < 2p$, by Lemma \ref{osnap-identity}, the transform $I_{d^{q-1}} \times \tbase$, satisfies $(2 s_\lambda + 2, 2, O(\eps/q), O(\delta/n^2q),n)$-spectral property. Moreover, by Lemma \ref{tensor-fast-jl-identity}, the transform $I_{m^{q-2}} \times \sbase$ satisfies $(5 s_\lambda + 9, 9, O(\eps/q), O(\delta/n^2q^2),n)$-spectral property. Therefore, by Lemma \ref{spectral-sketch-pi-t}, the sketch $\Pi^q$ satisfies $( s_\lambda + 1, 1, \eps, \delta,n)$-spectral property, hence,

$$\Pr\left[ \left\| (\Pi^q U)^\top \Pi^q U - U^\top U \right\|_{op} \le \eps \right] \ge 1-\delta.$$
Since $U^\top U = (A^\top A + \lambda I_n)^{-1/2} A^\top A (A^\top A + \lambda I_n)^{-1/2}$ and $\Pi^q U = \Pi^p A (A^\top A + \lambda I_n)^{-1/2}$ we have the following,

$$\Pr\left[(1-\epsilon)(A^\top A + \lambda I_n) \preceq (\Pi^p A)^\top \Pi^p A + \lambda I_n \preceq (1+\epsilon) (A^\top A + \lambda I_n) \right] \ge 1-\delta.$$

\paragraph{Runtime:} By Lemma \ref{algo-sketch-correctness}, for any $\sbase$ and $\tbase$, if $A$ is the matrix whose columns are obtained by $p$-fold self-tensoring of each column of some $X \in \RR^{d \times n}$ then the sketched matrix $\Pi^p A$ can be computed using Algorithm \ref{alg:main}. When $\sbase$ is {\sf TensorSRHT} and $\tbase$ is {\sf OSNAP}, the runtime of Algorithm \ref{alg:main} for a fixed vector $w \in \RR^d$ is as follows; Computing $Y^0_j$'s for each $j$ in lines~\ref{y-0-q} and \ref{y-0-p} of algorithm requires applying an {\sf OSNAP} sketch on $w \in \RR^{d}$ which on expectation takes time $O(s \cdot \text{nnz}(w))$. Therefore computing all $Y^0_j$'s takes time $O(q s \cdot \text{nnz}(w))$.

Computing each of $Y^l_j$'s in line~\ref{Y-l-j} of algorithm amounts to applying a {\sf TensorSRHT} of input dimension $m^2$ and target dimension of $m$ on $Y^{l-1}_{2j-1} \otimes Y^{l-1}_{2j}$. This takes time $O( m \log m)$. Therefore computing all the $Y^l_j$'s takes time $O(q \cdot m \log m)$. Note that $q \le 2p$ hence the total time of running Algorithm \ref{alg:main} on a vector $w$ is $O(p \cdot m\log_2m + ps \cdot \text{nnz}(w))$. Therefore, sketching $n$ columns of a matrix $X \in \RR^{d \times n}$ takes time $O(p (nm\log_2m + s \cdot \text{nnz}(X)))$.

\end{proof}

\section{Oblivious Subspace Embedding for the Gaussian Kernel} \label{sec:gaussian-ose}

In this section we show how to sketch the Gaussian kernel matrix by polynomial expansion and then applying our proposed sketch for the polynomial kernels.

\paragraph{Data-points with bounded $\ell_2$ radius:}
Suppose that we are given a dataset of points $x_1, \cdots x_n \in \RR^d$ such that for all $i\in[n]$, $\| x_i \|_2^2 \le r$ for some positive value $r$. Consider the Gaussian kernel matrix $G\in\RR^{n \times n}$ defined as $G_{i,j} = e^{-\|x_i - x_j\|_2^2/2}$ for all $i,j \in [n]$. We are interested in sketching the data-points matrix $X$ using a sketch $S_g : \RR^{d} \rightarrow \RR^{m}$ such that the following holds with probability $1-\delta$,
$$(1-\epsilon)(G + \lambda I_n) \preceq (S_g(X))^\top S_g(X) + \lambda I_n \preceq (1+\epsilon)(G + \lambda I_n).$$

\gaussiansketch*

\begin{proof} Let $\delta = \frac{1}{\poly{n}}$ denote the failure probability.
Note that $G_{i,j} = e^{-\|x_i\|_2^2/2} \cdot e^{x_i^\top x_j} \cdot e^{-\|x_j\|_2^2/2}$ for every $i , j \in [n]$. Let $D$ be a $n\times n$ diagonal matrix with $i$th diagonal entry $e^{-\|x_i\|_2^2/2}$ and let $K \in \RR^{n \times n}$ be defined as $K_{i,j} = e^{x_i^\top x_j}$ (note that $D K D = G$). Note that $K$ is a positive definite kernel matrix. The Taylor series expansion for kernel $K$ is as follows,
$$K = \sum_{l=0}^{\infty} \frac{(X^{\otimes l})^\top X^{\otimes l}}{l!}.$$
Therefore $G$ can be written as the following series,
$$G = \sum_{l=0}^{\infty} \frac{(X^{\otimes l} D )^\top X^{\otimes l} D}{l!}.$$

Note that each of the terms $(X^{\otimes l} D )^\top X^{\otimes l} D = D (X^{\otimes l} )^\top X^{\otimes l} D$ are positive definite kernel matrices.
The statistical dimension of kernel $(X^{\otimes l} D )^\top X^{\otimes l} D$ for every $l \ge 0$ is upper bounded by the statistical dimension of kernel $G$ through the following claim.
\begin{claim}\label{claim:statdim}
	For every $\mu \ge 0$ and every integer $l $,
	$$s_\mu \left( (X^{\otimes l} D )^\top X^{\otimes l} D \right) \le s_\mu(G).$$
\end{claim}
\begin{proof}
	From the Taylor expansion $G = \sum_{l=0}^{\infty} \frac{(X^{\otimes l} D )^\top X^{\otimes l} D}{l!}$ along with the fact that the polynomial kernel of any degree is positive definite, we have that $(X^{\otimes l} D)^\top X^{\otimes l} D \preceq G$. Now, by Courant-Fischer's min-max theorem we have that,
	$$\lambda_j((X^{\otimes l} D)^\top X^{\otimes l} D) = \max_{U \in \RR^{(j-1) \times n}} \min_{\substack{\alpha\neq 0\\ U\alpha=0}} \frac{\alpha^\top (X^{\otimes l} D)^\top X^{\otimes l} D \alpha}{\|\alpha\|_2^2}.$$
	Let $U^*$ be the maximizer of the expression above. Then we have,
	\begin{align*}
	\lambda_j(G) &= \max_{U \in \RR^{(j-1) \times n}} \min_{\substack{\alpha\neq 0\\ U\alpha=0}} \frac{\alpha^\top G \alpha}{\|\alpha\|_2^2}\\
	&\ge \min_{\substack{\alpha\neq 0\\ U^*\alpha=0}} \frac{\alpha^\top G \alpha}{\|\alpha\|_2^2}\\
	&\ge \min_{\substack{\alpha\neq 0\\ U^*\alpha=0}} \frac{\alpha^\top (X^{\otimes l} D)^\top X^{\otimes l} D \alpha}{\|\alpha\|_2^2}\\
	&= \lambda_j((X^{\otimes l} D)^\top X^{\otimes l} D).
	\end{align*}
	for all $j$. Therefore, the claim follows from the definition of statistical dimension, 
	$$ s_\mu(G) = \sum_{j=1}^{n} \frac{\lambda_j(G)}{\lambda_j(G) + \mu} \ge \sum_{j=1}^{n} \frac{\lambda_j((X^{\otimes l} D )^\top X^{\otimes l} D)}{\lambda_j((X^{\otimes l} D )^\top X^{\otimes l} D) + \mu} = s_\mu \left( (X^{\otimes l} D )^\top X^{\otimes l} D \right).$$
\end{proof}

If we let $P = \sum_{l=0}^{q} \frac{(X^{\otimes l} )^\top X^{\otimes l} }{l!}$, where $q = C\cdot(r^2 + \log ( \frac{n}{\epsilon \lambda}))$ for some constant $C$, then by the triangle inequality we have

\begin{align*}
\left\| K - P \right\|_{op} &\le \sum_{l >q} \left\| \frac{(X^{\otimes l} )^\top X^{\otimes l} }{l!} \right\|_{op}\\
&\le \sum_{l >q} \left\| \frac{(X^{\otimes l} )^\top X^{\otimes l} }{l!} \right\|_{F}\\
&\le \sum_{l >q} \frac{n \cdot r^{2l}}{l!}\\
&\le \epsilon\lambda/2.
\end{align*}

$P$ is a positive definite kernel matrix. Also note that all the eigenvalues of the diagonal matrix $D$ are bounded by $1$. Hence, in order to get a subspace embedding it is sufficient to satisfy the following with probability $1-\delta$,
$$(1-\epsilon/2)(D P D + \lambda I_n) \preceq (S_g(X))^\top S_g(X) + \lambda I_n \preceq (1+\epsilon/2)( D P D + \lambda I_n).$$

Let the sketch $\Pi^l \in \RR^{m_l \times d^{l}}$ be the sketch from Theorem \ref{high-prob-sketch} therefore by Claim \ref{claim:statdim} we get the following guarantee on $\Pi^l$:

\begin{equation}
(1-\frac{\epsilon}{9})( (X^{\otimes l} D )^\top X^{\otimes l} D + \lambda I_n ) \preceq (\Pi^l X^{\otimes l} D)^\top \Pi^l X^{\otimes l} D + \lambda I_n \preceq (1+\frac{\epsilon}{9}) ( (X^{\otimes l} D )^\top X^{\otimes l} D + \lambda I_n ), \label{spectralguarantee-one-polynomial-term}
\end{equation}

with probability $1-\frac{\delta}{q+1}$ as long as $m_l = \Omega\left( l^4 \log^3(nd/\delta) \cdot s_\lambda  / \epsilon^2 \right)$ and moreover $\Pi^l X^{\otimes l} D$ can be computed using $O\left( n \cdot l \cdot m_l\log_2m_l + \frac{l^5}{\epsilon^2} \cdot \log^3(nd/\delta) \cdot \text{nnz}(X) \right)$ runtime where $s_\lambda$ is the $\lambda$-statistical dimension of $G$.

We let $S_P$ be the sketch of size $m \times (\sum_{l=0}^q d^l)$ which sketches the kernel $P$. The sketch $S_P$ is defined as

$$S_P = \frac{1}{\sqrt{0!}} \Pi^0 \oplus \frac{1}{\sqrt{1!}} \Pi^1 \oplus \frac{1}{\sqrt{2!}} \Pi^2  \cdots \frac{1}{\sqrt{q!}} \Pi^q.$$
Let $Z$ be the matrix of size $(\sum_{l=0}^q d^l) \times n$ whose $i^{\text{th}}$ column is

$$z_i = x_i^{\otimes 0} \oplus x_i^{\otimes 1} \oplus x_i^{\otimes 2} \cdots x_i^{\otimes q},$$
where $x_i$ is the $i^{\text{th}}$ column of $X$.
Therefore the following holds for $(S_P Z)^\top S_P Z$,
$$(S_P Z )^\top S_P Z  = \sum_{l=0}^{q} \frac{( \Pi^l X^{\otimes l}  )^\top \Pi^l X^{\otimes l} }{l!},$$
and hence,
$$(S_P Z D)^\top S_P Z D = \sum_{l=0}^{q} \frac{( \Pi^l X^{\otimes l} D )^\top \Pi^l X^{\otimes l} D}{l!}.$$
Therefore by combining the terms of \eqref{spectralguarantee-one-polynomial-term} for all $0\le l \le q$, using a union bound we get that with probability $1-\delta$, the following holds,
$$(1-\epsilon/2)(D P D + \lambda I_n) \preceq (S_P Z D)^\top S_P Z D + \lambda I_n \preceq (1+\epsilon/2)(D P D + \lambda I_n).$$
Now we define $S_g(x)$ which is a non-linear transformation on the input $x$ defined as 

$$S_g(x) = {e^{-\|x\|_2^2/2}} \left(\frac{1}{\sqrt{0!}} \cdot \Pi^0(x^{\otimes 0}) \oplus \frac{1}{\sqrt{1!}} \cdot \Pi^1(x^{\otimes 1}) \oplus \frac{1}{\sqrt{2!}} \cdot \Pi^2(x^{\otimes 2}) \cdots \frac{1}{\sqrt{q!}} \cdot \Pi^q(x^{\otimes q}) \right).$$

We have that $S_g (X) = S_P Z D$, therefore with probability $1-\delta$, the following holds,
$$(1-\epsilon)(G + \lambda I_n) \preceq (S_g (X))^\top S_g (X) + \lambda I_n \preceq (1+\epsilon)(G + \lambda I_n).$$
Note that the target dimension of $S_g$ is $m = m_0 + m_1 + \cdots + m_q \approx q^5 \log^3(nd/\delta) s_\lambda /\epsilon^2$. Also, by Theorem \ref{high-prob-sketch}, time to compute $S_g (X)$ is $O\left( \frac{n q^6}{\eps^2} \cdot \log^4(nd/\delta) \cdot s_\lambda + \frac{q^6}{\epsilon^2} \cdot \log^3(nd/\delta)\cdot \text{nnz}(X) \right)$.

\end{proof}

\subsection*{Acknowledgements}\label{sec:ack}
Michael Kapralov is supported by ERC Starting Grant SUBLINEAR.
Thomas D.~Ahle, Jakob B.~T.~Knudsen, and Rasmus Pagh are supported by Villum Foundation grant~16582 to Basic Algorithms Research Copenhagen (BARC).
David Woodruff is supported in part by Office of Naval Research (ONR) grant N00014-18-1-2562. 
Part of this work was done while Michael Kapralov, Rasmus Pagh, and David Woodruff were visiting the Simons Institute for the Theory of Computing.

\newcommand{\etalchar}[1]{$^{#1}$}

\appendix

\section{Direct Lower and Upper Bounds}\label{sec:proof:lower}
\renewcommand{\(}{\left(}
\renewcommand{\)}{\right)}

We introduce the following notation.
We say $f(x) \lesssim g(x)$ if for some some universal constant $C$ we have $f(x) \le C g(x)$ for all $x\in\R$ and .
Note this is slightly different from the usual $f(x)=O(g(x))$ in that it is uniform in $x$ rather than asymptotic.
We similarly say $f(x)\gtrsim g(x)$ if $g(x)\lesssim f(x)$ and $f(x)\sim g(x)$ if both $f(x)\lesssim g(x)$ and $f(x)\gtrsim g(x)$.

We will also make heavy use of the $L^p$ norm notation for random variables in $\\R$, that is
for $p\ge 1$ we write $\PNorm{X}p = (E |X|^p)^{1/p}$.
A very useful result for computing the $L^p$-norm of a sum of random variables is the following:
\begin{lem}[Latala's inequality, \cite{latala1997estimation}]\label{lem:latala-sup}
	If $p\ge 2$ and $X, X_1, \dots, X_n$ are iid. mean 0 random variables, then we have
	\begin{align}
	\PNorm{\sum_{i=1}^n X_i}p \sim \sup
	\left\{ \frac{p}s\left(\frac np\right)^{1/s}\PNorm{X}s \,\middle\vert \,  \max\left\{2,\frac pn\right\}\le s\le p\right\}.\label{eq:latala}
	\end{align}
\end{lem}

The following simple corollary will be used for both upper and lower bounds:
\begin{cor}\label{cor:latala-cor}
	Let $p\ge2, C>0$ and $\alpha\ge 1$.
	Let $(X_i)_{i\in[n]}$ be iid. mean 0 random variables such that $\PNorm{X_i}p\sim (C p)^\alpha$,
	then $\PNorm{\sum_i X_i}p \sim C^\alpha\max\{2^\alpha\sqrt{pn},\, (n/p)^{1/p}p^{\alpha}\}$.
\end{cor}
\begin{proof}
	We will show that the expression in \cref{eq:latala} is maximized either by minimizing or maximizing $s$.
	Hence we need to chat that $\frac{p}s\left(\frac np\right)^{1/s}s^\alpha$ it has no other optimums in the valid range.
	For this, we note that $\frac{d}{ds}\frac{p}{s}\(\frac np\)^{1/s}s^\alpha = \frac{-p}{s^{3-\alpha}}\(\frac np\)^{1/s}\left((1-\alpha)s+\log\frac np\right)$.
	Given $\alpha\ge1$ the derivative is non-decreasing in $s$, which gives the lemma.
\end{proof}

For the lower bound we will also use the following result by Hitczenko, which provides an improvement on Khintchine for Rademacher random variables.

\begin{lem}[Sharp bound on Rademacher sums~\cite{hitczenko1993domination}]\label{lem:hitczenko}
	Let $\sigma\in\{-1,1\}^n$ be a random Rademacher sequence and let $a\in\RR^n$ be an arbitrary real vector with sorted entries $|a_1| \ge |a_2| \ge \cdots |a_n|$, then
	\begin{align}
	\PNorm{\langle a,\sigma\rangle}p \sim \sum_{i\le p}a_i + \sqrt{p}\big(\sum_{i>p}a_i^2\big)^{1/2}
	\end{align}
\end{lem}

Finally the lower bound will use the Paley-Zygmund inequality (also known as the one-sided Chebyshev inequality):
\begin{lem}[Paley-Zygmund]\label{lem:pz}
	Let $X\ge 0$ be a real random variable with finite variance, and let $\theta\in[0,1]$, then
	\begin{align}
	\Prp{X \ge \theta \Ep{X}} \ge (1-\theta)^2\frac{\Ep{X}^2}{\Ep{X^2}}.
	\end{align}
\end{lem}
A classical strategy when using Paley-Zygmund is to prove $\Ep{X}\ge2\eps$ for some $\eps>0$, and then take $\theta=1/2$ to give
$\Prp{X \ge \eps} \ge \Ep{X}^2/(4\Ep{X^2})$.

\subsection{Lower Bound for Sub-Gaussians}\label{appendix:lowerbound-subgaussian}

The following lower bound considers the sketching matrix consisting of the direct composition of matrices with Rademacher entries.
Note however that the assumptions on Rademachers are only used to show that the $p$-norm of a single row with a vector is $\sim\sqrt{p}$.
For this reason the same lower bound hold if the Rademacher entries are substituted for, say Gaussians.

\begin{thm}[Lower bound]
	For some constants $C_1,C_2,B>0$, let $d,m,c\ge1$ be integers, let $\eps\in[0,1]$ and $\delta\in[0,1/16]$.
	Further assume that $d\ge\log1/\delta\ge c/B$.
	Then the following holds.
	
	Let $M^{(1)}, \dots, M^{(c)} \in \R^{m\times d}$ be matrices with all independent Rademacher entries
	and let $M=\tfrac1{\sqrt m} M^{(1)}\bullet\dots\bullet M^{(c)}$.
	Then there exists some unit vector $y\in\R^{d^c}$ such that if
	\begin{align}
	m < C_1 \max\left\{
	3^c \eps^{-2}\frac{\log1/\delta}c
	,\, \eps^{-1}\left(\frac{C_2\log1/\delta}c\right)^{c}\right\}
	\quad\text{then}\quad \Pr\left[\abs{\norm{My}_2^2-1}>\eps\right]>\delta.
	\end{align}
\end{thm}
\begin{proof}
	Let $y=[1,\dots,1]^T/\sqrt{d}\in\R^d$ and let $x=y^{\otimes c}$.
	We have
	\begin{align}
	\norm{My}_2^2-1
	= \frac1m\left\|M^{(1)}x \circ \dots \circ M^{(c)}x\right\|_2^2-1
	= \frac1m\sum_{j\in[m]} \big(\prod_{i\in[c]} Z_{i,j}^2\big)-1
	\end{align}
	where each $Z_{i,j}=\sum_{k\in[d]}M^{(i)}_{j,k}/\sqrt{d}$ are independent averages of $d$ independent Rademacher random variables.
	By~\cref{lem:hitczenko} we have $\PNorm{Z_{i,j}}p \sim \min\{\sqrt{p},\sqrt{d}\}$ which is $\sqrt{p}$ by the assumption $d\ge\log1/\delta$ as long as $p\le\log1/\delta$.
	By the expanding $Z_{i,j}^4$ into monomials and linearity of expectation we get $\|Z_{i,j}\|_4 = \frac1{\sqrt{d}}(d+3d(d-1))^{1/4} = (3-2/d)^{1/4}$.
	
	Now define $X_j = \prod_{i\in[c]}Z_{i,j}^2-1$, then $EX_j=0$
	and $\PNorm{X_j}p \ge \PNorm{\prod_{i\in[c]}Z_{i,j}^2}p-1 = \PNorm{Z_{i,j}}{2p}^{2c}-1 \ge K^c p^c$ for some $K$, assuming $p\ge 2$.
	In particular, $\PNorm{X_j}2\ge \PNorm{Z_{i,j}}{4}^{2c}-1 = (3-2/d)^{c/2}-1
	\sim 3^{c/2}$ by the assumption $d\ge c\ge1$.
	
	We have $\PNorm{\norm{My}_2^2-1}p = \frac1m\PNorm{\sum_{j\in[m]}X_m}p$ is a sum of iid. random variables,
	so we can use \cref{cor:latala-cor} to show
	\begin{align}
	K_3\max\left\{\sqrt{3^c p/m}, (m/p)^{1/p}K_1^c p^c/m\right\}
	&\lesssim \PNorm{\norm{My}_2^2-1}p
	\\&\lesssim K_4\max\left\{\sqrt{3^c p/m}, (m/p)^{1/p}K_2^c p^c/m\right\}
	\label{eq:upper-lower}
	\end{align}
	for some universal constants $K_1,K_2,K_3,K_4>0$.
	
	Assume now that
	$m < \max\left\{
	AK_3^23^c \eps^{-2}\frac{\log1/\delta}c,
	\frac{K_3}4\eps^{-1}\left(4AK_1\frac{\log1/\delta}c\right)^{c}
	\right\}$ as in the theorem.
	We take $p=4A\frac{\log1/\delta}{c}$ for some constant $A$ to be determined.
	We want to show $\PNorm{\norm{My}_2^2-1}p \ge 2\eps$.
	For this we split into two cases depending on which term of $m<\max\{(1),(2)\}$ dominates.
	If $(1) \ge (2)$ we pick the first lower bound in \cref{eq:upper-lower} and get
	$\PNorm{\norm{My}_2^2-1}p \ge K_3\sqrt{3^c p/m} \ge K_3 \sqrt{\frac{4\eps^2}{K_3^2}} = 2\eps$.
	Otherwise, if $(2)\ge(1)$, we pick the other lower bound and also get:
	\begin{align}
	\PNorm{\norm{My}_2^2-1}p \ge
	K_3 (m/p)^{1/p}\frac{K_1^c p^c}m
	\ge
	\frac{K_3}2 \frac{K_1^c \left(4A\frac{\log1/\delta}{c}\right)^c}{\frac{K_3}4\eps^{-1}\left(4AK_1\frac{\log1/\delta}c\right)^{c}}
	= 2\eps,
	\end{align}
	where we used $(m/p)^{1/p}\ge e^{-1/(em)}\ge1/2$ for $m\ge 1$.
	Plugging into Paley-Zygmund (\cref{lem:pz}) we have
	\begin{align}
	\Pr\left[\abs{\norm{My}_2^2-1}\ge\eps\right]
	&\ge \Pr\left[\abs{\norm{My}_2^2-1}^p\ge\PNorm{\norm{My}_2^2-1}p^p 2^{-p}\right]
	\\&\ge\frac14 \left(\frac{\PNorm{\norm{My}_2^2-1}p}{\PNorm{\norm{My}_2^2-1}{2p}}\right)^{2p},
	\label{eq:pz}
	\end{align}
	where we used that $p\ge 1$ so $(1-2^{-p})^2\ge1/4$.
	
	There are again two cases depending on which term of the upper bound in \cref{eq:upper-lower} dominates.
	If $\sqrt{3^cp/m} \ge(m/p)^{1/p}K_2^c p^c/m$ we have using the first lower bound that
	$\frac{\PNorm{\norm{My}_2^2-1}p}{\PNorm{\norm{My}_2^2-1}{2p}} \ge \frac{K_3}{\sqrt{2}K_4}$.
	For the alternative case, $(m/p)^{1/p}K_2^c p^c/m \ge \sqrt{3^cp/m}$, we have
	\begin{align}
	\frac{\PNorm{\norm{My}_2^2-1}p}{\PNorm{\norm{My}_2^2-1}{2p}}
	\ge \frac{K_3}{\sqrt2 K_4} \frac{(m/p)^{1/p}}{(m/2p)^{1/2p}}\left(\frac{K_1}{2 K_2}\right)^c
	\ge \frac{K_3}{2 K_4}\left(\frac{K_1}{2 K_2}\right)^c
	\end{align}
	where $\frac{(m/p)^{1/p}}{(m/2p)^{1/2p}}\ge e^{-1/(4em)}\ge 1/\sqrt2$ for $m\ge 1$.
	
	Comparing with \eqref{eq:pz} we see that it suffices to take $A \le \min\{\frac1{\log 2K_4/K_3}, \frac1{\log 2K_2/K_1}\}/32$.
	This choice also ensures that $1\le p\le \log1/\delta$ as we promised.
	Note that we may assume in \cref{eq:upper-lower} that $K_3\le K_4$ and $K_1\le K_2$.
	We then finally have
	\begin{align}
	\frac14\left(\frac{K_3}{\sqrt{2} K_4}\right)^{2p}
	\ge \frac14 \delta^{1/(4c)}
	\quad\text{and}\quad
	\frac14\left(\frac{K_3}{2K_4}\left(\frac{K_1}{2 K_2}\right)^c\right)^{2p}
	\ge \frac14 \delta^{1/(4c)+1/4},
	\end{align}
	which are both $\ge\delta$ for $c\ge1$ and $\delta<1/16$.
\end{proof}

\subsection{Upper bound for Sub-Gaussians}\label{section:upper-subgauss}

\begin{thm}[Upper bound]\label{thm:rademacher-upper}
	Let $\eps,\delta\in[0,1]$ and let $\gamma>0$, $1\le c \le \frac{\log1/\delta}{4\gamma}$ be some constants.
	Let $T\in\RR^{m\times d}$ be a matrix with iid. rows $T_1, \dots, T_m \in \R^d$ such that $\Ep{(T_1x)^2}=\|x\|_2^2$ and $\PNorm{T_1x}p \le \sqrt{a p} \|x\|_2$ for some $a>0$ and $p\ge 4$.
	Let $M = T^{(1)}\bullet\dots\bullet T^{(c)}$ where $T^{(1)}, \dots, T^{(c)}$ are independent copies of $T$.
	Then M has the JL-moment property, $\PNorm{\|Mx\|_2-\|x\|_2}p \le \eps\delta^{1/p}$,
	given
	\begin{align}
	m
	\gtrsim (4ae^\gamma)^{2c} \eps^{-2} \frac{\log1/\delta}{c\gamma}
	+ (4ae^\gamma)^{c} \eps^{-1} \left(\frac{\log1/\delta}{c\gamma}\right)^{c}
	.
	\end{align}
\end{thm}

\begin{rem}
	In the case of random Rademachers we set $a=\sqrt{3}/4$ to get
	$$
	m =O\(
	3^{c} \eps^{-2} \frac{\log1/\delta}{c\gamma}\,e^{2c\gamma}
	+ \eps^{-1} \left(\sqrt{3}\frac{\log1/\delta}{c\gamma}\right)^{c}e^{c\gamma}
	\).
	$$
	Note that depending on $\gamma$ this matches either of the terms of the lower bound.
	Setting $\gamma=\Theta(1/c)$ or $\gamma=\Theta(1)$ we have either
	$$
	m =O\(
	3^{c} \eps^{-2} \log1/\delta
	+ \eps^{-1} \left(\sqrt{3}\log1/\delta\right)^{c}
	\)
	\quad
	\text{or}
	\quad
	m =O\(
	{(3e^2)}^{c} \eps^{-2} \frac{\log1/\delta}{c}
	+ \eps^{-1} \left(\sqrt{3}e\frac{\log1/\delta}{c}\right)^{c}
	\).
	$$
	Finally, in the case of constant $c=O(1), \gamma=\Theta(1)$
	we simply get
	$$m = O\left(\eps^{-2} \log1/\delta + \eps^{-1} \left(\log1/\delta\right)^{c}\right).$$
\end{rem}

\begin{proof}[Proof of \cref{thm:rademacher-upper}]
	Without loss of generalization we may assume $\PNorm{x}2=1$.
	We notice that $\PNorm{\|Mx\|_2^2 - 1}p \le \PNorm{\tfrac1m\sum_i (M_i x)^2 - 1}p$
	is the mean of iid. random variables.
	Call these $Z_i = (M_ix)^2-1$.
	Then $EZ_i = 0$
   and $\PNorm{Z_i}p = \PNorm{(M_ix)^2-1}p \lesssim \PNorm{(M_ix)^2}p = \PNorm{M_ix}{2p}^2$ by the triangle inequality and the assumption that $p\ge 1$.
	Now by the assumption $\PNorm{T_1x}p \le \sqrt{a p} \|x\|_2 = \sqrt{ap}$, and by \cref{lem:gen-khinchine}, we get that
	$\PNorm{M_ix}{p} = \PNorm{T^{(1)}_i\otimes\dots\otimes T_i^{(c)}x}p \le (ap)^{c/2}$,
	and so $\PNorm{Z_i}p \le (2ap)^{c}$ for all $i\in[m]$.
	
	We now use \cref{cor:latala-cor} which implies
	\begin{align}
	\PNorm{\frac1m\sum_i Z_i}p
	\lesssim (4a)^{c}\sqrt{p/m} + m^{1/p} (2ap)^{c}/m \label{eq:proto-upper}
	\lesssim (4a)^{c}\sqrt{p/m} + (4ap)^{c}/m.
	\end{align}
	The second inequality comes from the following consideration:
	If the second term of \eqref{eq:proto-upper} dominates, then
	$(4a)^{c}\sqrt{p/m} \le m^{1/p} (2ap)^{c}/m$
	which implies $m^{1/p} \le (p/2)^{\frac{2c-1}{p-2}} \le 2^{c}$ for $p\ge 4$.
	
	All that remains is to decide on $p$.
	We take $p=\frac{\log1/\delta}{c\gamma}$ which is $\ge4$ by assumption, and $m=\max\{(4ae^\gamma)^{2c} p \eps^{-2}, (4ae^\gamma)^c p^c \eps^{-1}\}$.
	Then
	\begin{align}
	\PNorm{\frac1m\sum_i Z_i}p^p
	&\lesssim (4a)^{cp}\max\{\eps^{p}(4ae^\gamma)^{-cp}, \eps^{p}(4ae^\gamma)^{-cp}\}
	\\&= e^{-c\gamma p}\eps^p
	\\&= \delta\eps^p,
	\end{align}
	which is exactly the JL moment property.
\end{proof}

\subsection{Lower Bound for TensorSketch}\label{appendix:lowerbound-tensorsketch}

For every integer $d,q$, the {\sf TensorSketch} of degree $q$, $M: \RR^{d^q} \rightarrow \RR^m$ is defined as,

\begin{equation}\label{tensorsketch-q-def}
M(x^{\otimes q}) = \mathcal{F}^{-1} \left( (\mathcal{F} C_1 x) \circ (\mathcal{F} C_2 x) \circ \cdots (\mathcal{F} C_q x) \right),
\end{equation} 
for every $x \in \RR^d$ where $C_1, \cdots C_q \in \RR^{m \times d}$ are independent instances of {\sf CountSketch} and $\mathcal{F} \in \CC^{m \times m}$ is the Discrete Fourier Transform matrix with proper normalization which satisfies the convolution theorem, also note that, $\circ$ denotes entry-wise (Hadamard) product of vectors of the same size.

\begin{lem}
	For every integer $d,q$, let $M: \RR^{d^q} \rightarrow \RR^m$ be the {\sf TensorSketch} of degree $q \le d$, see \eqref{tensorsketch-q-def}. For the all ones vector $x = \{1\}^d$,
	$$\Var\left[ \|M x^{\otimes q} \|_2^2 \right] \ge \left( \frac{3^q}{2m^2} -1 \right)\|x^{\otimes q}\|_2^4.$$
\end{lem}
\begin{proof}
	Note that since $\mathcal{F}$ is normalized such that it satisfies the convolution theorem, $\mathcal{F}^{-1}$ is indeed a unitary matrix times $1/\sqrt{m}$, $\|M x^{\otimes q}\|_2^2 = \frac{1}{m}\| (\mathcal{F} C_1 x) \circ (\mathcal{F} C_2 x) \circ \cdots (\mathcal{F} C_q x) \|_2^2$. Consider the first entry of the vector $(\mathcal{F} C_1 x) \circ (\mathcal{F} C_2 x) \circ \cdots (\mathcal{F} C_q x)$. Because the first row of $\mathcal{F}$ is all ones $\{1\}^m$, the first element of the mentioned vector for the choice of $x = \{1\}^d$ is ${\prod_{i=1}^{q} \left( \sum_{j \in [d]} \sigma^i(j) \right)} = {\prod_{i=1}^{q} \left( \sum_{j \in [d]} \sigma^i(j) \right)}$, where $\sigma^i : [d] \rightarrow \{-1, +1\}$ are fully independent random hash functions used by the {\sf CountSketch} $C_i$ for all $i \in [q]$. Let us denote by $V$ the following positive random variable,
	
	$$V = {\prod_{i=1}^{q} \left( \sum_{j \in [d]} \sigma^i(j) \right)^2}.$$
	Note that $\|M x^{\otimes q}\|_2^2 \ge \frac{V}{m}$, hence $\Ep{\|M x^{\otimes q}\|_2^4} \ge \frac{\Ep{V^2}}{m^2}$. Also note that $\Ep{V^2} = \prod_{i=1}^{q} \Ep{\left( \sum_{j \in [d]} \sigma^i(j) \right)^{4} }$ because $\sigma^i$'s are independent. We can write

	$$\Ep{\left( \sum_{j \in [d]} \sigma^i(j) \right)^{4} } = 3d^2 - 2d = 3(1 - \frac{1}{6d}) \|x\|_2^4,$$
	hence if $d \ge q$,
	
	$$\Ep{V^2} \ge (1/2) \cdot 3^q \cdot \|x^{\otimes q}\|_2^4,$$
	Therefore $\Ep{\|M x^{\otimes q}\|_2^4} \ge \frac{\Ep{V^2}}{m^2} \ge \frac{3^q}{2m^2} \|x^{\otimes q}\|_2^2$. It is also true that $\Ep{\|M x^{\otimes q}\|_2^2} = \|x^{\otimes q}\|_2^2$ \cite{avron2014subspace}.
\end{proof}

\begin{lem}
	For every integer $d,q$ every $\eps>0$, every $0<\delta \le \frac{1}{2 \cdot 12^q}$, let $M: \RR^{d^q} \rightarrow \RR^m$ be the {\sf TensorSketch} of degree $q$, see \eqref{tensorsketch-q-def}. If $m <  {3^{q/2}}$ then for the all ones vector $x = \{1\}^d$ we have,
	
	$$\Pr\left[ | \|M x^{\otimes q} \|_2^2 - \|x^{\otimes q}\|_2^2 | > 1/2 \cdot \|x^{\otimes q}\|_2^2 \right] > \delta.$$
\end{lem}
\begin{proof}
	Note that since $\mathcal{F}$ is normalized such that it satisfies the convolution theorem, $\mathcal{F}^{-1}$ is indeed a unitary matrix times $1/\sqrt{m}$, $\|M x^{\otimes q}\|_2^2 = \frac{1}{m}\| (\mathcal{F} C_1 x) \circ (\mathcal{F} C_2 x) \circ \cdots (\mathcal{F} C_q x) \|_2^2$. Consider the first entry of the vector $(\mathcal{F} C_1 x) \circ (\mathcal{F} C_2 x) \circ \cdots (\mathcal{F} C_q x)$. Because the first row of $\mathcal{F}$ is all ones $\{1\}^m$, the first element of the mentioned vector for the choice of $x = \{1\}^d$ is ${\prod_{i=1}^{q} \left( \sum_{j \in [d]} \sigma^i(j) \right)} = {\prod_{i=1}^{q} \left( \sum_{j \in [d]} \sigma^i(j) \right)}$, where $\sigma^i : [d] \rightarrow \{-1, +1\}$ are fully independent random hash functions used by the {\sf CountSketch} $C_i$ for all $i \in [q]$. Let us denote by $V$ the following positive random variable,
	
	$$V = {\prod_{i=1}^{q} \left( \sum_{j \in [d]} \sigma^i(j) \right)^2}.$$
	Note that $\|M x^{\otimes q}\|_2^2 \ge \frac{V}{m}$. Note that $\Ep{V^t} = \prod_{i=1}^{q} \Ep{\left( \sum_{j \in [d]} \sigma^i(j) \right)^{2t} }$ for every $t$ because $\sigma^i$'s are independent. Note that for $t=2$ we have,
	
	$$\Ep{\left( \sum_{j \in [d]} \sigma^i(j) \right)^{4} } = 3d^2 - 2d \ge 3(1 - \frac{1}{6d}) \|x\|_2^4,$$
	hence if $d \ge q$,
	
	$$\Ep{V^2} \ge (3^q/2) \cdot \|x^{\otimes q}\|_2^4.$$
	Now consider $t=4$. By Khintchine's inequality, Lemma \ref{lem:khintchine}, we have,
	
	$$\Ep{\left( \sum_{j \in [d]} \sigma^i(j) \right)^{8} } \le 105 \cdot \|x\|_2^8,$$
	hence,
	
	$$\Ep{V^4} \le 105^q \cdot \|x^{\otimes q}\|_2^8.$$
	Therefore by Paley Zygmund we have the following,
	
	\begin{align*}
	\Pr\left[ \|M x^{\otimes q}\|_2^2 \ge \frac{3^{\frac{q}{2}}}{2m} \cdot \|x^{\otimes q}\|_2^2 \right] &\ge \Pr\left[ V \ge 3^{\frac{q}{2}}/2 \cdot \|x^{\otimes q}\|_2^2 \right]\\
	&= \Pr\left[ V^2 \ge 3^q/4 \cdot \|x^{\otimes q}\|_2^4 \right]\\
	&\ge \Pr\left[ V^2 \ge 1/4 \cdot \Ep{V^2} \right]\\
	&\ge 1/2 \cdot \frac{\Ep{V^2}^2}{\Ep{V^4}}\\
	& \ge \frac{9^q}{2 \cdot 105^q}\\
	&> \frac{1}{2 \cdot 12^q} \ge \delta.
	\end{align*}
	
\end{proof}

\end{document}